\documentclass[12pt]{article}
\usepackage{amsfonts,amssymb}
\usepackage{hyperref}
\usepackage{titlesec}
\usepackage{titletoc}
\usepackage{amsmath}
\usepackage{listings}
\usepackage{verbatim}
\usepackage{graphicx}
\usepackage{geometry}
\usepackage{mathrsfs}
\usepackage{mathtools}
\numberwithin{equation}{section}
\usepackage{amsthm}
\usepackage{xcolor, tikz}
\newtheorem{theorem}{Theorem}
\newtheorem{pro}{Proposition}
\newtheorem{con}{Condition}
\newtheorem{lem}{Lemma}
\newtheorem{rem}{Remark}

\usepackage[autostyle=false, style=english]{csquotes}
\MakeOuterQuote{"}
\usepackage{bm}

\DeclareMathOperator{\sign}{sgn}
\DeclareMathOperator{\ii}{\textbf{i}}
\DeclareMathOperator{\rank}{rank}
\DeclareMathOperator{\diag}{dg}
\DeclareMathOperator{\x}{\textbf{x}}
\DeclareMathOperator{\y}{\textbf{y}}
\DeclareMathOperator{\A}{\textbf{A}}

\DeclareMathOperator{\Vx}{\mathrm{V}_{\x}}
\DeclareMathOperator{\Ax}{\textbf{Ax}}
\DeclareMathOperator{\Ay}{\textbf{Ay}}

\usepackage{cite}
\title{Signal Reconstruction from Phase-only Measurements: Uniqueness Condition, Minimal Measurement Number and Beyond}
\author{Junren Chen\thanks{Department of 
		Mathematics, The University of Hong Kong.
		(E-mail: chenjr58@connect.hku.hk).
		The work of this author was supported by an HKPFS scholarship from the Hong Kong Research Grants Council (RGC).
		} \and Michael K. Ng\thanks{Department of 
		Mathematics, The University of Hong Kong. (E-mail: mng@maths.hku.hk). The work of this author was supported in part by the Hong Kong RGC GRF 12300519, 17201020, 17300021 C1013-21GF,
C7004-21GF and Jointly by NSFC-RGC N-HKU76921.	
		}}
\date{\today}
\begin{document}
	
	\maketitle
	
	
	
	
	\begin{abstract}
	This paper studies the phase-only reconstruction problem of recovering a complex-valued signal $\textbf{x}$ in $\mathbb{C}^d$ from the phase of $\textbf{Ax}$ where $\textbf{A}$ is a given measurement matrix in $\mathbb{C}^{m\times d}$. The reconstruction, if possible, should be up to a positive scaling factor. 
By using 	the rank of discriminant matrices, 
uniqueness conditions are derived to characterize 
whether the underlying signal can be uniquely reconstructed. 
We are also interested in the problem of minimal measurement number.  
We show that at least $2d$ but no more than $4d-2$ measurements are needed for the reconstruction of all $\textbf{x}\in\mathbb{C}^d$, whereas the minimal measurement number is exactly $2d-1$ if we pursue the recovery of almost all signals.
Moreover, when 
adapted to 
the phase-only reconstruction of $\textbf{x}\in\mathbb{R}^d$, our uniqueness conditions are 
more practical and general than existing ones. 
Finally, we show that our theoretical results can be straightforwardly extended to affine 
phase-only reconstruction where the phase of $\textbf{Ax}+\textbf{b}$ is observed 
for some $\textbf{b}\in\mathbb{C}^d$.  
	\end{abstract}
	
	\vspace{2mm}
	
	\vspace{2mm}

	\section{Introduction}\label{sec1}
 Signal reconstruction from only the phase or magnitude of Fourier transform was intensively studied in the 1980s 
\cite{hayes1980signal,hayes1982reconstruction,oppenheim1981importance,fienup1978reconstruction}. The reconstruction from phase-only observations is often referred  to as a phase-only  reconstruction  problem in the literature 
 (e.g., \cite{hayes1980signal,oppenheim1981importance}), while the reconstruction based on  measurement magnitude is commonly termed as a phase retrieval problem 
 (e.g., \cite{fienup1982phase}). 
 
 This paper concerns the phase-only reconstruction problem that has found many real-world applications, including blind deconvolution \cite{hayes1982reconstruction,stockham1975blind}, signal and image coding \cite{hayes1980signal,oppenheim1981importance}, kinoforms \cite{espy1983effects,oppenheim1981importance}, image alignment \cite{kughlin1975phase}, radiolocation \cite{li1983arrival}. For example, the recovery of a     signal blurred with an unknown distorting signal is called blind deconvolution, and it  reduces to phase-only reconstruction if the distorting signal has zero Fourier phase. This special case occurs  in   images  blurred by defocused lenses with circular aperture stops or signals under long-term exposure to atmospheric turbulence (e.g., see   \cite[section IV]{oppenheim1981importance} or   \cite[Section I]{hayes1982reconstruction}). Then, many subsequent works further explored and developed the applications, specifically in the cases of image processing like  restoration \cite{behar1992image,urieli1998optimal} and inpainting \cite{hua2007image}, object shape retrieval \cite{bartolini2005warp} and speech reconstruction   \cite{loveimi2010objective}. More recently,   phase-only measurement was applied in multiple-input and multiple-output (MIMO) because of its potential for high-bandwidth communication \cite{wang2016multiuser}. Also, researchers from the compressed sensing community began to study a phase-only   sensing scenario because phase-only measurement is robust to multiplicative corruption and enjoys easier quantization \cite{feuillen2020ell,boufounos2013sparse,jacques2021importance,chen2022uniform}.

From an algorithmic perspective, some algorithms have been proposed for solving the problem of phase-only reconstruction. Hayes  
et al. proposed an iterative algorithm and closed form solution in \cite{hayes1980signal}. The iterative algorithm alternatively imposes the  signal support and the Fourier phase as constraints, while the closed form solution is derived by solving a linear system.  
Later, Levi and Stark developed the Projection Onto Convex Sets (POCS) algorithm  that incorporates the Fourier phase in a different manner \cite{levi1983signal}. The performance of POCS algorithm in image restoration was   extensively investigated in \cite{urieli1998optimal}. We note that these algorithms   only apply to the reconstruction of a real-valued signal from its Fourier phase, for which the theoretical basis is the uniqueness condition established in \cite{hayes1980signal,hayes1982reconstruction}. 
More precisely, 
if a real-valued signal admits $z$-transform that does not have zero in reciprocal pair or on the unit circle, then it can be uniquely specified (up to a positive scaling factor\footnote{This is   the unavoidable trivial ambiguity in phase-only reconstruction. In this paper, terms such as ``exact reconstruction'', ``exactly recovered", ``uniquely specified" are used up to this trivial ambiguity.}) by the Fourier phase. It should also be noted that a   more practical  necessary and sufficient condition was obtained in \cite{ma1991novel}. More recently, the MagnitudeCut algorithm was proposed in \cite{wu2016phase}, and a quadratic programming algorithm was developed in \cite{kishore2020phasesense}. These two algorithms can be used in the reconstruction of a complex-valued signal from the phase  of general linear measurements. Nevertheless, the theoretical foundation of this generalized setting is far from solid. Specifically, it is unclear how to determine whether the signal can be uniquely recovered, and if so, then how many measurements are required.  

Despite the aforementioned applications and algorithms,   phase-only reconstruction has received far less attention than   phase retrieval in the past two decades. Note that, the theories for recovering a complex-valued signal from the magnitudes of general linear measurements\footnote{This is often referred to as a generalized phase retrieval problem (e.g., \cite{candes2015phase,sun2018geometric}) because the  measurement is not restricted to be Fourier magnitude. Accordingly, the problem in our work can be termed as   generalized phase-only reconstruction, but herein we simply refer to it  as phase-only reconstruction.} have been well established, especially from the perspective of the minimal measurement number, see   \cite{balan2006signal,bandeira2014saving,conca2015algebraic,huang2021almost} for instance. In particular,  these   references investigated 
at least 
how many measurements are sufficient for phase retrieval of either all signals \cite{balan2006signal,bandeira2014saving,conca2015algebraic} or almost all signals \cite{balan2006signal,huang2021almost}. 
This line of works motivates us to study the minimal measurement number for phase-only reconstruction.

The main aim of this paper is 
to provide a theoretical study for phase-only reconstruction that accommodates general measurement matrix and complex-valued signal, mainly from the viewpoint of minimal measurement number. Specifically, we study the reconstruction of $\textbf{x}\in\mathbb{C}^d$ from the phase of $\textbf{Ax}$, where $\textbf{A}\in \mathbb{C}^{m\times d}$ is the measurement matrix with measurement number $m$. A theoretical framework  equipped with a complete suite of notations is built, and the framework is then employed to study the minimal measurement number for reconstruction of all or almost all signals in $\mathbb{C}^d$. We note that the technicalities   in the proofs   essentially  depart from the algebraic argument in phase retrieval (e.g., \cite{conca2015algebraic}), and indeed, most analyses are based on linear algebra. Our main contributions are summarized as follows: 
\begin{itemize}
    \item We propose necessary and sufficient      uniqueness conditions based on  the rank of discriminant matrices (Theorems \ref{T1}-\ref{T2}). These results can be directly adjusted to phase-only reconstruction of a real-valued signal 
    (Theorems \ref{realcasedisd}-\ref{realdise}) and then recover all previously known uniqueness criteria presented in \cite{hayes1980signal,ma1991novel}, see Section \ref{compare}. 
    
    \item We prove that the minimal measurement number for the reconstruction of all signals in $\mathbb{C}^d$ is at least $2d$ but no more than $4d-2$, see Theorems \ref{T5}-\ref{T4}. We also show $2d-1$ is the minimal measurement number for recovering almost all signals. Specifically, phase of $2d-1$ generic linear measurements can specify a generic signal up to a positive scaling factor (Theorem \ref{T6}). 
\end{itemize}

Besides the uniqueness conditions and the results on minimal measurement number, we present 
some interesting properties   
for a phase-only system ($\sign(\textbf{Ax}) = \textbf{b}$) 
as a side contribution, 
see Theorem \ref{T8} and Remark \ref{remark4}. 
Moreover, we  note that our theories not only support the algorithms in \cite{wu2016phase,kishore2020phasesense},  but also shed light on the understanding of previous simulation results, see Remark \ref{remark3}. By using similar technical analyses, 
we carry over the theoretical framework to affine phase-only reconstruction, which is  inspired by some recent works on affine phase retrieval \cite{gao2018phase,gao2022newton,huang2021phase}.

This paper is structured as follows. Some preliminaries and notations are given in the 
remaining Section \ref{sec1}.	In Section \ref{sec2}, we propose uniqueness conditions based on    the rank of discriminant matrix that precisely characterize whether a signal can be uniquely specified. Using these uniqueness conditions as main tools,	  we study the minimal measurement number required for reconstruction of all signals or almost all signals in Sections \ref{sec3}--\ref{sec4}. Two other interesting results  are presented in Section \ref{sec5}. In Section \ref{compare}, we carefully compare our results with existing works to show our technical contributions explicitly. In Section \ref{sec6}, the whole theory is straightforwardly extended to affine phase-only reconstruction. Finally, some concluding remarks are given in Section \ref{sec7} to close the paper. Most proofs for phase-only reconstruction of real-valued signal (Section \ref{compare}) and affine phase-only reconstruction (Section \ref{sec6}) are provided in Appendices.

	\subsection{Preliminaries and notations}
\label{preli}
    Following the convention in previous works on minimal measurement number of phase retrieval, we will use the terminology   “generic”. 
     Here, we adopt the definition in  \cite[section 2.2]{conca2015algebraic} (see also a similar introduction  in \cite{bandeira2014saving}), and keep it as concise as possible. A subset of $\mathbb{R}^n$ is called a real algebraic variety if it is  defined to be the common zeros of finitely many polynomials in $  \mathbb{R}[x_1,...,x_n]$. Then, declaring all real algebraic varieties to be closed set defines the Zariski topology of $\mathbb{R}^n$. Evidently, a non-empty Zariski open set has full Lebesgue measure, and is open, dense under the standard Euclidean topology (recall that a set is said to be dense if its closure is the full space)\footnote{For readers unfamiliar with Zariski topology, it shall be fine to simply think of non-empty Zariski open set as an extremely large set whose complement is of zero Lebesgue measure, and nowhere dense under Euclidean topology (recall that a set is said to be nowhere dense if its closure has no interior).}. In this work, similar to \cite{conca2015algebraic,bandeira2014saving}, $\mathbb{C}^n$ is identified with $\mathbb{R}^{2n}$ when we talk about the Zariski topology of $\mathbb{C}^n$. That is, $V_0\subset \mathbb{C}^n$ is Zariski closed/open if $\{(\textbf{x},\textbf{y})\in\mathbb{R}^{2n}:\textbf{x}+\textbf{i}\textbf{y}\in V_0\}$ is Zariski closed/open in $\mathbb{R}^{2n}$. Likewise, suppose $V$ is a $s$-dimensional linear subspace (or linear submanifold) of $\mathbb{C}^n$, we talk about its Zariski topology by viewing it as $\mathbb{C}^s$\footnote{More precisely, $V_0\subset V$ is Zariski closed/open if for some invertible linear (or affine) transformation $\mathcal{A}(\cdot)$ between $V$ and $\mathbb{C}^s$, $\mathcal{A}(V_0) = \{\mathcal{A}(\textbf{x}):\textbf{x}\in V_0\}$ is Zariski closed/open in $\mathbb{C}^s$.}. Following \cite{conca2015algebraic}, we say that a generic point of $U$ has certain property if there is a non-empty Zariski open set of points having this property. For example, $\{\textbf{x}\in \mathbb{C}^n:f_j(\textbf{x})\neq 0,~\forall j = 1,2,\cdots,n_0\}$ is non-empty Zariski open 
    set of $\mathbb{C}^d$ if $f_j(\textbf{x})$ is non-zero polynomials of    $[\Re(\textbf{x}),\Im(\textbf{x})]$ with complex coefficients. This simple fact will be used in the proof of Lemma \ref{lemma4}.

    Next, we   give some general notations, with more introduced when appropriate in subsequent developments. One principle is that capital boldface letters, lowercase boldface letters, regular letters are used to denote matrices, vectors, scalars, respectively. $\bm{0}_{n_1\times n_2}$ (or $\bm{1}_{n_1\times n_2}$) represents the $n_1\times n_2$ matrix with all entries being $0$ (or $1$), but we simply write $\bm{0}$ (or $\bm{1}$) if the dimension is self-evident. We use $\textbf{i}$ to represent $\sqrt{-1}$, while $i$ still   serves frequently for other purposes. For   $a\in \mathbb{C}$, we let $\Re(a)$, $\Im(a)$, $|a|$ be its real part, imaginary part, absolute value, respectively.
    For non-zero $a$ we   define its phase to be $\sign(a)=\frac{a}{|a|}$, and we adopt the convention $\sign(0)=0$. In this paper, $\Re(\cdot)$, $\Im(\cdot)$, $|\cdot|$ and $\sign(\cdot)$ entry-wisely operate on vectors or matrices. Now, the phase-only reconstruction can be formulated to be the recovery of $\textbf{x}\in \mathbb{C}^d$ from $\sign(\textbf{Ax})$. We let $[n] = \{1,\cdots, n\}$. For $\textbf{v}\in\mathbb{C}^{n}$, we define $\mathrm{dg}(\textbf{v})$ to be the $n\times n$ diagonal matrix with main diagonal $\textbf{v}$. The Hadamard product between $\textbf{x} = [x_i]$ and $\textbf{y}= [y_i]$ is given by $\textbf{x}\odot \textbf{y} = \sum_i x_iy_i$. $\mathrm{N}(\textbf{v})$ is used to represent the support set, e.g., $\mathrm{N}(\textbf{v}) = \{1,3\}$ if $\textbf{v}= (1,0,2)^\top$. Let $|\mathcal{S}|$ be the number of elements in the finite set $\mathcal{S}$. For $\textbf{A}\in\mathbb{C}^{n_1\times n_2}$, $\mathcal{S}\subset [n_1]$, $\mathcal{T}\subset [n_2]$, $\textbf{A}^\mathcal{S}_\mathcal{T}$ denotes the submatrix constituted by rows in $\mathcal{S}$ and columns in $\mathcal{T}$. To keep notation light, we let $\textbf{A}^\mathcal{S}:=\textbf{A}^\mathcal{S}_{[n_2]}$, $\textbf{A}_\mathcal{T} : =\textbf{A}^{[n_1]}_{\mathcal{T}}$. In addition, we define  $\ker(\textbf{A}) = \{\textbf{x}\in \mathbb{C}^{n_2}:\textbf{Ax}=\bm{0}\}$.

Given $\textbf{A}\in\mathbb{C}^{n_1\times n_2}$, the mappings $\varphi$ and $\varphi_1$ are defined as  
	\begin{equation}
		\begin{aligned}
			\varphi(\textbf{A})=\begin{bmatrix}
					\Re(\A)& \Im(\A)\\
					-\Im(\A) & \Re(\A)
				\end{bmatrix};~\varphi_1(\textbf{A}) = \begin{bmatrix}
					\Re(\A) \\
					-\Im(\A) 
				\end{bmatrix}.
				\label{1.1}
		\end{aligned}
	\end{equation}
	 It is evident that $\textbf{A} = \textbf{0}$, $\varphi(\textbf{A})=\textbf{0}$ and $\varphi_1(\textbf{A})=\textbf{0}$ are equivalent. Besides, we note the following useful relations that can be easily verified: $\varphi_1(\textbf{A}+\textbf{B})=\varphi_1(\textbf{A})+\varphi_1(\textbf{B})$, $\varphi_1(\textbf{AB})=\varphi(\textbf{A})\varphi_1(\textbf{B})$, $\varphi(\textbf{AB})=\varphi(\textbf{A})\varphi(\textbf{B})$, $\rank(\varphi(\textbf{A}))=2 \cdot \rank(\textbf{A})$.

	\section{Discriminant matrices} \label{sec2}
	
	In this section, we introduce two discriminant matrices 	whose rank precisely characterizes whether the signal is uniquely specified by the phase-only measurements. These two matrices will be the key ingredients in our subsequent analysis.

	We consider the reconstruction of $\x$ from $\sign(\textbf{Ax})$. As has been noted, the trivial ambiguity of a positive scaling factor is unavoidable. Thus, given the measurement matrix $\A$,  the set of signals that can be uniquely recovered from phase is 
	\begin{equation}
	    \mathcal{W}_{\textbf{A}} =\{\textbf{x}\in\mathbb{C}^d:\sign(\textbf{Ay})=\sign(\textbf{Ax})\text{ implies }\textbf{x}=t\cdot \textbf{y}\text{ for some }t>0\}.
	\end{equation}
Note that $\rank(\A)<d$ implies $\dim(\ker(\A))\geq 1$, which trivially leads to $\mathcal{W}_{\bm{\mathrm{A}}}=\varnothing$. Thus, to study the problem of minimal measurement number that can uniquely reconstruct (almost) all $\textbf{x}\in \mathbb{C}^d$, we simply focus on $\A$ of full column rank.

Evidently, exchanging rows of $\A$ will not change $\mathcal{W}_{\A}$. Moreover, for any invertible $\textbf{P} \in \mathbb{C}^{d\times d}$, we have $\x \in \mathcal{W}_{\A}$ if and only if $\textbf{P}^{-1}\textbf{x}\in \mathcal{W}_{\textbf{AP}}$. Taking the \enquote{if} part for example, provided $\textbf{P}^{-1}\x \in \mathcal{W}_{\textbf{AP}}$, we note that $\sign(\Ax)=\sign(\Ay)$ equals $\sign(\textbf{APP}^{-1}\textbf{x})=\sign(\textbf{APP}^{-1}\textbf{y})$, so we have $\textbf{P}^{-1}\textbf{y}=t\cdot \textbf{P}^{-1}\textbf{x}$ for some $t>0$, which gives $\textbf{y}=t\cdot \textbf{x}$. It follows that $\x \in \mathcal{W}_{\A}$. The \enquote{only if} part can be similarly verified. Thus, we have $\mathcal{W}_{\A}=\textbf{P}\mathcal{W}_{\textbf{AP}}$, i.e., $\mathcal{W}_{\A}$ and $\mathcal{W}_{\textbf{AP}}$ only differ by an invertible linear transformation $\textbf{P}$.

Given any $\textbf{A}\in \mathbb{R}^{m\times d}$ of full column rank,  we now derive a corresponding measurement matrix in a special form. Specifically, 
we can assume $\textbf{A}^{[d]}$ is invertible by exchanging suitable rows, 
and note that 
$$\widetilde{\textbf{A}}=\textbf{A}(\A^{[d]})^{-1} =\begin{bmatrix}
\textbf{A}^{[d]}\\ \textbf{A}^{[n]\setminus[d]}
\end{bmatrix} (\A^{[d]})^{-1} = \begin{bmatrix}
    \bm{\mathrm{I}_d} \\ \bm{\mathrm{A_1}}
\end{bmatrix}.$$
Here $\widetilde{\textbf{A}}$   in a special form $[\bm{\mathrm{I}_d},\bm{\mathrm{A_1^\top}}]^\top$ has $\mathcal{W}_{\widetilde{\A}}$ that differs from $\mathcal{W}_{\A}$ by only an invertible linear transformation. This implies that $\mathcal{W}_{\widetilde{\A}}$ and $\mathcal{W}_{\A}$ are 
simultaneously 
the full space $\mathbb{C}^d$ or having full Lebesgue measure, so in many cases (e.g., when studying the problem of minimal measurement number) one can simply consider the special $\textbf{A}$ with $\textbf{A}^{[d]} = \bm{\mathrm{I}_d}$ without loss of any generality. For convenience,    $\A$ in the special  form $[\bm{\mathrm{I}_d},{\bm{\mathrm{B^\top}}}]^\top$ for some $\bm{\mathrm{B}}\in\mathbb{R}^{(m-d)\times d}$ is hereinafter referred to as 
a canonical measurement matrix.

\subsection{General measurement matrix}

We first characterize recoverable signal $\x$ under a general $\textbf{A}$ with no specific structure. For $\A$ of full column rank one always has $\bm{0} \in \mathcal{W}_{\A}$, so we simply focus on non-zero $\x$. Shortly, we will show that whether $\x \in \mathcal{W}_{\A}$ holds can be characterized by the rank of discriminant matrix $\mathcal{D}_{\textbf{A}}(\textbf{x})$  defined as  
\begin{equation}
	\begin{aligned}
	\label{disD}
			 \mathcal{D}_{\bm{\mathrm{A}}}(\x) =[\varphi(\A)\  \varphi_1(\mathrm{dg}(\Ax))]
			=\begin{bmatrix}
				\Re(\A)& \Im(\A)& \Re\big(\mathrm{dg}(\Ax)\big)\\
				-\Im(\A)& \Re(\A) & -\Im\big(\mathrm{dg}(\Ax)\big)
			\end{bmatrix}.
	\end{aligned}
\end{equation}
	For specific $\A$ and $\x$, we 
	define a linear system with real variables $\bm{\lambda} \in \mathbb{R}^{|\mathrm{N}(\Ax)|}$ and complex variables $\y\in \mathbb{C}^d$ as  
	\begin{equation}
		\Ay=[\mathrm{dg}(\sign(\Ax))]_{\mathrm{N}(\Ax)}\bm{\lambda}. 
		\label{eq1}
	\end{equation}
    This becomes a linear system of $\y$ when $\bm{\lambda}$ is specified, and one can easily verify that \begin{equation}
        \label{subspace}
       \mathrm{V}_{\x}: =\{\bm{\lambda}\in \mathbb{R}^{|\mathrm{N}(\Ax)|}: (\ref{eq1}) ~\text{is consistent (i.e., has a solution)}\}
    \end{equation} 
	is a linear subspace. The following Lemma characterizes $\x \in \mathcal{W}_{\A}$ via the dimension of $\mathrm{V}_{\x}$. We write $\mathrm{U}_{\x}: = \diag(\sign (\Ax)) $, and we assume $\A$ is fixed when we use the notations $\mathrm{V}_{\x},\mathrm{U}_{\x}$.  
	
	\begin{lem}
		Assume $\bm{\mathrm{A}}\in\mathbb{C}^{m\times d}$ has full column rank, $\bm{\mathrm{x}}\in \mathbb{C}^d$ is non-zero, then $\bm{\mathrm{x}} \in \mathcal{W}_{\bm{\mathrm{A}}}$ if and only if $\dim (\mathrm{V}_{\bm{\mathrm{x}}}) = 1$.
		\label{lemma1}
	\end{lem}
    	
    	\begin{proof}
	Note that  $(\mathrm{U}_{\x})_{\mathrm{N}(\Ax)}|\Ax|^{\mathrm{N}(\Ax)}=\mathrm{U}_{\x}|\Ax|=\Ax$, which gives $|\Ax|^{\mathrm{N}(\Ax)}\in \mathrm{V}_{\x}$. Since $|\Ax|^{\mathrm{N}(\Ax)}\neq 0$, we have $\dim(\mathrm{V}_{\x})\geq1$.

	We start from the "if" part. Assume  $\sign(\bm{\mathrm{A\tilde{x}}})=\sign(\Ax)$ for some $\bm{\mathrm{\tilde{x}}}$ (note that this implies $\mathrm{N}(\Ax)=\mathrm{N}(\bm{\mathrm{A\tilde{x}}})$), then the relation:
	 $$\bm{\mathrm{A\tilde{x}}}=\mathrm{dg}(\sign(\bm{\mathrm{A\tilde{x}}}))|\bm{\mathrm{A\tilde{x}}}|=(\mathrm{U}_{\x})_{\mathrm{N}(\Ax)}|\bm{\mathrm{A\tilde{x}}}|^{\mathrm{N}(\Ax)}$$ gives $|\bm{\mathrm{A\tilde{x}}}|^{\mathrm{N}(\Ax)}\in \mathrm{V}_{\x}$. Combining with $\dim(\mathrm{V}_{\x})=1$, we obtain $|\bm{\mathrm{A\tilde{x}}}|^{\mathrm{N}(\Ax)}=t\cdot |\Ax|^{\mathrm{N}(\Ax)}$ for some $t>0$. Note that $\mathrm{N}(\Ax)=\mathrm{N}(\bm{\mathrm{A\tilde{x}}})$, we obtain $|\bm{\mathrm{A\tilde{x}}}|=t\cdot|\Ax|$, and hence $\mathrm{U}_{\x}|\bm{\mathrm{A\tilde{x}}}|=t\cdot\mathrm{U}_{\x}|\Ax|$, i.e., $\bm{\mathrm{A\tilde{x}}} =t\cdot\Ax$. Now we use $\rank(\A)=d$ and obtain $\bm{\mathrm{\tilde{x}}}=t \cdot\x$. Hence $\x \in \mathcal{W}_{\bm{\mathrm{A}}}$ is concluded.

	For the "only if" part, because of $\dim(\mathrm{V}_{\x})\geq1$, it remains to show  that $\dim(\mathrm{V}_{\x})>1$   leads to a contradiction. 
	Suppose $\dim(\mathrm{V}_{\x})>1$. 
	Since $|\Ax|^{\mathrm{N}(\Ax)}\in \mathbb{R}_+^{|\mathrm{N}(\Ax)|} \cap \Vx$, we can choose $\bm{\lambda_1 }\in \mathrm{V}_{\x} \cap \mathbb{R}_+^{|\mathrm{N}(\Ax)|} $ such that there exists no $t>0$ such that $\bm{\lambda_1 }=t\cdot |\Ax|^{\mathrm{N}(\Ax)}$. By definition of $\Vx$, we have $\Ay_1=(\mathrm{U}_{\x})_{\mathrm{N}(\Ax)}\bm{\lambda_1 }$ for some $\bm{\mathrm{y_1}}$.
	This implies that $\sign(\bm{\mathrm{Ay_1}})=\sign(\Ax)$. Now we invoke the  condition $\x \in \mathcal{W}_{\bm{\mathrm{A}}}$, and obtain $\bm{\mathrm{y_1}}=t_1\cdot \x$ for some $t_1 > 0$. Thus, $\bm{\mathrm{Ay_1}}=t_1\cdot \Ax$, or equivalently  $(\mathrm{U}_{\x})_{\mathrm{N}(\Ax)}\bm{\lambda_1 }=t_1\cdot(\mathrm{U}_{\x})_{\mathrm{N}(\Ax)}|\Ax|^{\mathrm{N}(\Ax)}$. Since $(\mathrm{U}_{\x})_{\mathrm{N}(\Ax)}$ is of full column rank, we obtain $\bm{\lambda_1}=t_1\cdot |\Ax|^{\mathrm{N}(\Ax)}$, which is contradictory to our choice of $\bm{\lambda_1}$. 
\end{proof}
	
	The explicit calculation of  
	$\dim(\Vx)$ yields our first uniqueness condition.
	
	\begin{theorem}
		Assume $\bm{\mathrm{A}} \in \mathbb{C}^{m\times d}$ has full column rank, $\bm{\mathrm{x}}\in \mathbb{C}^d$ is non-zero, we have 
		$\bm{\mathrm{x}} \in \mathcal{W}_{\bm{\mathrm{A}}}$  if and only if $\rank( \mathcal{D}_{\bm{\mathrm{A}}}(\bm{\mathrm{x}}))=2d+|\mathrm{N}(\bm{\mathrm{Ax}})|-1$.
		\label{T1}
	\end{theorem}
	
	\begin{proof}
	By using the notations $\mathrm{U}_{\x}$, $\mathrm{V}_{\x}$,     (\ref{eq1}) becomes $\Ay=(\mathrm{U}_{\x})_{\mathrm{N}(\Ax)}\bm{\lambda }$.   We  deduce (\ref{eq1}) to a real linear system to explicitly calculate $\dim(\mathrm{V}_{\x})$. By using $\varphi(\cdot)$ and $\varphi_1(\cdot)$, we obtain \begin{equation}
	    \nonumber
	    \begin{aligned}
	       & \Ay=(\mathrm{U}_{\x})_{\mathrm{N}(\Ax)}\bm{\lambda}\iff\varphi_1(\Ay)=\varphi_1((\mathrm{U}_{\x})_{\mathrm{N}(\Ax)}\bm{\lambda})\\ &\iff\varphi(\A)\varphi_1(\y)-\varphi((\mathrm{U}_{\x})_{\mathrm{N}(\Ax)})\varphi_1(\bm{\lambda})=\bm{0}.
	    \end{aligned}
	\end{equation} 
	Since $\bm{\lambda} \in \mathbb{R}^{|\mathrm{N}(\Ax)|}$, we have $\varphi_1(\bm{\lambda})=[\bm{\lambda}^\top~\bm{0}^\top]^\top$, hence  
	  (\ref{eq1}) is equivalent to $\varphi(\A)\varphi_1(\y)-\varphi_1((\mathrm{U}_{\x})_{\mathrm{N}(\Ax)})\bm{\lambda}=\bm{0}$, which can be given in a matrix form
	\begin{equation}
		\begin{bmatrix}
			\varphi(\A)& -\varphi_1((\mathrm{U}_{\x})_{\mathrm{N}(\Ax)})
		\end{bmatrix}\begin{bmatrix}
			\varphi_1(\y)\\ \bm{\lambda}
		\end{bmatrix}=0.
		\label{eq2}
	\end{equation} 
	Note that $\rank(\varphi(\A))=2\cdot\rank(\A)=2d$, so $\varphi_1(\textbf{y})$ can be uniquely determined by a specific $\bm{\lambda}\in\Vx$. Therefore, the solution space of (\ref{eq2}), namely $\ker([\varphi(\A)\ -\varphi_1((\mathrm{U}_{\x})_{\mathrm{N}(\Ax)})])$, has the same dimension as $\mathrm{V}_{\x}$. This delivers
	\begin{equation}
		\begin{split}
			\begin{aligned}
				&\dim(\mathrm{V}_{\x}) =\dim(\ker([\varphi(\A)\ -\varphi_1((\mathrm{U}_{\x})_{\mathrm{N}(\Ax)})]))\\
				&=2d+|\mathrm{N}(\Ax)|-\rank([\varphi(\A)\ -\varphi_1((\mathrm{U}_{\x})_{\mathrm{N}(\Ax)})]).
				\label{eeq1}
			\end{aligned}
		\end{split}
	\end{equation}
	Comparing with (\ref{disD}), we confirm 
	$\rank( \mathcal{D}_{\bm{\mathrm{A}}}(\x))=\rank([\varphi(\A)~~ -\varphi_1((\mathrm{U}_{\x})_{\mathrm{N}(\Ax)})])$. 
	Then we conclude the proof by using Lemma \ref{lemma1}. 
	\end{proof}
	
    Although   $\mathcal{D}_{\textbf{A}}(\textbf{x})$ involves the linear measurement $\textbf{Ax}$ that is in general unknown, by noting $\rank(\mathcal{D}_{\textbf{A}}(\textbf{x}))= \rank([\varphi(\textbf{A}),-\varphi_1(\mathrm{U}_{\x})])$, and that the latter matrix can be constructed from $(\textbf{A},\sign(\textbf{Ax}))$, our uniqueness condition can be applied in practice.

	\subsection{Canonical measurement matrix}
	\label{canoin}
	
	In this subsection, we consider the canonical measurement matrix $\A$	that has the special form $[\bm{\mathrm{I}_d},{\bm{\mathrm{A_1^\top}}}]^\top$ for some $\bm{\mathrm{A_1}}\in\mathbb{R}^{(m-d)\times d}$. 
	 Let $\bm{\mathrm{\gamma_j^\top}}$ be the $j$-th row of $\A$. 
	Herein, we   exclusively use the entry-wise notation
	\begin{equation}
	    \label{entryax}
	    \textbf{A} = [r_{jk}\cdot e^{\textbf{i}\theta_{jk}}]_{{j\in [m] , k\in [d]}},~\x = [|x_k|\cdot e^{\textbf{i}\alpha_k}]_{k\in [d]}.
	\end{equation}
	For the zero entry of $\bm{\mathrm{A_1}}$ or $\x$, we simply let $\theta_{jk}=0$ or $\alpha_k=0$. Note that for $k\in [d]$ with $x_k\neq 0$, the current notation can already express $k$-th measurement as $e^{\ii \alpha_k}$, but we still need to introduce the notation of the $j$-th measurement when $j\geq d+1$. For this purpose, we define
	\begin{equation}
	    \label{jthme}
	    e^{\textbf{i}\delta_j}:=\sign(\bm{\mathrm{\gamma_j^\top}}\textbf{x}), \text{ when }j\in [m]\setminus[d],~\bm{\mathrm{\gamma_j^\top}}\textbf{x}\neq 0.
	\end{equation}
	Then, for $j\in[m]\setminus[d]$ such that $\bm{\mathrm{\gamma_j^\top}}\textbf{x}\neq 0$, 
	the measurement can  be written as   
	\begin{equation}
	    \label{2.6}
	    \sign\Big(\sum_{k=1}^d r_{jk}\cdot e^{\textbf{i}\theta_{jk}}\cdot |x_k| \cdot e^{\textbf{i}\alpha_k}\Big)= e^{\textbf{i}\delta_j} \iff  \sum_{k=1}^d r_{jk}\cdot e^{\textbf{i} (\theta_{jk}+\alpha_k-\delta_j)}|x_k| >0.
	\end{equation}
	By taking the real part, the right-hand side of (\ref{2.6}) implies a 
	linear equation $\sum_{k=1}^dr_{jk}\sin(\theta_{jk}+\alpha_k-\delta_j)|x_k|=0$. Therefore, for $j>d$ such that $\bm{\mathrm{\gamma_j^\top}}\textbf{x}\neq 0$, we define 
\begin{equation}
    \label{2.7}
    	\Psi_j(\x):=\begin{bmatrix}
		r_{j1}\sin(\theta_{j1}+\alpha_1-\delta_j) & r_{j2}\sin(\theta_{j2}+\alpha_2-\delta_j)&\cdots&r_{jd}\sin(\theta_{jd}+\alpha_d-\delta_j)
	\end{bmatrix},
\end{equation}
so that the linear equation can be written as $\Psi_j(\textbf{x})|\textbf{x}|=0$.
For $j\in [m]\setminus [d]$ such that $\bm{\mathrm{\gamma_j^\top}}\textbf{x} = 0$, because of $\sign(0)=0$, the linear equation $\bm{\mathrm{\gamma_j^\top}}\textbf{x}=0$ is provided. By using the element-wise notation, this can be formulated as $\sum_{k=1}^dr_{jk}e^{\textbf{i}\theta_{jk}}|x_k|e^{\textbf{i}\alpha_k}=0$ and evidently equals to two linear equations $\sum_{k=1}^dr_{jk}\sin(\theta_{jk}+\alpha_k)|x_k|=0$, $\sum_{k=1}^dr_{jk}\cos(\theta_{jk}+\alpha_k)|x_k|=0$.  Thus, for $j>d$ such that $\bm{\mathrm{\gamma_j^\top}}\textbf{x}=0$, we similarly define
\begin{equation}
    \label{2.8}
    	\Psi_j(\x):=\begin{bmatrix}
		r_{j1}\sin(\theta_{j1}+\alpha_1)&r_{j2}\sin(\theta_{j2}+\alpha_2)&\cdots &r_{jd}\sin(\theta_{jd}+\alpha_d)\\
		r_{j1}\cos(\theta_{j1}+\alpha_1)&r_{j2}\cos(\theta_{j2}+\alpha_2)&\cdots&r_{jd}\cos(\theta_{jd}+\alpha_d)
	\end{bmatrix}
\end{equation}
so that the two linear equations can be given in a compact form $\Psi_j(\textbf{x})|\textbf{x}|=0$. Note that we have defined $\Psi_j(\textbf{x})$ for all $j\in [m]\setminus [d]$, see (\ref{2.7}) if $\bm{\mathrm{\gamma_j^\top}}\textbf{x}\neq 0$, and (\ref{2.8}) otherwise.

 Now, we stack all   $\Psi_j(\x)$ for $j\in [m]\setminus [d]$ to define the discriminant matrix $\mathcal{E}_{\textbf{A}}(\x)$ as  
\begin{equation}
    \label{2.9}
    \mathcal{E}_{\textbf{A}}^0(\x)=\begin{bmatrix}
		\Psi_{d+1}(\x)\\
		\Psi_{d+2}(\x)\\
		\vdots\\
		\Psi_m(\x)
	\end{bmatrix} \quad 
	{\rm and} \quad
	\mathcal{E}_{\textbf{A}}(\x)=(\mathcal{E}_{\textbf{A}}^0(\x))_{\mathrm{N}(\x)}.
\end{equation}
By definition of $\Psi_j(\x)$, we always have 
\begin{equation}
    \label{2.10}
    \mathcal{E}_{\textbf{A}}^0(\x)|\textbf{x}| = \mathcal{E}_{\textbf{A}}(\x)|\textbf{x}|^{\mathrm{N}(\x)}=0.
\end{equation}
Under the canonical measurement matrix, we show in Theorem \ref{T2} that, $\rank(\mathcal{E}_{\textbf{A}}(\x))$ exactly characterizes the validity of $\textbf{x} \in \mathcal{W}_{\textbf{A}}$.
	
	\begin{theorem}
		Assume $\bm{\mathrm{A}}=[\bm{\mathrm{I}_d},{\bm{\mathrm{A_1^\top}}}]^{\bm{\top}}$ for some $\bm{\mathrm{A_1}}\in\mathbb{R}^{(m-d)\times d}$, $\bm{\mathrm{x}} \neq 0$, then $\bm{\mathrm{x}}
		\in \mathcal{W}_{\bm{\mathrm{A}}}$ if and only if $\rank(\mathcal{E}_{\bm{\mathrm{A}}}(\x))=|\mathrm{N}(\bm{\mathrm{x}})|-1$.
		\label{T2}
	\end{theorem}
	
\begin{proof}	We will use the    notations introduced in (\ref{entryax}) and (\ref{jthme}). For such canonical $\A$ and signal $\x$, the first $d$ measurement gives $\sign(\x)$. For $j\in \mathcal{J}_{\neq 0}:=\{j\in [m]\setminus [d]:\bm{\mathrm{\gamma_j^\top x}}\neq 0\}$, from (\ref{2.6}) one can see   
	$$\text{The $j$-th measurement } (\ref{2.6} )\iff \Psi_j(\x)|\textbf{x}|=0,~\sum_{k=1}^d r_{jk}\cos(\theta_{jk}+\alpha_k-\delta_j)|x_k|>0.$$
	When $j\in \mathcal{J}_{=0}(\x):= \{j\in [m]\setminus [d]:\bm{\mathrm{\gamma_j^\top x}}=0\}$, the $j$-th measurement is equivalent to $\Psi_j(\x)|\textbf{x}|=0$ with $\Psi_j(\textbf{x})$ defined in (\ref{2.8}).

	Therefore, for $\y = [y_k]\in \mathbb{C}^d$, we have a reformulation of  $	\sign(\Ay)=\sign(\Ax)$ 
	\begin{equation}
			\begin{aligned}
			\sign(\Ay)=\sign(\Ax) \iff	\begin{cases}
					&\sign(\y)=\sign(\x),\\
					&\mathcal{E}_{\textbf{A}}^0(\x)|\textbf{y}| = \bm{0}, \\
					&{\displaystyle \sum_{k\in [d]}  r_{jk}\cos(\theta_{jk}+\alpha_k-\delta_j)|y_k|>0}, ~ \forall j\in \mathcal{J}_{\neq0}(\x).
					\label{2.11}
				\end{cases}
			\end{aligned}
	\end{equation}
	The second equation in right-hand side of (\ref{2.11}) is due to (\ref{2.10}) and the observation $\mathcal{E}_{\textbf{A}}^0(\y)=\mathcal{E}_{\textbf{A}}^0(\x)$.
    Moreover, by noting that $\sign(\y)=\sign(\x)$   implies $\mathrm{N}(\y)=\mathrm{N}(\x)$, we can further restrict the summation to $k\in \mathrm{N}(\x)$, it follows that 
	\begin{equation}
			\begin{aligned}
			\sign(\Ay)=\sign(\Ax) \iff		\begin{cases}
					\label{2.12}
					&\sign(\y)=\sign(\x),\\
			 	&\mathcal{E}_{\textbf{A}}(\x)|\textbf{y}|^{\mathrm{N}(\x)}=0,\\
					& \displaystyle \sum_{k\in \mathrm{N}(\x)}  r_{jk}\cos(\theta_{jk}+\alpha_k-\delta_j)|y_k|>0, ~ \forall j\in \mathcal{J}_{\neq 0}(\x).
				\end{cases}
			\end{aligned}
	\end{equation}
	Specifically, letting $\y=\x$ in (\ref{2.12}) yields
	\begin{equation}
		\begin{split}
			\begin{aligned}
		       \mathcal{E}_{\textbf{A}}(\x)|\textbf{x}|^{\mathrm{N}(\x)}=\bm{0},&\\
					  \displaystyle \sum_{k\in \mathrm{N}(\x)}  r_{jk}\cos(\theta_{jk}+\alpha_k-\delta_j)|x_k|>0&, ~ \forall ~j\in \mathcal{J}_{\neq0}(\x). 
					\label{2.13}
			\end{aligned}
		\end{split}
	\end{equation}
	We now consider  the equivalence between $\bm{\mathrm{x}}
		\in \mathcal{W}_{\bm{\mathrm{A}}}$ and $\rank(\mathcal{E}_{\bm{\mathrm{A}}}(\x))=|\mathrm{N}(\bm{\mathrm{x}})|-1$.

	For the "if" part, we assume $\rank(\mathcal{E}_{\bm{\mathrm{A}}}(\x))=|\mathrm{N}(\bm{\mathrm{x}})|-1$. By (\ref{2.12}), $\sign(\Ay)=\sign(\Ax)$ implies
	$\mathcal{E}_{\textbf{A}}(\x)|\textbf{y} |^{\mathrm{N}(\x)}=\bm{0}$. Combined with $\rank(\mathcal{E}_{\textbf{A}}(\x))=|\mathrm{N}(\x)|-1$ and the first equation in (\ref{2.13}), it gives $|\textbf{y}|^{\mathrm{N}(\x)} =t\cdot  |\textbf{x}|^{\mathrm{N}(\x)}$ for some $t>0$, and hence $|\textbf{y}|= t\cdot |\textbf{x}|$. We further use
	$\sign(\y)=\sign(\x)$ and write $\sign(\y)\odot|\textbf{y}|=t\cdot \sign(\x)\odot|\textbf{x}|$. This delivers $\y = t\cdot \x$, so $\x\in\mathcal{W}_{\textbf{A}}$ follows.

	For the "only if" part, we assume $\x\in\mathcal{W}_{\textbf{A}}$. By the first equation in (\ref{2.13}) and $|\textbf{x}|^{\mathrm{N}(\x)}\neq \bm{0}$ we obtain
	$\rank(\mathcal{E}_{\textbf{A}}(\x))\leq |\mathrm{N}(\x)|-1$.
	Thus, we only need to assume $\rank(\mathcal{E}_{\textbf{A}}(\x))<|\mathrm{N}(\x)|-1$ and show that this   leads to a contradiction. 
	By noting $|\textbf{x}|^{\mathrm{N}(\x)} \in \mathbb{R}_+^{|\mathrm{N}(\x)|}\cap \ker(\mathcal{E}_{\textbf{A}}(\x))$ and 
	$$\dim(\ker(\mathcal{E}_{\textbf{A}}(\x))) = |\mathrm{N}(\x)| - \rank(\mathcal{E}_{\textbf{A}}(\x))\geq 2,$$ we can find
	$\bm{\lambda} \in \mathbb{R}_+^{|\mathrm{N}(\x)|}\cap \ker(\mathcal{E}_{\textbf{A}}(\x))$ such that $\bm{\lambda}$ and $|\textbf{x}|^{\mathrm{N}(\x)}$ are sufficiently close but linearly independent. Note that we can uniquely construct a signal $\y \in\mathbb{C}^d$ such that $\sign(\y)=\sign(\x)$ and $|\textbf{y}|^{\mathrm{N}(\x)}=\bm{\lambda}$. Thus, the construction of $\y$ gives $ \bm{0} =\mathcal{E}_{\textbf{A}}(\x)\bm{\lambda} = \mathcal{E}_{\textbf{A}}(\x) |\textbf{y}|^{\mathrm{N}(\x)}$. Now, the first two equations in (\ref{2.12}) are displayed.
Moreover, sufficiently close $\bm{\lambda}$ and $|\textbf{x}|^{\mathrm{N}(x)}$ can guarantee sufficiently close $\y$ and $\x$. Thus, recalling the second equation in (\ref{2.13}), the third equation in (\ref{2.12}) can be guaranteed.
  Hence, (\ref{2.12}) gives $\sign(\Ay)=\sign(\Ax)$, and combined with $\x\in\mathcal{W}_{\textbf{A}}$ we obtain $\x = t_1\cdot   \y$ for some $t_1>0$. To complete the proof, we take absolute value and then restrict the vectors to $\mathrm{N}(\x)$, it gives $|\textbf{x}|^{\mathrm{N}(\x)} = t_1\cdot |\textbf{y}|^{\mathrm{N}(\x)}=t_1\cdot \bm{\lambda} $, which is contradictory to linearly independence between $\bm{\lambda}$ and $|\textbf{x}|^{\mathrm{N}(\x)}$.
	\end{proof}

\begin{rem}
	With two discriminant matrices in place, we briefly comment on how to choose a suitable one in application. Generally speaking,	the  $2m\times (2d+m)$ matrix $\mathcal{D}_{\bm{\mathrm{A}}}(\bm{\mathrm{x}})$  has relatively simple polynomials entries, so it is  more amenable to the analyses involving properties of polynomial, e.g., those needed in Theorem \ref{T6}. 
	Although the entries of $\mathcal{E}_{\bm{\mathrm{A}}}(\bm{\mathrm{x}})$ are slightly more complicated, one may select some $\bm{\mathrm{x}}$ such that $\mathcal{E}_{\bm{\mathrm{A}}}(\bm{\mathrm{x}})$ has specific structure that may be conducive to the proof, see Theorem \ref{T4} for instance.  
\label{remark1}
\end{rem}

	\section{Reconstruction of all signals}\label{sec3}
    
    In this section, we 
    study the minimal measurement number $m$ required for $\A \in \mathbb{C}^{m\times d}$ to recover all $\x \in\mathbb{C}^d$, or equivalently  $\mathcal{W}_{\bm{\mathrm{A}}}=\mathbb{C}^d$. We mimic the term in \cite{huang2021almost} and say $\A$ is {\it magnitude retrievable} if $\mathcal{W}_{\bm{\mathrm{A}}}=\mathbb{C}^d$. The minimal measurement number of interest can be precisely formulated as 
    \begin{equation}
        \label{3.1}
        \bm{\mathrm{m_{all}}}(d) = \big\{m\in 
        \mathbb{N}_+: \text{some }\textbf{A}\in\mathbb{C}^{m\times d} \text{ is magnitude retrievable} \big\}.
    \end{equation}
    To the best of our knowledge, $\bm{\mathrm{m_{all}}}(d)$ has not yet been explored   previously, and   it is even not clear 
    whether $\bm{\mathrm{m_{all}}}(d)$ is finite. Indeed, most existing theoretical results are restricted to the Fourier measurement matrix whose rows read as $\textbf{f}(\omega)=[e^{-\textbf{i}\omega},e^{-\textbf{i}(2\omega)},\cdots,e^{-\textbf{i}(d\omega)}]$ for some frequency $\omega$. Note that  the convolution theorem gives $\textbf{f}(\omega) (\textbf{x}\ast \textbf{h}) =\big(\textbf{f}(\omega)\textbf{x}\big)\cdot \big(\textbf{f}(\omega)\textbf{h}\big)$, so  $\textbf{x}\ast \textbf{h}$ and $\textbf{x}$ cannot be distinguished from the Fourier phase  if $\textbf{f}(\omega)\textbf{h}>0$ for all $\omega$\footnote{One may consider symmetric $\textbf{h}$ to see this is possible.}. Thus, the Fourier measurement matrix is not magnitude retrievable. A more concrete example is the failure of recovering a symmetric signal from the Fourier phase, which has been noted in previous works (e.g., \cite{levi1983signal,urieli1998optimal}) and will be rigorously presented in Proposition \ref{pro2} of this work.

    To study $\bm{\mathrm{m_{all}}}(d)$, we first give a proposition to show that magnitude retrievable becomes possible under a  general measurement matrix, thus confirming 
    $\bm{\mathrm{m_{all}}}(d)<\infty$.

    \begin{pro}
    \label{pro1}
     For  all signals $\bm{\mathrm{x}}=[x_k]\in\mathbb{C}^d$, the $d + d\cdot (d-1)=d^2$ measurements given by $\{\sign(x_k):k\in [d]\}$, $\{\sign(x_k+ x_l),\sign(x_k+\bm{\mathrm{i}} x_l): 1\leq k<l\leq d\}$ can reconstruct $\bm{\mathrm{x}}$ up to a positive scaling factor. Thus, for each positive integer $d$, some $\bm{\mathrm{A}} \in \mathbb{C}^{d^2\times d}$ is magnitude retrievable, and $\bm{\mathrm{m_{all}}}(d) \leq d^2$.   
    \end{pro}
    
   \begin{proof} From  measurements  $\sign(\x)$ one knows     $\mathrm{N}(\x)$ and $ 
	\{\sign(x_k):k\in \mathrm{N}(\x)\}
	$. To reconstruct $\x$ up to a positive scaling factor, we only need to recover the magnitude ratio $ \{\frac{|x_k|}{|x_l|}:k,l\in \mathrm{N}(\textbf{x})\}$. Hence, it remains to show for fixed $k,l\in \mathrm{N}(\textbf{x})$, $k\neq l$, the measurements $(\sign(x_k),  \sign(x_l),  \sign(x_k+x_l),  \sign(x_k+\bm{\mathrm{i}}x_l))$ suffice to recover $\frac{|x_k|}{|x_l|}$. We discuss two cases.
	
	\noindent {\it Case 1.} When $\sign(x_k)\neq \pm \sign(x_l)$, then $x_k+x_l \neq 0$, hence we have 
	$$
	\frac{x_k+x_l}{\sign(x_k+x_l)}>0 \iff 
	\frac{|x_k|\sign(x_k)+|x_l|\sign(x_l)}{\sign(x_k+x_l)}>0.
	$$
	We take the imaginary part and obtain  
	\begin{equation}
		a\cdot \frac{|x_k|}{|x_l|}+b=0, \mathrm{where~} a:=\Im \left(\frac{\sign(x_k)}{\sign(x_k+x_l)} \right),~b:=\Im
		\left (\frac{\sign(x_l)}{\sign(x_k+x_l)} \right ).
		\label{E3}
	\end{equation} 
	If $a = 0$, then $\frac{\sign(x_k)}{\sign(x_k+x_l)}\in \mathbb{R}$, hence $\frac{|x_k|\sign(x_k)}{|x_k+x_l|\sign(x_k+x_l)}=\frac{x_k}{x_k+x_l}\in \mathbb{R}$. It is not difficult to see that this   leads to  $\frac{x_k}{x_l}\in\mathbb{R}$,
	which is contradictory to the assumption $\sign(x_k)\neq \pm \sign(x_l)$. 
	Thus, based on the measurements $\sign(x_k),\sign(x_l),\sign(x_k+x_l)$ one can obtain $\frac{|x_k|}{|x_l|}$ by solving (\ref{E3}).
	
	\vspace{2mm}
	\noindent {\it Case 2.} When $\sign(x_k)=\pm \sign(x_l)$, we have $\sign(x_k)\neq \pm \sign(\textbf{i}x_l)$. 
	By using the same arguments for {\it Case 1}, one can recover $\frac{|x_k|}{|\textbf{i}x_l|}$, i.e. $\frac{|x_k|}{|x_l|}$, from the measurements $\sign(x_k),\sign(x_l),\sign(x_k+\textbf{i}x_l)$.
	 Therefore, from the $d^2$ measurements $\{\sign(x_k):k\in [d]\}$ and $ \{\frac{|x_k|}{|x_l|}:k,l\in \mathrm{N}(\textbf{x})\}$, any $\textbf{x}\in\mathbb{C}^d$ can be recovered up to a positive scaling factor. Thus the proof is concluded.	
	\end{proof}
	
\begin{rem}
\label{remarknew}
    If we do not pursue the reconstruction of all $d$-dimensional complex-valued signals but  only  the recovery of a   fixed  $\bm{\mathrm{x}}$, the above construction suggests   adaptively using $d+|\mathrm{N}(\bm{\mathrm{x}})|-1$ measurements. More precisely, we can first measure $\sign(\bm{\mathrm{x}})$ that indicates $|\mathrm{N}(\bm{\mathrm{x}})|$. Furthermore, taking a specific $k_0 \in \mathrm{N}(\bm{\mathrm{x}})$,  the additional $|\mathrm{N}(\bm{\mathrm{x}})|-1$ measurements  $\{\sign(x_{k_0}+c_lx_l):l\in \mathrm{N}(\bm{\mathrm{x}}),l\neq k_0\}$   can deliver the magnitude ratio $\{\frac{|x_{k_0}|}{|x_l|}:l\in  \mathrm{N}(\bm{\mathrm{x}})\}$, hence $\bm{\mathrm{x}}$ is specified up to a positive scaling. Note that in the latter $|\mathrm{N}(\bm{\mathrm{x}})|-1$ measurements, based on $\sign(\textbf{x})$, $c_l$ should be selected adaptively  to guarantee $|c_l|=1$ and $c_l\neq \pm\frac{\sign(x_{k_0})}{\sign(x_l)}$.
\end{rem}


Next, we aim to lower the upper bound of $\bm{\mathrm{m_{all}}}(d)$. Recall that in the proof of Proposition \ref{pro1}  we   propose a concrete $\textbf{A}$ and then show its magnitude retrievable property. One shall see that we only use
elementary arguments. This is due to the simplicity of $\textbf{A}$, specifically its rows have two non-zero entries at most. However, it can be shown that for   $\textbf{A}$ with such simple rows, at least ${d(d-1)} $ measurements are required to deliver magnitude retrievable property\footnote{For any specific $1\leq j<k\leq d$, we need at least two rows with $j$-th and $k$-th entries being non-zero.}. Hence, to essentially reduce the current sample complexity $O(d^2)$, it is necessary to consider more complicated measurement matrix $\textbf{A}$.

We first present the new upper bound $4d-2$ as the following theorem.  

	\begin{theorem}
	 If $m\geq 4d-2$, then almost all $\bm{\mathrm{A}}$ in $\mathbb{C}^{m\times d}$ are magnitude retrievable, or equivalently, satisfy $\mathcal{W}_{\bm{\mathrm{A}}} = \mathbb{C}^d$. Specifically, when $d\geq 4$ we have the upper bound for the minimal measurement number $\bm{\mathrm{m_{all}}}(d) \leq 4d   -2$.
		\label{T5}
	\end{theorem}
 Note that the upper bound of $\bm{\mathrm{m_{all}}}(d)$ is presented   for $d\geq 4$ since the bound $\bm{\mathrm{m_{all}}}(d) \leq d^2$ in Proposition \ref{pro1} is tighter  when $d \in [3]$.

  To prove Theorem \ref{T5}, we will concentrate on the canonical measurement matrix    $\textbf{A} = [\bm{\mathrm{I}_d},\bm{\mathrm{B^\top}}]^{\bm{\top}}$ for some $\textbf{B}\in\mathbb{C}^{(m-d)\times d}$, and we denote the $(j,k)$-th entry of \textbf{A}   by $b_{jk}$ (rather than adopting (\ref{entryax})). The proof strategy is to first identify   \textbf{A} that is not magnitude retrievable   with the image of some smooth mappings, as done in Lemma \ref{lemmanew}. Then,   Sard's Theorem \cite{sard1942measure} delivers that the set of these undesired \textbf{A} is of zero Lebesgue measure (in $\mathbb{C}^{m\times d}$) when $m\geq 4d-2$. 
  
  \begin{lem}
      \label{lemmanew}
      Assume $m> d \geq 2$, $\bm{\mathrm{A}} = [\bm{\mathrm{I}_d},\bm{\mathrm{B^\top}}]^{\bm{\top}}$ for some $\bm{\mathrm{B}}\in\mathbb{C}^{(m-d)\times d}$. For any $\mathcal{J}\subset \{d+1,\cdots,m\}$, we write $\mathcal{J}_1 = \{j\in [m]\setminus [d]: j\notin \mathcal{J} \}$ and   define $(m-d)(d-1)$ positions of $([m]\setminus[ d])\times [d]$ as 
       \begin{equation}
           \label{726add}
           \mathcal{T}_{\mathcal{J}}=\{(j,k)\in ([m]\setminus[ d])\times [d]: k\neq 1\ \mathrm{when}\ j\in \mathcal{J},\ k\neq 2\ \mathrm{when}\ j\in \mathcal{J}_1\}.
       \end{equation}
       Moreover, we define $f_{\mathcal{J}}$ that maps the domain \begin{equation}
           \label{726add2}
           \Omega_{\mathcal{J}} =\mathbb{C}^{(m-d)(d-1)}\times (\mathbb{C}\setminus \{0\})\times \mathbb{C}^{d-2}\times (\mathbb{R}_+ \setminus \{1\}) \times \mathbb{R}_+^{d-2}\times (\mathbb{R}_+\setminus \{1\})^{|\mathcal{J}|}
       \end{equation} to $\mathbb{C}^{(m-d)d}$. Specifically, for the element 
       \begin{equation}
       \label{726add4}
           \bm{{\beta}}:=\big([\beta_{jk}]_{(j,k)\in \mathcal{T}_{\mathcal{J}}},\beta_2,[\beta_k]_{3\leq k\leq d}, \lambda_2, [\lambda_k]_{3\leq k\leq d}, [\lambda_j]_{j\in \mathcal{T}_{\mathcal{J}}} \big) \in \Omega_{\mathcal{J}},
       \end{equation}
       we define $\bm{\zeta}:=f_{\mathcal{J}}(\bm{{\beta}})=[\zeta_{jk}]_{j\in [m]\setminus[d],k\in [d]}$ entry-wisely by
       \begin{equation}
           \label{726add3}
           \begin{aligned}
               \zeta_{jk} = \beta_{jk}\text{ if }(j,k)\in \mathcal{T}_{\mathcal{J}},~\zeta_{j1} =\sum_{k=2}^d \frac{   \beta_k (1-\lambda_k\lambda_j )\beta_{jk}}{\lambda_j-1}\text{ if }j\in\mathcal{J},&\\\zeta_{j2} =\sum_{k=3}^d\frac{\beta_k(1-\lambda_k)\beta_{jk}}{(\lambda_2-1)\beta_2} \text{ if }j\in\mathcal{J}_1.&
           \end{aligned}
       \end{equation}
         For $1\leq p\leq q \leq d$ we let $\mathscr{T}_{p,q}$ be matrix obtained by exchanging the $p$-th row and $q$-th row of $\bm{\mathrm{I}_d}$. If we denote the set of the $\bm{\mathrm{B}}$ such that $\bm{\mathrm{A}} = [\bm{\mathrm{I}_d},\bm{\mathrm{B^\top}}]^{\bm{\top}}$ is not magnitude retrievable by $\mathscr{X}$, then it holds that 
       \begin{equation}
           \label{3.3}
           \mathscr{X}=  \bigcup_{1\leq p<q\leq d}\bigcup_{\mathcal{J}= [m]\setminus [d]} f_{\mathcal{J}}(\Omega_{\mathcal{J}}) \mathscr{T}_{2,q}\mathscr{T}_{1,p}.
       \end{equation}
  \end{lem}
  
 \begin{proof}
  We start by rephrasing an element $\textbf{B}$ in $\mathscr{X}$. Evidently, $\textbf{B}\in\mathscr{X}$ if and only if there exist $\x,\y \in\mathbb{C}^d$, $\sign(\x)=\sign(\y)$, $\sign(\textbf{Bx})=\sign(\textbf{By})$, but one cannot find $t>0$ such that $\x = t\cdot \y$. It is not hard to see that this is equivalent to 
  \begin{equation}
      \label{3.4}
  \exists \x,\y\in \mathbb{C}^d,~\text{such that }    \begin{cases}
       \sign(\x)=\sign(\y),~\sign(\textbf{Bx})=
			\sign(\textbf{By}), {\rm \ and} \\
       \exists~ p,q\in \mathrm{N}(\x),~p<q,~\mathrm{such ~that}~\frac{x_q}{x_p}\neq \frac{y_q}{y_p}.
      \end{cases}
  \end{equation}
  We define $\textbf{B}^{p,q} :=\textbf{B} \mathscr{T}_{1,p}\mathscr{T}_{2,q}$, then (\ref{3.4}) can be further equivalently reformulated as 
  \begin{equation}
      \label{3.5}
      \exists \x,\y\in \mathbb{C}^d,~1\leq p<q\leq d,~\text{such that }    \begin{cases}
       \sign(\x)=\sign(\y),\\
			\sign({\textbf{B}^{p,q}\textbf{x}})=\sign(\textbf{B}^{p,q}\textbf{y}), \\
   x_1=  y_1 = 1,~x_2 \neq y_2.
      \end{cases}
  \end{equation}
  To see the equivalence between (\ref{3.4}) and (\ref{3.5}), we   assume $\x,\y\in\mathbb{C}^d$ satisfy (\ref{3.4}) and consider $\bm{\mathrm{\hat{x}}}=[\hat{x}_i] := \mathscr{T}_{2,q}\mathscr{T}_{1,p}\textbf{x}/x_p$,     $\bm{\mathrm{\hat{y}}}=[\hat{y}_i] := \mathscr{T}_{2,q}\mathscr{T}_{1,p}\textbf{y}/y_p$. Then by some algebraic operations, we obtain
  \begin{equation}
      \begin{aligned}
           \sign(\bm{\mathrm{\hat{x}}}) = \mathscr{T}_{2,q}\mathscr{T}_{1,p} \sign(\textbf{x}) /\sign(x_p) =\mathscr{T}_{2,q}\mathscr{T}_{1,p} \sign(\textbf{y}) /\sign(y_p) = \sign(\bm{\mathrm{\hat{y}}}), &\\
            \sign(\textbf{B}^{p,q}\bm{\mathrm{\hat{x}}})=\sign\Big(\frac{\textbf{Bx}}{x_p}\Big) = \frac{\sign(\textbf{Bx})}{\sign(x_p)} = \frac{\sign(\textbf{By})}{\sign(y_p)}= \sign\Big(\frac{\textbf{By}}{y_p}\Big)= \sign(\textbf{B}^{p,q}\bm{\mathrm{\hat{y}}}), &\\
            \hat{x}_1 = \frac{x_p}{x_p } = 1,~\hat{y}_1 = \frac{y_p}{y_p}= 1,~\hat{x}_2 = \frac{x_q}{x_p}\neq \frac{y_q}{y_p} = \hat{y}_2.&
      \end{aligned}
  \end{equation}
Similary, we show that 
(\ref{3.5}) can lead to (\ref{3.4}). Now we consider  fixed $p,q$, $1\leq p<q\leq d$, and entry-wisely denote $\textbf{B}^{p,q} = [b_{jk}]_{j\in [m]\setminus[d],k\in [d]}$. Since $\sign(a_1) = \sign(a_2)$ if and only if $a_1 = t\cdot a_2$ for some $t>0$, we can substitute $y_k,2\leq k\leq d$ in (\ref{3.5}) with $\lambda_k\cdot x_k$ for some $\lambda_k >0$, and $\lambda_2 \neq 1$ to guarantee $x_2\neq y_2$. Analogously, $\sign(\textbf{B}^{p,q}\textbf{x}) = \sign(\textbf{B}^{p,q}\textbf{y})$ can be written as $\sign(\sum_{k=1}^d b_{jk}x_k) = \sign(\sum_{k=1}^d b_{jk}y_k)$ for all $d+1\leq j\leq m$, then for each $j$ we can introduce positive number $\lambda_j$  and reformulate it as $\sum_{k=1}^d b_{jk}x_k = \lambda_j \sum_{k=1}^d b_{jk}y_k$. Thus, under fixed $p,q$ ($1\leq p<q\leq d$), (\ref{3.5}) can be equivalently given as (we let $x_1=\lambda_1 = 1$ so $y_1 = \lambda_1x_1= 1$)
  \begin{equation}
      \begin{aligned}
          \label{3.7}
         & \exists~ x_2 \in \mathbb{C}\setminus \{0\}, x_3,\cdots,x_d\in \mathbb{C} ,\lambda_2 \in \mathbb{R}_+\setminus\{1\}, \lambda_3,\cdots,\lambda_d,\lambda_{d+1},\cdots,\lambda_m \in \mathbb{R}_+,\\ &\text{ such that ~}\forall ~d+1\leq j\leq m,~(1-\lambda_j)b_{j1}+ x_2(1-\lambda_2\lambda_j)b_{j2}+\sum_{k=3}^d x_k(1-\lambda_k\lambda_j)b_{jk} = 0.
      \end{aligned}
  \end{equation}
  For $d+1\leq j\leq m$, if $\lambda_j \neq 1$,   the   equation contained in (\ref{3.7}) is equivalent to solving $b_{j1}$ as
  \begin{equation}
      \label{3.8}
      b_{j1} = \frac{\sum_{k=2}^d x_k(1-\lambda_k\lambda_j)b_{jk} }{\lambda_j - 1}.
  \end{equation}
  Otherwise, if $\lambda_j = 1$, the equation in (\ref{3.7}) can be written as 
  \begin{equation}
      \label{3.9}
      b_{j2}=\frac{\sum_{k=3}^dx_k(1-\lambda_k )b_{jk}}{(\lambda_2-1)x_2}.
  \end{equation}
  Therefore, we can use $\mathcal{J}\subset \{d+1,\cdots,m\}$ to denote the set of  $j$ such that $\lambda_j \neq  1$, and let $\mathcal{J}_1 = \big([m]\setminus[d]\big)\setminus \mathcal{J}$, then (\ref{3.7}) can be rephrased as 
    \begin{equation}
      \begin{aligned}
          \label{3.10}
         & \exists~ \mathcal{J}\subset \{d+1,\cdots,m\}, x_2 \in \mathbb{C}\setminus \{0\}, x_3,\cdots,x_d\in \mathbb{C} ,\\ &\lambda_2 \in \mathbb{R}_+\setminus\{1\},\lambda_3,\cdots,\lambda_d \in \mathbb{R}_+, \{\lambda_j:j\in\mathcal{J}\}\subset \mathbb{R}_+,\\ &\text{ such that ~} \forall ~j\in\mathcal{J},~(\ref{3.8}) \text{ holds},~\forall j\in\mathcal{J}_1,~(\ref{3.9}) \text{ holds}.
      \end{aligned}
  \end{equation}
  Thus, recall the definition of $\mathcal{T}_{\mathcal{J}}$ given in (\ref{726add}), under fixed $p,q$, for those $\textbf{B}^{p,q}=[b_{jk}]$ satisfying (\ref{3.10}), $b_{jk}$ when $(j,k)\in \mathcal{T}_{\mathcal{J}}$ can take any   value in $\mathbb{C}$, while $b_{jk}$ with $(j,k)\notin \mathcal{T}_{\mathcal{J}}$ should be determined by   (\ref{3.8}) or (\ref{3.9}). Compared with the mapping $f_{\mathcal{J}}$ defined in the Theorem (See (\ref{726add2}), (\ref{726add4}) and (\ref{726add3})), we conclude that   $\bm{\mathrm{A}} = [\bm{\mathrm{I}_d},\bm{\mathrm{B^\top}}]^{\bm{\top}}$ is not magnitude retrievable, if and only if  for some $p,q$ ($1\leq p <q \leq d$), for some $\mathcal{J} \subset [m]\setminus [d]$, $\textbf{B}^{p,q} =\textbf{B} \mathscr{T}_{1,p}\mathscr{T}_{2,q} \in f_\mathcal{J}(\Omega_\mathcal{J})$. We slightly abuse the notation and allow $\mathscr{T}^{p,q}$ to element-wisely operate on a set, then 
    this can be written  as $\textbf{B} \in f_\mathcal{J}(\Omega_\mathcal{J}) \mathscr{T}_{2,q}\mathscr{T}_{1,p}$. By taking the union over $p,q,\mathcal{J}$, (\ref{3.3}) follows. 
    \end{proof}

  \vspace{2mm}
  Now we are in a position to give the proof of Theorem \ref{T5}.

 \vspace{2mm}

	\noindent
	{\it Proof of Theorem \ref{T5}:}
	By identifying $z\in \mathbb{C}$ with $(\Re(z)\ \Im(z))^{\top}\in \mathbb{R}^2$, $f_\mathcal{J}$ can be equivalently viewed as a mapping from $$\Omega_{\mathcal{J},\mathbb{R}}:=\mathbb{R}^{2(m-d)(d-1)}\times \big(\mathbb{R}^2 \setminus \{(0,0)\}\big)\times \mathbb{R}^{2(d-2)}\times \mathbb{R}_+\setminus \{1\}\times \mathbb{R}_+^{d-2}\times \big(\mathbb{R}_+\setminus \{1\}\big)^{|\mathcal{J}|}$$ to $\mathbb{R}^{2(m-d)d}$. Note that the $\Omega_{\mathcal{J},\mathbb{R}}$ is an open subset of $\mathbb{R}^{(2m-2d+3)(d-1)+|\mathcal{J}|}$. It is quite obvious that $f_\mathcal{J}$ is   smooth  (i.e., infinitely continuously differentiable). To be more concrete, we confirm this via calculations. 
	We use the notations in the definition of $f_\mathcal{J}$ in   Lemma \ref{lemmanew} (See   (\ref{726add2}), (\ref{726add4}) and (\ref{726add3})) and write $\bm{\zeta}=[\zeta_{jk}]=f_\mathcal{J}(\bm{\beta})$. Here we view $f_\mathcal{J}$ as a real  mapping, and so $\Re(\beta_{jk}),\Im(\beta_{jk}),\Re(\beta_k),\Im(\beta_k),\lambda_j$ are variables (i.e., components of $\bm{\beta}$ in (\ref{726add4})), and $\Re (\zeta_{jk}),\Im (\zeta_{jk})$ are components of the $\bm{\zeta}$. Thus, some algebra gives the translation of (\ref{726add3}) as    When   $(j,k)\in \mathcal{T}_{\mathcal{J}}$, 
	$\Re(\zeta_{jk})=\Re(\beta_{jk}),\ 
	\Im(\zeta_{jk})=\Im(\beta_{jk})$;  when $j\in \mathcal{J}$,
	\begin{equation}
		\begin{cases}
		 \Re(\zeta_{j1})= {\displaystyle \sum_{k=2}^d} \frac{1-\lambda_k\lambda_j}{\lambda_j-1} \cdot  
			\big(\Re(\beta_{jk})\Re(\beta_k)-\Im(\beta_{j k})\Im(\beta_k)\big),  \\
			\Im(\zeta_{j1})= {\displaystyle \sum_{k=2}^d} 
			\frac{1-\lambda_k\lambda_j}{\lambda_j-1} \cdot  \big(\Re(\beta_{jk}) \Im (\beta_k)+ \Im(\beta_{jk})\Re(\beta_k) \big). 
		\end{cases}
	\nonumber
	\end{equation}
	When $j\in \mathcal{J}_1$,
		\begin{equation}
		\begin{cases}
			  \Re(\zeta_{j2})= \\
				{\displaystyle \sum_{k=3}^d
				\frac{(1-\lambda_k) \cdot  (
			 \Re(\beta_{jk})\Re(\beta_k)\Re(\beta_2)+\Re(\beta_{jk})\Im(\beta_k)\Im(\beta_2)-
			 \Im(\beta_{jk})\Im(\beta_k)\Re(\beta_2)+\Im(\beta_{jk})\Re(\beta_k)\Im(\beta_2)  )}{(\lambda_2-1)\cdot(   [\Re(\beta_2)]^2+[\Im(\beta_2)]^2) }},    \\
			\Im(\zeta_{j2})= \\
			{\displaystyle \sum_{k=3}^d
			\frac{ (1-\lambda_k)\cdot(
			\Re(\beta_{jk})\Im(\beta_k)\Re(\beta_2)-\Re(\beta_{jk})\Re(\beta_k)\Im(\beta_2)+
			\Im(\beta_{jk})\Re(\beta_k)\Re(\beta_2)+\Im(\beta_{jk})\Im(\beta_k)\Im(\beta_2))}{(\lambda_2-1)\cdot([\Re(\beta_2)]^2+[\Im(\beta_2)]^2)} }.
		\end{cases}
		\nonumber
	\end{equation} 
  Thus, $f_\mathcal{J}$ is smooth. When $m\geq 4d-2$, for any $\mathcal{J}\subset [m]\setminus [d]$ we have 
$$
2(m-d)d >(2m-2d+3)(d-1)+m-d \geq (2m-2d+3)(d-1)+|\mathcal{J}|.
$$ 
 Then by Sard Theorem (see, e.g., \cite{sard1942measure}), $f_\mathcal{J}(\Omega_Y)$ has zero Lebesgue measure, which by (\ref{3.3}) implies that $\mathscr{X}$, the set of "undesired" \textbf{B}, 
has zero Lebesgue measure. Hence, when $m\geq 4d-2$, almost all $\textbf{B}$ in $\mathbb{C}^{(m-d)\times d}$ belong to $\mathbb{C}^{(m-d)\times d}\setminus \mathscr{X}$, i.e., $[\bm{\mathrm{I}_d},\bm{\mathrm{B^\top}}]^{\bm{\top}}\in \mathbb{C}^{m\times d}$ is magnitude retrievable.

To prove the first statement in Theorem \ref{T5} we still need to extend   canonical measurement matrix to the general $\textbf{A}\in\mathbb{C}^{m\times d}$ when $m\geq 4d-2$.
Given $ \bm{\mathrm{A}}= [\bm{\mathrm{A_1^\top}},\bm{\mathrm{A_2^\top}}]^{\bm{\top}}\in \mathbb{C}^{n\times d}$ with $\bm{\mathrm{A_1 }}\in\mathbb{C}^{d\times d}$, $\bm{\mathrm{A_2 }}\in\mathbb{C}^{(m-d)\times d}$,  a   mapping $\widetilde{f}$ from $\mathbb{C}^{m\times d}$ to $\mathbb{C}^{m\times d}$ is defined by $\widetilde{f}(\bm{\mathrm{A}}) =[\bm{\mathrm{A_1^\top}},\bm{\mathrm{(A_2A_1)^\top}}]^{\bm{\top}}$. Moreover, we let $  \widetilde{\mathscr{X}}:= \big\{\textbf{A} \in\mathbb{C}^{m\times d}: \mathcal{W}_{\textbf{A}} \neq \mathbb{C}^d \big\}$ and aim to show $\widetilde{\mathscr{X}}$ is of zero Lebesgue measure. Assume $\textbf{A} = [\bm{\mathrm{A_1^\top}},\bm{\mathrm{A_2^\top}}]^{\bm{\top}}\in\widetilde{\mathscr{X}}$. If $\bm{\mathrm{A_1 }}\in\mathbb{C}^{d\times d}$ is invertible, then we have $\mathcal{W}_{\bm{\mathrm{AA_1^{-1}}}} = \bm{\mathrm{A_1}}\mathcal{W}_{\textbf{A}}$ (see the discussion at the beginning of Section \ref{sec2}). Thus, $\bm{\mathrm{AA_1^{-1}}} = [\bm{\mathrm{I}_d}, \bm{\mathrm{(A_2A_1^{-1})^\top}}]^{\bm{\top}}\in \widetilde{\mathscr{X}}$, which gives $\bm{\mathrm{A_2A_1^{-1}}}\in \mathscr{X}$. Evidently, it holds that $$\textbf{A} = \begin{bmatrix}
    \bm{\mathrm{A_1}} \\ \bm{\mathrm{A_2}}
\end{bmatrix} = \widetilde{f}\left(\begin{bmatrix}
    \bm{\mathrm{A_1}} \\ \bm{\mathrm{A_2A_1^{-1}}}
\end{bmatrix}\right).$$
Therefore, we obtain $\textbf{A}\subset \widetilde{f}\big(\mathbb{C}^{d\times d} \times \mathscr{X}\big)$. We further take singular $\bm{\mathrm{A_1}}$ into account, it delivers 
\begin{equation}
    \label{3.15}
    \widetilde{\mathscr{X}}\subset  \widetilde{f}\big(\mathbb{C}^{d\times d} \times \mathscr{X}\big) \cup \{\textbf{A}:\bm{\mathrm{A_1}}\text{ is singular}\}.
\end{equation}
From (\ref{3.15}) it is evident that when $m\geq 4d-2$ almost all $\textbf{A}$ are magnitude retrievable.
\hfill $\square $

\vspace{2mm}

    In the following, we turn to bound $\bm{\mathrm{m_{all}}}(d) $ from below. 
   From now on, we assume $\A = [\bm{\gamma_1},\bm{\gamma_2},\cdots,\bm{\gamma_m}]^\top$ has no zero row, that is, $\bm{\mathrm{\gamma_j^\top}}\neq \bm{0}$ for each $j\in [m]$. 
   We define the set of the $\x$ such that $|\mathrm{N}(\Ax)| = m $ to be
	\begin{equation}
	    \label{3.16}
	    \mathcal{H}_{\textbf{A}}:=\{\x\in \mathbb{C}^d: \bm{\mathrm{\gamma_j^\top}}\x\neq 0,~\forall j\in [m]\}.
	\end{equation}
    Note that   no zero measurement  $\bm{\mathrm{\gamma_j^\top}}\textbf{x}=0$ occurs if $\x \in \mathcal{H}_{\textbf{A}}$, hence for such signals one observes  purely phase-only measurements. Besides,
	note that $\ker(\bm{\mathrm{\gamma_j^\top}})=\{\x\in \mathbb{C}^d:\bm{\mathrm{\gamma_j^\top}}\textbf{x}=0\}$ is $(d-1)$-dimensional linear subspace of $\mathbb{C}^d$, by writing
	$$\mathcal{H}_{\textbf{A}}= \bigcap_{j=1}^m  \big(\mathbb{C}^d\setminus\ker(\bm{\mathrm{\gamma_j^\top}})\big)  =
	\mathbb{C}^d\setminus \big(\bigcup_{j=1}^m\ker(\bm{\mathrm{\gamma_j^\top}})\big),
	$$ 
	one shall easily see that $\mathcal{H}_{\textbf{A}}$ is Zariski open.

	If $\x \in \mathcal{H}_{\textbf{A}}$ can be recovered, Theorem \ref{T1} gives $\rank (\mathcal{D}_{\textbf{A}}(\textbf{x})) = 2d+|\mathrm{N}(\textbf{Ax})|-1 = 2d+m-1$, thus we arrive at $2m \geq 2d+m-1$ (since $\mathcal{D}_{\textbf{A}}(\textbf{x})$ has $2d$ rows in total). This directly delivers the lower bound $\bm{\mathrm{m_{all}}}(d)  \geq 2d-1$. In the next Theorem, we use $\mathcal{E}_{\textbf{A}}(\textbf{x})$ instead and show a slightly tighter result. The key idea is to construct a specific $\bm{\mathrm{x_0}}\in\mathcal{H}_{\textbf{A}}$ such that one row of $\mathcal{E}_{\textbf{A}}(\bm{\mathrm{x_0}})$ is zero, then the result follows from similar rank argument.   
 
	\begin{theorem}
		When $d>1$, $\bm{\mathrm{m_{all}}}(d) \geq 2d$.
		\label{T4}
	\end{theorem} 
	
\begin{proof}	We can only consider $\bm{\mathrm{A}} = [\bm{\mathrm{I}_d},\bm{\mathrm{A_1^\top}}]^{\bm{\top}}(\in \mathbb{C}^{m\times d})$ that is magnitude retrievable, i.e.,  $\mathcal{W}_{\bm{\mathrm{A}}}=\mathbb{C}^d$. For $\x\in \mathcal{H}_{\textbf{A}}$, recall the notations introduced at the beginning of Section \ref{canoin} $\textbf{A} = [r_{jk}e^{\textbf{i}\theta_{jk}}]$, $\textbf{x} = [|x_k|e^{\textbf{i}\alpha_k}]$, $e^{\textbf{i}\delta_j} = \sign(\bm{\mathrm{\gamma_j^\top}} \textbf{x})$ when $j\in [m]\setminus [d]$, and also the construction of $\mathcal{E}_{\textbf{A}}(\x)$ in (\ref{2.7}), (\ref{2.8}), (\ref{2.9}). It is evident that $\mathcal{E}_{\textbf{A}}(\x)=(r_{jk}\sin(\theta_{jk}+\alpha_k-\delta_j))_{j\in [m]\setminus [d],k\in [d] }\in \mathbb{R}^{(m-d)\times d}$. Besides, Theorem \ref{T2} gives $\rank(\mathcal{E}_{\textbf{A}}(\x))=|\mathrm{N}(\x)|-1=d-1$. 
	
	We now further specify a signal $\bm{\mathrm{x_0}}$ in $\mathcal{H}_{\textbf{A}}$ such that the first row of $\mathcal{E}_{\textbf{A}}(\bm{\mathrm{x_0}})$ vanishes.	Firstly we let $\alpha_k = -\theta_{d+1,k}$, hence we are considering $$	\x=[\lambda_1e^{-\textbf{i}\theta_{d+1,1}},\lambda_2e^{-\textbf{i}\theta_{d+1,2}},\cdots,\lambda_de^{-\textbf{i}\theta_{d+1,d}}]^{\top}=\bm{\mathrm{E\Lambda}}$$
	where $\textbf{E}=\mathrm{dg}([e^{-\textbf{i}\theta_{d+1,1}},\cdots,e^{-\textbf{i}\theta_{d+1,d}}]^{\top})$ has been specified, while $\bm{\Lambda}=[\lambda_1,\cdots,\lambda_d]^{\top}\in \mathbb{R}_+^d$ will be properly set later, to guarantee our initial assumption $\x\in\mathcal{H}_{\textbf{A}}$.

	For $j \in [m]\setminus [d]$, if $\Re(\bm{\mathrm{\gamma_j^\top}}\textbf{E})\neq \bm{0}$, we let $\bm{\mathrm{\xi_j^{\top}}} =\Re(\bm{\mathrm{\gamma_j^\top}}\textbf{E})$. Otherwise, (i.e., $\Re(\bm{\mathrm{\gamma_j^\top}}\textbf{E})= \bm{0}$), since $\bm{\mathrm{\gamma_j}}\neq \bm{0}$, we have $\bm{\mathrm{\gamma_j^\top}}\textbf{E}\neq \bm{0}$, so we can let 
	$\bm{\mathrm{\xi_j^{\top}}} =\Im(\bm{\mathrm{\gamma_j^\top}}\textbf{E})$. Note that  $\bm{\mathrm{\xi_j}} \in \mathbb{R}^d \setminus \bm{0}$, and so there exists $\bm{\Lambda_0} = [\lambda_{10},\cdots,\lambda_{d0}]^\top \in\mathbb{R}_+^d$ such that $\bm{\mathrm{\xi_j^{\top}}} \bm{\Lambda_0} \neq 0 $ for all $j\in [m]\setminus [d]$. On the other hand, since either $\bm{\mathrm{\xi_j^{\top}}} =\Re(\bm{\mathrm{\gamma_j^\top}}\textbf{E})$ or $\bm{\mathrm{\xi_j^{\top}}} =\Im(\bm{\mathrm{\gamma_j^\top}}\textbf{E})$ holds, $\bm{\mathrm{\xi_j^{\top}}}\bm{\Lambda_0} \neq 0 $ can imply $\bm{\mathrm{\xi_j^{\top}}} \big(\textbf{E}\bm{\Lambda_0}\big) \neq 0$. Thus, we can consider the signal $$\bm{\mathrm{x_0}} :=\textbf{E}\bm{\Lambda_0}= [\lambda_{10}e^{-\textbf{i}\theta_{d+1,1}},\lambda_{20}e^{-\textbf{i}\theta_{d+1,2}},\cdots,\lambda_{d0}e^{-\textbf{i}\theta_{d+1,d}}]^{\top}$$
	that satisfies $\bm{\mathrm{x_0}}\in \mathcal{H}_{\textbf{A}}$, $\sign(x_k) = e^{\textbf{i}\alpha_k} = e^{-\textbf{i}\theta_{d+1,k}}$. 
	
	We calculate   the $d+1$ measurement $e^{\textbf{i}\delta_{d+1}}$ as follows $$
	e^{\textbf{i}\delta_{d+1}}=\sign \left 
	(\sum_{k=1}^dr_{d+1,k}\cdot \lambda_{k0}\cdot e^{\textbf{i}(\theta_{d+1,k}+\alpha_k)} \right )=\sign \left (\sum_{k=1}^dr_{d+1,k}\cdot \lambda_{k0} \right )=1.
	$$ This gives $\delta_{d+1}=0$, which leads to  $ \sin(\theta_{d+1,k}+\alpha_k-\delta_{d+1})=0$, indicating that the first row of $\mathcal{E}_{\textbf{A}}(\bm{\mathrm{x_0}})$ vanishes. By applying $\rank(\mathcal{E}_{\textbf{A}}(\bm{\mathrm{x_0}}))=d-1$, we have $m-d-1\geq d-1$, which gives $m\geq 2d$ and hence the proof is concluded.
	\end{proof}

	\section{Reconstruction of almost all signals}\label{sec4}
 
	Although the minimal measurement number for $\mathcal{W}_{\textbf{A}}=\mathbb{C}^d$ (i.e., $\bm{\mathrm{m_{all}}}(d)$) is of fundamental theoretical interest, from the perspective of practicality, a sufficiently large set of recoverable signals is often satisfactory. We point out that, the measurement number   for  phase retrieval  of almost all signals has   been  studied in \cite{balan2006signal,huang2021almost}.

	In this section, we study the minimal measurement number for recovering   almost all $d$-dimensional complex-valued signals from phase. By convention, the measurement matrix $\textbf{A}\in\mathbb{C}^{m\times d}$ is said to be {\it almost everywhere magnitude retrievable} if $\mathbb{C}^d \setminus \mathcal{W}_{\textbf{A}}$ has zero Lebesgue measure. Thus, the measurement number of interest can be formally defined as
	\begin{equation}
	    \label{4.1}
	    \bm{\mathrm{m_{ae}}}(d)=\big\{m\in 
        \mathbb{N}_+: \text{some }\textbf{A}\in\mathbb{C}^{m\times d} \text{ is almost everywhere magnitude retrievable}\big\}.
	\end{equation}

	  The main result in this section states that a generic $\textbf{A}$ possesses $\mathcal{W}_{\textbf{A}}$ that contains a generic signal if $m\geq 2d-1$. Recall that $\mathcal{W}_{\textbf{A}}$ containing a generic $\x$ has   complement $\mathbb{C}^d\setminus\mathcal{W}_{\textbf{A}}$  of zero Lebesgue measure (Section \ref{preli}), this main result directly implies $ \bm{\mathrm{m_{ae}}}(d)\leq 2d-1$. Combining with an easier fact $\bm{\mathrm{m_{ae}}}(d)\geq 2d-1$,  $ \bm{\mathrm{m_{ae}}}(d) = 2d-1$ can be concluded.

	Before proceeding we need to introduce some notations. We first extend the definition of $\mathcal{H}_{\textbf{A}}$ in (\ref{3.16}). Given $\mathcal{S}\subset [m]$, we write    $[m]\setminus \mathcal{S}$ as $\mathcal{S}^c$ and then define 
	\begin{equation}
	    \label{4.2}
	    \mathcal{H}_{\textbf{A}}(\mathcal{S})=\{\x\in \mathbb{C}^d:\bm{\mathrm{\gamma_j^{\top}x}}=0,~\forall j\in \mathcal{S};~\bm{\mathrm{\gamma_j^{\top}x}}\neq 0,~\forall j\in \mathcal{S}^c\},
	\end{equation}
	which can also be equivalently given by
	$ \mathcal{H}_{\textbf{A}}(\mathcal{S})=\ker(\textbf{A}^\mathcal{S})\cap\mathcal{H}_{\textbf{A}^{\mathcal{S}^c}}.  $
	Note that $\mathcal{H}_{\textbf{A}}(\varnothing)$ recovers $\mathcal{H}_{\textbf{A}}$ defined in (\ref{3.16}), and  we have the partition 
	$ \mathbb{C}^d=\bigcup_{\mathcal{S}\subset [m]}\mathcal{H}_{\textbf{A}}(\mathcal{S})$  where  $ \mathcal{H}_{\textbf{A}}(\mathcal{S})\cap \mathcal{H}_{\textbf{A}}(\mathcal{T})=\varnothing$ if $ \mathcal{S}\neq \mathcal{T}.$ We will deal with signals in each $\mathcal{H}_{\textbf{A}}(\mathcal{S})$ separately.

	Previously, we focus on analyzing $\mathcal{W}_{\textbf{A}}$ for a fixed $\textbf{A}$, but in some cases it is more conducive to consider the measurement matrices that can recover a fixed signal $\x$. More notations are needed to this end. For a fixed measurement number $m$, we collect the   $\textbf{A}$'s such that $\x\in \mathcal{W}_{\textbf{A}}$ in the set 
	\begin{equation}
	    \label{4.3}
	    \mathcal{W}_{\textbf{x}}(m)= \big\{\textbf{A}\in\mathbb{C}^{m\times d}:\textbf{x}\in\mathcal{W}_{\textbf{A}}\big\}.
	\end{equation}
	Similarly, corresponding to $\mathcal{H}_{\textbf{A}}$, the measurement matrices that give purely phase-only measurements for a fixed $\textbf{x}$ are collected in the set 
	\begin{equation}
	    \label{4.4}
	        \mathcal{H}_{\textbf{x}}(m) = \big\{\textbf{A}\in\mathbb{C}^{m\times d}: \textbf{x}\in\mathcal{H}_{\A}\big\}.
	\end{equation}

	We will use the shorthand $[\mathcal{S}]$ to denote $[|\mathcal{S}|]$, i.e., $[\mathcal{S}]=\{1,\cdots,|\mathcal{S}|\}.$
	
	\begin{theorem}
		Consider $\bm{\mathrm{A}}\in \mathbb{C}^{m\times d}$. When $m\leq 2d-2$, $\mathcal{W}_{\bm{\mathrm{A}}}$ is nowhere dense (under Euclidean topology) and of zero Lebesgue measure. 
		When $m\geq 2d-1$, a generic  $\bm{\mathrm{A}}$ satisfies that
		\begin{equation}
			\text{for all } \mathcal{S}\subset [m],\ \mathcal{W}_{\bm{\mathrm{A}}}\cap \ker(\bm{\mathrm{A}}^\mathcal{S})\ is \ generic\  in\  \ker(\bm{\mathrm{A}}^\mathcal{S}).
			\label{4.5}
		\end{equation}
		Specifically, let $\mathcal{S}=\varnothing$, it gives that $\mathcal{W}_{\bm{\mathrm{A}}}$ is generic in $\mathbb{C}^d$. Therefore, the minimal measurement number for almost everywhere magnitude retrievable property is  $\bm{\mathrm{m_{ae}}}(d) = 2d-1$.
		\label{T6}
	\end{theorem}
	
For clarity, let us first present some lemmas 
for the proof of Theorem \ref{T6}. We give Lemma \ref{lemma3} to characterize signals in $\mathcal{W}_{\bm{\mathrm{A}}}\cap \mathcal{H}_{\textbf{A}}$ via discriminant matrix $\mathcal{D}_{\textbf{A}}(\textbf{x})$. 

\begin{lem}
	$\bm{\mathrm{x}}\in \mathcal{W}_{\bm{\mathrm{A}}} \cap \mathcal{H}_{\bm{\mathrm{A}}}$ if and only if $\rank( \mathcal{D}_{\bm{\mathrm{A}}}(\bm{\mathrm{x}}))\geq2d+m-1$.
	\label{lemma3}
\end{lem}

\begin{proof}
The "only if" part follows directly from Theorem \ref{T1}, so it  remains to show the "if" part. Recall the linear system (\ref{eq1}) and linear subspace (\ref{subspace}), the beginning of the proof of Lemma \ref{lemma1} gives $\dim(\mathrm{V}_{\x})\geq 1$. Moreover, the proof of Theorem \ref{T1} indeed delivers $\dim(\mathrm{V}_{\x})=2d+|\mathrm{N}(\Ax)|-\rank( \mathcal{D}_{\bm{\mathrm{A}}}(\x))$, see (\ref{eeq1}). Thus, we obtain $\rank( \mathcal{D}_{\bm{\mathrm{A}}}(\x))\leq 2d+|\mathrm{N}(\Ax)|-1$. We invoke the condition $\rank( \mathcal{D}_{\bm{\mathrm{A}}}(\x))\geq 2d+m-1$, it gives $2d+|\mathrm{N}(\Ax)|-1\geq 2d+m-1$, and hence $|\mathrm{N}(\Ax)|\geq m$. Evidently, $|\mathrm{N}(\Ax)|\leq m$, so it leads to $|\mathrm{N}(\Ax)|=m$, or equivalently, $\x\in \mathcal{H}_{\textbf{A}}$. On the other hand, it has also been verified that $\rank(\mathcal{D}_{\bm{\mathrm{A}}}(\x))=2d+|\mathrm{N}(\Ax)|-1$, thus Theorem \ref{T1} delivers $\x\in \mathcal{W}_{\bm{\mathrm{A}}}$. Hence, $\x\in \mathcal{W}_{\bm{\mathrm{A}}} \cap \mathcal{H}_{\textbf{A}}$. The proof is  complete. 
\end{proof}

The next lemma is concerned with the rank of a  matrix with components being fractional functions of a complex-valued vector $\x$.  

\begin{lem}
	Assume $\bm{\mathrm{x}}\in \mathbb{C}^{n_0}$, $\bm{\Phi}(\bm{\mathrm{x}})=\left 
	[\displaystyle \frac{f_{i j}(\bm{\mathrm{x}})}{g_{i j}(\bm{\mathrm{x}})} \right ]\in \mathbb{C}^{n_1\times n_2}$, where $ f_{i j}(\bm{\mathrm{x}}), g_{i j}(\bm{\mathrm{x}}) $ are polynomials with real variables $[\Re(\bm{\mathrm{x}})^\top, \Im(\bm{\mathrm{x}})^\top ]^\top$ and possibly complex coefficients, $g_{ij}(\bm{\mathrm{x}})$ is not zero polynomial. 
	 We consider  $ \Omega=\{\bm{\mathrm{x}}\in\mathbb{C}^{n_0}:g_{i  j}(\bm{\mathrm{x}})\neq 0,~\forall ~(i,j)\in [n_1]\times [n_2]\}$. Then given any positive integer
	  $r $,   $\{\bm{\mathrm{x}}\in\Omega:\rank(\bm{\Phi}(\bm{\mathrm{x}}))\geq r \}$ is Zariski open set. 
	\label{lemma4}
\end{lem}

\begin{proof}
The conclusion is trivial when $\{\textbf{x}:\rank(\bm{\Phi}(\x))\geq r \}=\varnothing$, so we only consider non-empty $\{\textbf{x}:\rank(\bm{\Phi}(\x))\geq r \}$, this will lead to $r\leq \min\{n_1,n_2\}$. We use $\bm{\Phi}_t(\x),\ t\in \mathcal{T}_r$ to denote all the $r\times r$ submatrices of $\bm{\Phi}(\x)$, and evidently $\mathcal{T}_r$ is a finite set. Then we have 
\begin{equation}
	\begin{split}
		\begin{aligned}
			&\{\x\in\Omega:\rank(\bm{\phi}(\x))\geq r \}=\bigcup_{t\in \mathcal{T}_r}\{\x\in \Omega: \det(\bm{\phi}_t(\x))\neq 0\}\\
			&=\bigcup_{t\in \mathcal{T}_r}\{\x\in\mathbb{C}^{n_0}:g_{ij}(\x)\neq 0, ~\forall ~(i,j)\in [n_1]\times [n_2];\ \det(\bm{\phi}_t(\x))\neq 0\}\\
			&=\bigcup_{t\in \mathcal{T}_r}\{\x\in\mathbb{C}^{n_0}:g_{ij}(\x)\neq 0,~\forall ~(i,j)\in [n_1]\times [n_2];\ \prod_{i,j}g_{i j}(\x)\det(\bm{\phi}_t(\x))\neq 0\},
			\nonumber
		\end{aligned}
	\end{split}
\end{equation}
and it is not hard to see that  $\{\prod_{i,j}g_{i j}(\x)\det(\bm{\phi}_t(\x)):t\in \mathcal{T}_r\}$ are polynomials with real variables $[\Re(\bm{\mathrm{x}})^\top, \Im(\bm{\mathrm{x}})^\top ]^\top$ and   complex coefficients.  Since $\{\textbf{x}:\rank(\bm{\phi}(\x))\geq r\}\neq\varnothing$, there exists $t_0 \in \mathcal{T}_r$ such that $\prod_{i,j}g_{i,j}(\x)\det(\bm{\phi}_t(\x))$ is nonzero polynomial, and hence $\{\textbf{x}:\rank(\bm{\phi}(\x))\geq r\}$ is non-empty Zariski open set. 
\end{proof}

Note that in Lemma \ref{lemma3} $\mathcal{W}_{\textbf{A}}\cap \mathcal{H}_{\textbf{A}}$ is precisely characterized by $\{\x:\rank\big(\mathcal{D}_{\textbf{A}}(\textbf{x})\big)\geq 2d+m-1\}$, which is Zariski open by Lemma \ref{lemma4}. However, it is currently unclear whether similar results can be established for $\mathcal{W}_{\textbf{A}}\cap \big(\mathbb{C}^d \setminus \mathcal{H}_{\textbf{A}}\big) = \bigcup_{\mathcal{S}\neq \varnothing}\mathcal{W}_{\textbf{A}}\cap \mathcal{H}_{\textbf{A}}(\mathcal{S})$. Lemma \ref{lemma5} affirmatively answers this question. Particularly, it   transfers $\mathcal{W}_{\textbf{A}}\cap \mathcal{H}_{\textbf{A}}(\mathcal{S})$   ($\mathcal{S}\neq \varnothing$) to the more amenable set $\mathcal{W}_{\bm{\mathrm{A^{'}}}}\cap \mathcal{H}_{\bm{\mathrm{A^{'}}}}$ that can be handled via Lemma \ref{lemma3}. Here, $ {\bm{\mathrm{A^{'}}}}$ is a new matrix constructed from $\textbf{A}$. The proof of Lemma \ref{lemma5} is quite tedious, whereas the core spirit is rather elementary and comes from the elimination method for solving a linear system. 

\begin{lem}
    Consider $\mathcal{S}\subset [m]$, $1\leq|\mathcal{S}|<d$, and recall the notation $[\mathcal{S}]=\{1,\cdots, |\mathcal{S}|\}$. Given $\bm{\mathrm{A}}\in\mathbb{C}^{m\times d}$ we assume $\rank(\bm{\mathrm{A}}^\mathcal{S}_{[\mathcal{S}]})=|\mathcal{S}|$, $\bm{\mathrm{x}}\in \ker(\bm{\mathrm{A}}^\mathcal{S})$. Define \begin{equation}
    \label{4.6}
        \bm{\mathrm{A}}(\mathcal{S})=\bm{\mathrm{A}}^{\mathcal{S}^c}_{[d]\setminus [\mathcal{S}]}-\bm{\mathrm{A}}^{\mathcal{S}^c}_{[\mathcal{S}]}(\bm{\mathrm{A}}^\mathcal{S}_{[\mathcal{S}]})^{-1}\bm{\mathrm{A}}^{\mathcal{S} }_{[d]\setminus [\mathcal{S}]} \in \mathbb{C}^{(m-|\mathcal{S}|)\times (d-|\mathcal{S}|)}.
    \end{equation} Then we have $\bm{\mathrm{x}}\in \mathcal{W}_{\bm{\mathrm{A}}}\cap \mathcal{H}_{\bm{\mathrm{A}}}(\mathcal{S})$ if and only if $\bm{\mathrm{x}}^{[d]\setminus [\mathcal{S}]} \in \mathcal{W}_{\bm{\mathrm{A}}(\mathcal{S})}\cap \mathcal{H}_{\bm{\mathrm{A}}(\mathcal{S})}$.
	\label{lemma5}
\end{lem}

\begin{proof}
Note that $\textbf{x}\in \ker\big(\textbf{A}^\mathcal{S}\big)$ equals $\textbf{A}^\mathcal{S}\textbf{x} = \bm{0}$. Due to the condition $\rank(\bm{\mathrm{A}}^\mathcal{S}_{[\mathcal{S}]})=|\mathcal{S}|$, it can  also be equivalently given by 
\begin{equation}
    \label{4.7}
    \begin{bmatrix}
         \textbf{A}^{\mathcal{S}}_{[\mathcal{S}] } & \textbf{A}^\mathcal{S}_{[d]\setminus [\mathcal{S}]}
    \end{bmatrix}\cdot \begin{bmatrix}
         \textbf{x}^{[\mathcal{S}]} \\ \textbf{x}^{[d]\setminus [\mathcal{S}]}
    \end{bmatrix} = \bm{0} \iff \textbf{x}^{[\mathcal{S}]}=-(\textbf{A}^\mathcal{S}_{[\mathcal{S}]})^{-1}\textbf{A}^\mathcal{S}_{[d]\setminus[\mathcal{S}]}\cdot \textbf{x}^{[d]\setminus [\mathcal{S}]}.
\end{equation}
Based on (\ref{4.7}), and recall $\textbf{A}(\mathcal{S})$ in (\ref{4.6}), some algebra yields 
\begin{equation}
    \label{4.8}
    \textbf{A}^{\mathcal{S}^c}\textbf{x} =  \begin{bmatrix}
         \textbf{A}^{\mathcal{S}^c}_{[\mathcal{S}] } & \textbf{A}^{\mathcal{S}^c}_{[d]\setminus [\mathcal{S}]}
    \end{bmatrix}\cdot \begin{bmatrix}
         -(\textbf{A}^\mathcal{S}_{[\mathcal{S}]})^{-1}\textbf{A}^\mathcal{S}_{[d]\setminus[\mathcal{S}]}\cdot \textbf{x}^{[d]\setminus [\mathcal{S}]} \\ \textbf{x}^{[d]\setminus [\mathcal{S}]}
    \end{bmatrix} = \textbf{A}(\mathcal{S})\textbf{x}^{[d]\setminus [\mathcal{S}]}.
\end{equation}
Under the above relations, we need to show the equivalence between $\bm{\mathrm{x}}\in \mathcal{W}_{\bm{\mathrm{A}}}\cap \mathcal{H}_{\bm{\mathrm{A}}}(\mathcal{S})$ and $\bm{\mathrm{x}}^{[d]\setminus [\mathcal{S}]} \in \mathcal{W}_{\bm{\mathrm{A}}(\mathcal{S})}\cap \mathcal{H}_{\bm{\mathrm{A}}(\mathcal{S})}$.

We assume $\bm{\mathrm{x}}\in \mathcal{W}_{\bm{\mathrm{A}}}\cap \mathcal{H}_{\bm{\mathrm{A}}}(\mathcal{S})$. By definition in  (\ref{4.2}), $\bm{\mathrm{A^{\mathcal{S}^c}x}}$   contains no zero entries, and so neither does $\textbf{A}(\mathcal{S})\textbf{x}^{[d]\setminus [\mathcal{S}]}$ by (\ref{4.8}). Hence, we obtain $\bm{\mathrm{x}}^{[d]\setminus [\mathcal{S}]} \in\mathcal{H}_{\bm{\mathrm{A}}(\mathcal{S})}$. To show $\bm{\mathrm{x}}^{[d]\setminus [\mathcal{S}]} \in\mathcal{W}_{\bm{\mathrm{A}}(\mathcal{S})}$ we start from $\sign\big(\bm{\mathrm{A}}(\mathcal{S})\bm{\mathrm{x}}^{[d]\setminus [\mathcal{S}]}\big) =\sign\big( \bm{\mathrm{A}}(\mathcal{S}) \bm{\mathrm{y_0}}\big)$ for some $\bm{\mathrm{y_0}}\in\mathbb{C}^{d-|\mathcal{S}|}$. Motivated by (\ref{4.7}) we consider the $d$-dimensional 
complex-valued signal  
$$ \textbf{y}=  \begin{bmatrix}
-(\textbf{A}^\mathcal{S}_{[\mathcal{S}]})^{-1}\textbf{A}^\mathcal{S}_{[d]\setminus[\mathcal{S}]}\cdot \bm{\mathrm{y_0}} \\ \bm{\mathrm{y_0}}
\end{bmatrix}$$
that satisfies $\textbf{A}^\mathcal{S}\textbf{y} = \bm{0} =\textbf{A}^\mathcal{S}\textbf{x}$. Similar to (\ref{4.8}), some algebra   confirms that $ \textbf{A}^{\mathcal{S}^c}\textbf{y} = \textbf{A}(\mathcal{S})\bm{\mathrm{y_0}}$, hence $\sign\big(\bm{\mathrm{A}}(\mathcal{S})\bm{\mathrm{x}}^{[d]\setminus [\mathcal{S}]}\big) =\sign\big( \bm{\mathrm{A}}(\mathcal{S}) \bm{\mathrm{y_0}}\big)$   implies $\sign\big(\textbf{A}^{\mathcal{S}^c} \textbf{x}\big) =\sign\big(\textbf{A}^{\mathcal{S}^c} \textbf{y}\big)$. Thus, we obtain $\sign(\Ax) = \sign(\Ay)$, hence applying $\textbf{x}\in\mathcal{W}_{\textbf{A}}$ delivers $\x = t\cdot \textbf{y}$ for some $t>0$. It is evident that this can yield $\bm{\mathrm{y_0}} = t^{-1}\cdot \bm{\mathrm{x}}^{[d]\setminus [\mathcal{S}]}$, so $\bm{\mathrm{x}}^{[d]\setminus [\mathcal{S}]} \in\mathcal{W}_{\bm{\mathrm{A}}(\mathcal{S})}$. Thus, $\bm{\mathrm{x}}^{[d]\setminus [\mathcal{S}]} \in \mathcal{W}_{\bm{\mathrm{A}}(\mathcal{S})}\cap \mathcal{H}_{\bm{\mathrm{A}}(\mathcal{S})}$ is concluded.

For the other direction, we assume $ \bm{\mathrm{x}}^{[d]\setminus [\mathcal{S}]} \in \mathcal{W}_{\bm{\mathrm{A}}(\mathcal{S})}\cap \mathcal{H}_{\bm{\mathrm{A}}(\mathcal{S})}$. By (\ref{4.8}) one can easily see that $\textbf{A}^{\mathcal{S}^c}\textbf{x} $   contains no zero entries, which together with the assumption $\textbf{A}^\mathcal{S}\textbf{x} = \bm{0}$ gives $\bm{\mathrm{x}}\in\mathcal{H}_{\bm{\mathrm{A}}}(\mathcal{S})$. It remains to show $\bm{\mathrm{x}}\in \mathcal{W}_{\bm{\mathrm{A}}} $, so we assume $\sign(\Ax) = \sign(\Ay)$, which is equivalent to $ \textbf{A}^\mathcal{S}\textbf{y}  =\textbf{A}^\mathcal{S}\textbf{x} = \bm{0}$ and $\sign\big(\textbf{A}^{\mathcal{S}^c}\textbf{y}\big) =\sign\big(\textbf{A}^{\mathcal{S}^c}\textbf{x}\big)$. Similar to (\ref{4.7}) and (\ref{4.8}), based on $\textbf{A}^\mathcal{S}\textbf{y}= \bm{0}$ one can verify the relation $\textbf{A}^{\mathcal{S}^c}\textbf{y} = \textbf{A}(\mathcal{S})\textbf{y}^{[d]\setminus [\mathcal{S}]}$, hence it holds that  $\sign\big(\textbf{A}(\mathcal{S})\textbf{x}^{[d]\setminus [\mathcal{S}]}\big)=\sign\big(\textbf{A}(\mathcal{S})\textbf{y}^{[d]\setminus [\mathcal{S}]}\big)$. Now we can invoke $\bm{\mathrm{x}}^{[d]\setminus [\mathcal{S}]} \in \mathcal{W}_{\bm{\mathrm{A}}(\mathcal{S})}$ and conclude that $\textbf{y}^{[d]\setminus [\mathcal{S}]} = t\cdot \textbf{x}^{[d]\setminus [\mathcal{S}]}$ for some $t>0$. By $\textbf{A}^\mathcal{S}\textbf{y}  =\textbf{A}^\mathcal{S}\textbf{x} = \bm{0}$, $\textbf{x}^{[\mathcal{S}]}$ and $\textbf{y}^{[\mathcal{S}]}$ can be uniquely determined by $\textbf{x}^{[d]\setminus [\mathcal{S}]}$ and $\textbf{y}^{[d]\setminus [\mathcal{S}]}$, respectively. Therefore, $\textbf{y} = t\cdot \textbf{x}$, and hence $\textbf{x}\in \mathcal{W}_{\textbf{A}}$. The proof is concluded. 
\end{proof}


Considering a fixed signal $\x$ ($\x \neq \bm{0}$),  the next Lemma shows that when $m\geq 2d-1$, $\mathcal{W}_{\textbf{x}}(m)\cap \mathcal{H}_{\textbf{x}}(m)$ contains a generic $\textbf{A}$ in $\mathbb{C}^{m\times d}$. We let $\bm{\mathrm{e_k}}  $ be the $k$-th column of $\bm{\mathrm{I}_d}$. For   positive integer $l\neq d$, $\bm{\mathrm{e_k}}[l]$ is also used to denote the $k$-th column of  $\bm{\mathrm{I}_{l}}$.

\begin{lem}
    \label{lemmanew6}
    Consider a fixed non-zero signal $\bm{\mathrm{x}}$. When $m\geq 2d-1$, a generic $\bm{\mathrm{A}}$ in $\mathbb{C}^{m\times d}$ can recover $\bm{\mathrm{x}}$ via purely phase-only measurements, i.e., $\bm{\mathrm{x}}\in \mathcal{W}_{\bm{\mathrm{A}}}\cap \mathcal{H}_{\bm{\mathrm{A}}}$. In other words, $\mathcal{W}_{\bm{\mathrm{x}}}(m) \cap \mathcal{H}_{\bm{\mathrm{x}}}(m)$ contains a generic $\bm{\mathrm{A}}$ when $m\geq 2d-1$.
\end{lem}

\begin{proof} By the definition in (\ref{4.3}) and (\ref{4.4}), for any invertible $\textbf{P}\in\mathbb{C}^{d\times d}$ one can easily show
\begin{equation}
    \label{4.9}
    \mathcal{W}_{\textbf{Px}}(m) = \mathcal{W}_{\textbf{x}}(m)\bm{\mathrm{P^{-1}}},~\mathcal{H}_{\textbf{Px}}(m) = \mathcal{H}_{\textbf{x}}(m)\bm{\mathrm{P^{-1}}}.
\end{equation}
Since for each nonzero $\x$ we have $\textbf{Px} = \bm{\mathrm{e_1}}$ for some invertible $\textbf{P}$, we can only consider $\textbf{x} = \bm{\mathrm{e_1}}$. By Lemma \ref{lemma3}, $$\textbf{A}\in \mathcal{W}_{\bm{\mathrm{e_1}}}(m) \cap \mathcal{H}_{\bm{\mathrm{e_1}}}(m)\iff \bm{\mathrm{e_1}}\in \mathcal{W}_{\bm{\mathrm{A}}}\cap \mathcal{H}_{\bm{\mathrm{A}}} \iff \rank\big(\mathcal{D}_{\textbf{A}}(\bm{\mathrm{e_1}})\big) \geq 2d +m -1.$$
Recall (\ref{disD}), obviously entries of  $\mathcal{D}_{\textbf{A}}(\bm{\mathrm{e_1}}) $ are polynomials of the $2md$ real variables $\Re(\textbf{A}),\Im(\textbf{A})$, hence Lemma \ref{lemma4} delivers that $\mathcal{W}_{\bm{\mathrm{e_1}}}(m) \cap \mathcal{H}_{\bm{\mathrm{e_1}}}(m)$ is Zariski open set of $\mathbb{C}^{m\times d}$. It remains to confirm $\mathcal{W}_{\bm{\mathrm{e_1}}}(m) \cap \mathcal{H}_{\bm{\mathrm{e_1}}}(m)\neq \varnothing$ when $m\geq 2d-1$, that is, we need to find $\bm{\mathrm{B}}\in\mathbb{C}^{m\times d}$ such that $\bm{\mathrm{e_1}} \in  \mathcal{W}_{\bm{\mathrm{B}}}\cap \mathcal{H}_{\bm{\mathrm{B}}}.$  We consider 
	\begin{equation}
	\label{constru}
		\bm{\mathrm{B_0}}=\begin{bmatrix}
			1 & \quad &\quad &\quad &\quad \\
			1 & 1 &\quad &\quad &\quad \\
			1 & \quad & 1 &\quad & \quad \\
			\vdots & \ & \ &\ddots & \ \\
			1&\ &\ &\ &1 \\
			1 & \ii&\quad &\quad &\quad \\
			1 & \quad & \ii &\quad & \quad \\
			\vdots & \ & \ &\ddots & \ \\
			1&\ &\ &\ &\ii 
		\end{bmatrix}\in \mathbb{C}^{(2d-1)\times d} ,~~ \textbf{B}=\begin{bmatrix}
			\bm{\mathrm{B_0}}\\ \bm{\mathrm{B_1}}
		\end{bmatrix}\in \mathbb{C}^{m\times d},
	\end{equation}
where $\bm{\mathrm{B_1 }}$ has all ones in its first column, and zeroes as its other entries. Evidently, $\sign(\bm{\mathrm{B_0e_1}}) = \bm{1}$ and so $\bm{\mathrm{e_1}}\in\mathcal{H}_{\textbf{B}}$. We assume $\sign(\bm{\mathrm{By}}) = \sign(\bm{\mathrm{Be_1}})=\mathbf{1}$ for some $\textbf{y} = [y_i]$, then the first measurement gives $y_1>0$. Moreover, the next $2(d-1)$ measurements can imply $y_1+y_k\in \mathbb{R}$, $y_1+\ii y_k\in\mathbb{R}$ for all $2\leq k\leq d$, which yields $(\ii -1)y_k\in\mathbb{R}$, and hence $y_k = 0 $ when $2\leq k\leq d$. Hence, we arrive at $\textbf{y} = y_1 \cdot \bm{\mathrm{e_1}}$, $y_1>0$. Therefore, when $m\geq 2d-1$, $\mathcal{W}_{\bm{\mathrm{e_1}}}(m) \cap \mathcal{H}_{\bm{\mathrm{e_1}}}(m)$ is non-empty Zariski open set. The proof is hence concluded. 
 \end{proof}

With the above lemmas in place, we are now ready to present the proof of   Theorem \ref{T6}.	
	
	\vspace{2mm}
	\noindent
{\it Proof of Theorem \ref{T6}:}	When $m\leq 2d-2$, since $ \mathcal{D}_{\bm{\mathrm{A}}}(\x)\in \mathbb{C}^{2m\times d+m)}$, we have $\rank( \mathcal{D}_{\bm{\mathrm{A}}}(\x))\leq 2m <2d+m-1$. Thus, Lemma \ref{lemma3} gives $\mathcal{W}_{\bm{\mathrm{A}}}\cap \mathcal{H}_{\textbf{A}} =\varnothing $, which further leads to \begin{equation}
	    \mathcal{W}_{\bm{\mathrm{A}}} \subset  \mathbb{C}^d \setminus \mathcal{H}_{\textbf{A}} =\bigcup_{j=1}^m\ker(\bm{\mathrm{\gamma_j^{\top}}}).
  \label{4.10}
	\end{equation}
	Note that (\ref{4.10}) confirms that $\mathcal{W}_{\bm{\mathrm{A}}}$ is nowhere dense (under Euclidean topology) and of zero Lebesgue measure. This immediately gives the lower bound $ \bm{\mathrm{m_{ae}}}(d)\geq 2d-1$.

	When $m \geq 2d-1$, we consider the set of $\textbf{A}$ denoted by $\Xi$ defined as follows: 
	\begin{equation}
	    \begin{aligned}
	        \label{4.11}
	        \Xi = \big\{\textbf{A}\in\mathbb{C}^{m\times d}: \textbf{A}\in \mathcal{W}_{\bm{\mathrm{e_1}}}(m) \cap \mathcal{H}_{\bm{\mathrm{e_1}}}(m);~\forall\mathcal{S}\subset [m],~0<|\mathcal{S}|\leq d,
	        \rank(\bm{\mathrm{A}}^{\mathcal{S}}_{[\mathcal{S}]})=|\mathcal{S}|;&\\\forall\mathcal{S}\subset [m],~0<|\mathcal{S}|<d,~\bm{\mathrm{e_1}}[d-|\mathcal{S}|]\in \mathcal{W}_{\textbf{A}(\mathcal{S})}\cap \mathcal{H}_{\textbf{A}(\mathcal{S})},~\text{see (\ref{4.6}) for } \textbf{A}(\mathcal{S})\big\}.&
	    \end{aligned}
	\end{equation}
In the following, we will  show $\Xi$ contains a generic $\textbf{A}$ in {\it step 1}, then in {\it step 2} we prove each $\textbf{A}$ in $\Xi$ satisfies (\ref{4.5}) that includes the almost everywhere magnitude retrievable property as a special case (i.e., when $\mathcal{S}=\varnothing$). This will give the upper bound $\bm{\mathrm{m_{ae}}}(d)\leq 2d-1$ and finally complete the proof.

\vspace{1mm}

\noindent{\it Step 1.} By Lemma \ref{lemmanew6} $\mathcal{W}_{\bm{\mathrm{e_1}}}(m) \cap \mathcal{H}_{\bm{\mathrm{e_1}}}(m)$ contains a generic $\textbf{A}$. It is evident that a generic $\textbf{A}$ satisfies $\rank(\textbf{A}^\mathcal{S}_{[\mathcal{S}]}) = |\mathcal{S}|$. Thus, we only need to show a generic $\textbf{A}$ satisfies the last property (the second line) in the definition of $\Xi$ (\ref{4.11}). More precisely, for a fixed $\mathcal{S} \subset [m]$, $0<|\mathcal{S}|<d$, we need to prove $\bm{\mathrm{e_1}}[d-|\mathcal{S}|]\in \mathcal{W}_{\textbf{A}(\mathcal{S})}\cap \mathcal{H}_{\textbf{A}(\mathcal{S})}$ holds for a generic $\textbf{A}$. This is indeed similar to the proof of Lemma \ref{lemmanew6}. (\ref{4.6}) gives $\textbf{A}(\mathcal{S})\in \mathbb{C}^{(m-|\mathcal{S})\times (d-|\mathcal{S}|)}$, combining with Lemma \ref{lemma3}, we have \begin{equation}
    \label{4.12}
    \begin{aligned}
       & \bm{\mathrm{e_1}}[d-|\mathcal{S}|]\in \mathcal{W}_{\textbf{A}(\mathcal{S})}\cap \mathcal{H}_{\textbf{A}(\mathcal{S})} \iff\\ &\rank \Big(\mathcal{D}_{\textbf{A}(\mathcal{S})}\big(\bm{\mathrm{e_1}}[d-|\mathcal{S}|]\big)\Big) \geq2(d-|\mathcal{S}|)+(m-|\mathcal{S}|)-1.
    \end{aligned}
\end{equation}
From (\ref{4.6}), one can see entries of $\textbf{A}(\mathcal{S})$ are of the form $\frac{f_{ij}(\textbf{A})}{g_{ij}(\textbf{A})}$ where $f_{ij}(\textbf{A})$, $g_{ij}(\textbf{A})$ are polynomials of $2md$ real variables $[\Re(\textbf{A}),\Im(\textbf{A})]$ with possibly complex coefficients, and so are the entries of $\mathcal{D}_{\textbf{A}(\mathcal{S})}\big(\bm{\mathrm{e_1}}[d-|\mathcal{S}|]\big)$ (check this from (\ref{disD})). Thus, we can invoke Lemma \ref{lemma4} to see the set of $\textbf{A}$ satisfying (\ref{4.12}) is Zariski open. We then show   (\ref{4.12}) holds  for some $\textbf{A}$. To see the existence of such $\textbf{A}$, we set $\textbf{A}^{\mathcal{S}}_{[d]\setminus [\mathcal{S}]}=\bm{0}$, $ \textbf{A}^\mathcal{S}_{[\mathcal{S}]} = \textbf{I}_{|\mathcal{S}|}$, $\textbf{A}_{[\mathcal{S}]}^{\mathcal{S}^c}=\bm{0}$, then   (\ref{4.6}) reads as $\textbf{A}(\mathcal{S}) = \textbf{A}^{\mathcal{S}^c}_{[d]\setminus [\mathcal{S}]}\in \mathbb{C}^{(m-|\mathcal{S})\times (d-|\mathcal{S}|)}$. Note that 
$$ m-|\mathcal{S}| \geq 2d-1-|\mathcal{S}| > 2(d-|\mathcal{S}|)-1,$$
thus we can set $\textbf{A}^{\mathcal{S}^c}_{[d]\setminus [\mathcal{S}]}$ to be a matrix of same form of $\textbf{B}$ in (\ref{constru}), then $\bm{\mathrm{e_1}}[d-|\mathcal{S}|]\in \mathcal{W}_{\textbf{A}(\mathcal{S})}\cap \mathcal{H}_{\textbf{A}(\mathcal{S})}$ follows. Taking a finite intersection over all $\mathcal{S}$, $0<|\mathcal{S}|<d$, it yields that a generic $\textbf{A}$ satisfies the second line of (\ref{4.12}). Therefore, $\Xi$ contains a generic $\textbf{A}$.

\vspace{1mm}

\noindent{\it Step 2.} We aim to show any element of $ \Xi$ satisfies (\ref{4.5}), so we consider a fixed $\textbf{A}\in\Xi$. For a fixed $\mathcal{S}\subset [m]$, we discuss the following three cases. 

\vspace{1mm}

\noindent{\it Case 1.} If $|\mathcal{S}|\geq d$, there exists $\mathcal{S}_0\subset \mathcal{S}$, $|\mathcal{S}_0| = d$. By (\ref{4.11}) we have $d\geq \rank(\textbf{A}^\mathcal{S}) \geq \rank(\textbf{A}^{\mathcal{S}_0})=d$, which gives $\rank(\textbf{A}^\mathcal{S})  = d$ and hence $\ker (\textbf{A}^\mathcal{S}) = \{\bm{0}\}$. Note that $\bm{0} \in \mathcal{W}_{\textbf{A}}$, (\ref{4.5}) holds trivially.

\vspace{1mm}

\noindent{\it Case 2.} If $\mathcal{S} =\varnothing$, (\ref{4.5}) states that $\mathcal{W}_{\textbf{A}}$ contains a generic $\x$ of $\mathbb{C}^d$. Our strategy is still similar to the proof of   Lemma \ref{lemmanew6} and   some arguments in {\it Step 1}, while the difference is that entries of $\mathcal{D}_{\textbf{A}}(\textbf{x})$ are viewed as functions of $\x$. From (\ref{disD}), a simple observation is that entries of $\mathcal{D}_{\textbf{A}}(\textbf{x})$ are 
polynomials (of degree at most 1) of $2d$ real variables $[\Re(\textbf{x})^{\bm{\top}},\Im(\textbf{x})^{\bm{\top}}]^{\bm{\top}}$. Besides, Lemma \ref{lemma3} gives $\mathcal{W}_{\textbf{A}}\cap \mathcal{H}_{\textbf{A}} = \{\textbf{x}: \rank\big(\mathcal{D}_{\textbf{A}}(\textbf{A})\big)\geq 2d+m-1\}$, which is a Zariski open set of $\mathbb{C}^{d}$ due to Lemma \ref{lemma4}. Also, it is non-empty since $\bm{\mathrm{e_1}}\in\mathcal{W}_{\textbf{A}}\cap \mathcal{H}_{\textbf{A}} $ by (\ref{4.11}). Hence, (\ref{4.5}) follows.

\vspace{1mm}

\noindent{\it Case 3.} If $0<|\mathcal{S}|<d$, by exactly the same argument in {\it Case 2} one can see $\mathcal{W}_{\textbf{A}(\mathcal{S})}\cap \mathcal{H}_{\textbf{A}(\mathcal{S})}$ is Zariski open, then $\bm{\mathrm{e_1}}[d-|\mathcal{S}|]\in\mathcal{W}_{\textbf{A}(\mathcal{S})}\cap \mathcal{H}_{\textbf{A}(\mathcal{S})} $ in (\ref{4.11}) can confirm it is non-empty. Thus, $\mathcal{W}_{\textbf{A}(\mathcal{S})}\cap \mathcal{H}_{\textbf{A}(\mathcal{S})}$ contains a generic point of $\mathbb{C}^{d-|\mathcal{S}|}$. We now invoke Lemma \ref{lemma5} to yield (\ref{4.5}). Recall $\rank(\textbf{A}^\mathcal{S}_{[\mathcal{S}]}) = |\mathcal{S}|$, then under the assumption that $\x\in\ker(\textbf{A}^\mathcal{S})$, Lemma \ref{lemma5} gives 
$$ \bm{\mathrm{x}}\in \mathcal{W}_{\bm{\mathrm{A}}}\cap \mathcal{H}_{\bm{\mathrm{A}}}(\mathcal{S})\iff \bm{\mathrm{x}}^{[d]\setminus [\mathcal{S}]} \in \mathcal{W}_{\bm{\mathrm{A}}(\mathcal{S})}\cap \mathcal{H}_{\bm{\mathrm{A}}(\mathcal{S})}.$$
Moreover, the above $\textbf{x}$ and $\bm{\mathrm{x}}^{[d]\setminus [\mathcal{S}]} $ are indeed connected by a linear isomorphism between $\ker(\textbf{A}^\mathcal{S})$ and $\mathbb{C}^{d-|\mathcal{S}|}$ (see the proof of Lemma \ref{lemma5}). Thus,  $ \mathcal{W}_{\bm{\mathrm{A}}} \cap \mathcal{H}_{\textbf{A}}(\mathcal{S})$ contains a generic point  $\ker(\textbf{A}^\mathcal{S})$, which implies (\ref{4.5}). 

Now we can conclude that  that $\bm{\mathrm{m_{ae}}}(d)\leq 2d-1$. Therefore, the minimal measurement number required for almost everywhere magnitude retrieval is $2d-1$.  \hfill $\square$

\begin{rem}
We point out that, Theorem \ref{T6} can be used to interpret some algorithms or numerical results   in previous works. For instance, the theoretical results in \cite{hayes1980signal,oppenheim1980iterative} guarantee that   $d-1$ Fourier phases are sufficient for reconstruction of $\bm{\mathrm{x}}\in\mathbb{R}^d$. However,  their iterative algorithm requires $2d$ Fourier phases, which seems a bit strange compared to their theoretical results. 
In fact, this is because their iterative algorithm does not utilize  the fact that $\bm{\mathrm{x}} $ is real-valued (see Figure 1 in \cite{hayes1980signal}), hence $\bm{\mathrm{x}}$ is treated as  complex-valued signal and requires $2d-1$ measurements. Moreover, the authors 
of \cite{kishore2020phasesense} compared their algorithm with MagnitudeCut in \cite{wu2016phase} 
and randomly generated $\bm{\mathrm{x}}\in \mathbb{C}^{64}$ and $\bm{\mathrm{A}}\in \mathbb{C}^{m\times 64}$ from Gaussian distribution. 
Consistent with Theorem \ref{T6}, both algorithms achieve successful reconstruction when $m\geq 128$, see Figure 3(a) 
in \cite{kishore2020phasesense}.
\label{remark3} 
\end{rem}

	\section{Two related new results}\label{sec5}
	
In this section, we exploit the previous theoretical framework 
to derive two related new results in phase-only reconstruction problems.

\subsection{Symmetric signal reconstruction}
	
	We first give a proposition concerning the ill-posedness of   symmetric signal reconstruction from Fourier phase. Indeed, this issue has been noticed in many early works. For instance, the uniqueness criteria described in  \cite{hayes1980signal,hayes1982reconstruction}   exclude   symmetric signals. Also,   \cite{levi1983signal} reported the failure in   recovering     symmetric signal from Fourier phase as an experimental result. In addition, the POCS algorithm    performances the worst under images of symmetric form \cite{urieli1998optimal}.

	Note that these previous works only considered $\textbf{x}\in\mathbb{R}^d$ and lacked rigorous argument on this issue of symmetric signal. Moreover, they assumed the $z$-transform of $\textbf{x}$ has no zeros on the unit circle, which evidently satisfies $\textbf{x}\in \mathcal{H}_{\textbf{A}}$ for Fourier measurement matrix $\textbf{A}$. Thus, the case when $\bm{x}\notin \mathcal{H}_{\textbf{A}}$ remains unclear.

	In the following, we consider $\textbf{x}\in \mathbb{C}^d$ that is conjugate symmetric. By using the discriminant matrix $\mathcal{D}_{\textbf{A}}(\textbf{x})$, we can precisely   present the ill-posedness of symmetric signal recovery from phase with a rigorous proof. Note that our result allows some Fourier measurement of $\textbf{x}$ to vanish (i.e., $\textbf{x}\notin \mathcal{H}_{\textbf{A}}$). 
	\begin{pro}
	\label{pro2}
	We consider a conjugate symmetric $\bm{\mathrm{x}}=\big[x_k\big]_{k=1-d}^{d -1}$ of odd length satisfying $x_{-k} = \overline{x_k}$, the $j$-th row of the measurement matrix $\bm{\mathrm{A}}\in \mathbb{C}^{m\times (2d-1)}$ is the Fourier measurement under frequency $\omega_j$, more precisely, it is given by 
	\begin{equation}
	    \nonumber
	    \bm{\mathrm{\gamma_j^\top}} = \begin{bmatrix}
	          e^{\bm{\mathrm{i}}(d-1)\omega_j} & \cdots & e^{\bm{\mathrm{i}}\omega_j} & 1 & e^{-\bm{\mathrm{i}}\omega_j} & \cdots &   e^{-\bm{\mathrm{i}}(d-1)\omega_j}
	    \end{bmatrix}.
	\end{equation}
	If $m-|\mathrm{N}(\bm{\mathrm{Ax}})| < 2d-2$, $\bm{\mathrm{x}}\notin \mathcal{W}_{\bm{\mathrm{A}}}$.
	\end{pro}


\begin{proof}
	 We apply the discriminant matrix $\mathcal{D}_{\textbf{A}}(\textbf{x})$ (see (\ref{disD})) to prove the claim. Note that $x_0 \in \mathbb{R}$, then some simple algebra  shows the Fourier measurement of the conjugate symmetric $\x$ is real:
	 \begin{equation}
	 \label{5.1}
	   \bm{\mathrm{\gamma_j^\top x}} =   \sum_{k=1-d}^{d-1} e^{-\textbf{i}k \omega_j}x_k = x_0 + \sum_{k=1}^{d-1}e^{-\textbf{i}k \omega_j}x_k + \sum_{k=1-d}^{-1} e^{-\textbf{i}k \omega_j}\overline{x_{-k}} = x_0 + 2\sum_{k=1}^{d-1}\Re\big(e^{-\textbf{i}k \omega_j}x_k\big)\in\mathbb{R}.
	 \end{equation}
	 Thus, $\Re\big(\mathrm{dg}(\Ax)\big) = \mathrm{dg}(\Ax)$ and $\Im\big(\mathrm{dg}(\Ax)\big) = \bm{0}$. Furthermore, for $\textbf{B} = [\bm{\mathrm{b_1}},\bm{\mathrm{b_2}},\cdots, \bm{\mathrm{b_p}}]$ where $\bm{\mathrm{b_k}}$ denotes the $k$-th column, we define $\textbf{B}_{fl} := [\bm{\mathrm{b_p}},\cdots,\bm{\mathrm{b_2}},\bm{\mathrm{b_1}}]$ as the matrix obtained by flipping the columns of $\textbf{B}$.  By (\ref{disD}) we have
	 \begin{equation}
	     \label{5.2}
	     \mathcal{D}_{\textbf{A}}(\x) = \begin{bmatrix}
	          \textbf{R}_{fl} & \bm{1}_{m\times 1} &  \textbf{R} & \textbf{I}_{fl} & \bm{0}_{m\times 1} & -\textbf{I} & \mathrm{dg}(\textbf{Ax})\\
	         - \textbf{I}_{fl} & \bm{0}_{m\times 1} & \textbf{I} & \textbf{R}_{fl} & \bm{1}_{m\times 1} &  \textbf{R} & \bm{0}_{m\times m}
	     \end{bmatrix}\in\mathbb{R}^{(2m)\times (4d-2+m)},
	 \end{equation}
	  where $\textbf{I} = [\sin(k\omega_j)]_{(j,k)\in [m]\times [d-1]}$, $\textbf{R} = [\cos(k\omega_j)]_{(j,k)\in [m]\times [d-1]}$ are $m\times (d-1)$ matrices. By Theorem \ref{T1}, $\x\in\mathcal{W}_{
	 \textbf{A}}$ if and only if $\rank(  \mathcal{D}_{\textbf{A}}(\x) ) = 2\cdot (2d-1)+ |\mathrm{N}(\textbf{Ax})|-1$, hence we only need to show $\rank(  \mathcal{D}_{\textbf{A}}(\x) ) < 4d-3 + |\mathrm{N}(\textbf{Ax})|$. To this end, we apply elementary operations to simplify $ \mathcal{D}_{\textbf{A}}(\x) $. We first deal with the six blocks on the left of (\ref{5.2}) as follows:
	 \begin{equation}
	     \nonumber
	     \begin{bmatrix}
	          \textbf{R}_{fl} & \bm{1}_{m\times 1}  &  \textbf{R}\\
	           - \textbf{I}_{fl} & \bm{0}_{m\times 1}  & \textbf{I}
	     \end{bmatrix}\xrightarrow[~]{(i)}  \begin{bmatrix}
	          \textbf{R}  & \bm{1}_{m\times 1}  &  \textbf{R}\\
	           - \textbf{I} & \bm{0}_{m\times 1}  & \textbf{I}
	     \end{bmatrix}\xrightarrow[~]{(ii)}\begin{bmatrix}
	          \textbf{R}  & \bm{1}_{m\times 1}  &  \textbf{R}\\
	           \bm{0} & \bm{0}_{m\times 1}  & \textbf{I}
	     \end{bmatrix} \xrightarrow[~]{(iii)}\begin{bmatrix}
	          \textbf{R}  & \bm{1}_{m\times 1}  &  \bm{0}\\
	           \bm{0} & \bm{0}_{m\times 1}  & \textbf{I}
	     \end{bmatrix} ,
	 \end{equation}
	 where $(i)$ flips the $(1,1),(2,1)$-th block, $(ii)$ adds the third column to the first column,   then multiplies the first column by $\frac{1}{2}$, finally $(iii)$ uses the $(1,1)$-th block to eliminate the $(1,3)$-th block. We can deal with the $(i,j)$-th block, $1\leq i\leq 2,4\leq j\leq 6$ in (\ref{5.2}) similarly, and then $\mathcal{D}_{\textbf{A}}(\textbf{x})$ is transformed to be  
	$$ \begin{bmatrix}
	     \textbf{R} & \bm{1}_{m\times 1} & \bm{0}_{m\times (d-1)} & \textbf{I} & \bm{0}_{m\times 1} & \bm{0}_{m\times (d-1)} & \mathrm{dg}(\textbf{Ax}) \\
	     \bm{0}_{m\times (d-1)} & \bm{0}_{m\times 1}  & \textbf{I} & \bm{0}_{m\times (d-1)} & \bm{1}_{m\times 1} & \textbf{R} & \bm{0}_{m\times m}
	\end{bmatrix}.$$
	We let $\bm{\mathrm{D_1}} = [\bm{\mathrm{D_2}},\mathrm{dg}(\textbf{Ax}) ]$,  $\bm{\mathrm{D_2}} = [\textbf{I},\bm{1}_{m\times 1},\textbf{R}]$, then $\rank \big(\mathcal{D}_{\textbf{A}}(\textbf{x})\big) = \rank(\bm{\mathrm{D_1}})+\rank(\bm{\mathrm{D_2}})$ holds. Moreover,
we let	$\mathrm{N}(\textbf{Ax}) = \mathcal{T}$, then it is not hard to see $\rank (\bm{\mathrm{D_1}})= \rank ((\bm{\mathrm{D_2}})^{[m]\setminus \mathcal{T}})+ |\mathcal{T}|$. Thus,
	by putting pieces together, 
	it yields   
	\begin{equation}
	    \nonumber
	    \begin{aligned}
	        &\rank\big(\mathcal{D}_{\textbf{A}}(\textbf{x})\big)=\rank((\bm{\mathrm{D_2}})^{[m]\setminus \mathcal{T}})+|\mathcal{T}|+\rank(\bm{\mathrm{D_2}}) \\
	        & \leq \min \big\{m- |\mathcal{T}|, 2d-1\big\}+ |\mathcal{T}| + \min\big\{m,2d-1\big\}
	        \\ & < 2d-2+|\mathrm{N}(\textbf{Ax})| + 2d-1 = 4d-3+|\mathrm{N}(\textbf{Ax})| .
	    \end{aligned}
	\end{equation}
	Note that in the last line   we invoke the assumption $m-|\mathrm{N}(\textbf{Ax})| <2d-2$. The proof is concluded. 
 \end{proof}


 \begin{rem}
     \label{remsymme}
     When $\mathrm{N}(\bm{\mathrm{{Ax}}}) = m$, for each $j\in [m]$ either $\bm{\mathrm{\gamma_j^\top x}}>0$ or $\bm{\mathrm{\gamma_j^\top x}}<0$ holds (by (\ref{5.1})). In this case, the intuition is that any conjugate symmetric $\bm{\mathrm{y}} = [y_k]_{k=1-d}^{d-1}$ sufficiently close to $\bm{\mathrm{x}}$ possesses the same Fourier phases (by (\ref{5.1}) again), which directly yields $\textbf{x}\notin\mathcal{W}_{\textbf{A}}$. However, the result is non-trivial when zero measurement occurs. To justify the condition $m-|\mathrm{N}(
     \bm{\mathrm{Ax}})|<2d-2$, We point out that the reconstruction becomes possible when $m- |\mathrm{N}(\bm{\mathrm{Ax}}) |\geq 2d-2$. Indeed, under the mild condition that $\{\omega_j\}\subset (0,\pi)$ are mutually different, we have $\rank \big(\bm{\mathrm{{A}}}^{[m]\setminus \mathrm{N}(\bm{\mathrm{Ax}})}\big) \geq 2d-2$, hence   $\bm{\mathrm{{A}}}^{[m]\setminus \mathrm{N}(\bm{\mathrm{Ax}})}\bm{\mathrm{x}}= \bm{0}$    restricts $\bm{\mathrm{x}}$ to a linear subspace of $\mathbb{C}^{2d-1}$ with dimension at most $1$. Thus, one additional phase-only measurement (that is non-zero) can uniquely specify $\bm{\mathrm{x}}$ up to a positive scaling factor.
 \end{rem}
Likewise, similar result can be established for conjugate symmetric signal of even length. The details are left to avid readers.	
	\subsection{Selection of measurements}
	
	In this subsection, we use the discriminant matrix $\mathcal{E}_{\textbf{A}}(\textbf{x})$ to show an interesting property of the phase-only reconstruction problem. Our result  guarantees that, if $m$ $(m\geq 2d-1)$ phase-only measurements can uniquely specify a signal $\x$, then one can always select $2d-1$  measurements for reconstruction of $\x$ up to a positive scaling. To our best knowledge, there exists no previous result of this kind for phase-only reconstruction. 
	
	\begin{theorem}                          		Assume $m\geq 2d-1$, $\bm{\mathrm{x}}\in \mathcal{W}_{\bm{\mathrm{A}}}$. Then there exists $\mathcal{S}\subset [m]$ with size $|\mathcal{S}|=2d-1$, such that $\bm{\mathrm{x}}\in \mathcal{W}_{\bm{\mathrm{A}}^{\mathcal{S}}}$. On the other hand,  it is possible that $\bm{\mathrm{x}}\notin \mathcal{W}_{\bm{\mathrm{A}}^{\mathcal{S}}}$holds for all $ \mathcal{S}\subset [m]$ with size $|\mathcal{S}|=2d-2$.
		\label{T8}
	\end{theorem}
	
\begin{proof}	Note that for invertible $\textbf{P}\in \mathbb{C}^{d\times d}$, $\x\in \mathcal{W}_{\textbf{A}^{\mathcal{S}}}$ if and only if $\bm{\mathrm{P^{-1}x}}\in \mathcal{W}_{\bm{\mathrm{A^{\mathcal{S}}P}}}$, and $\x\in\mathcal{W}_{\textbf{A}}$ indicates $\rank(\textbf{A})=d$, hence $\textbf{AP} = [\bm{\mathrm{I_d}},\bm{\mathrm{A_1^\top}}]^{\bm{\top}}$ for some invertible $\textbf{P}$. Therefore, without losing generality we can assume   $\textbf{A} = [\bm{\mathrm{I_d}},\bm{\mathrm{A_1^\top}}]^{\bm{\top}}$ ($\bm{\mathrm{A_1}}\in\mathbb{C}^{(m-d)\times d}$).

	This allows us to apply $\mathcal{E}_{\textbf{A}}(\textbf{x})$. Recall its construction (\ref{2.9}), we have
	$$\mathcal{E}_{\textbf{A}}(\x)=\begin{bmatrix}
		[\Psi_{d+1}(\x)]_{\mathrm{N}(\x)}\\
		[\Psi_{d+2}(\x)]_{\mathrm{N}(\x)}\\
		\vdots \\
		[\Psi_{m}(\x)]_{\mathrm{N}(\x)}
	\end{bmatrix}\in \mathbb{R}^{\hat{m}\times |\mathrm{N}(\textbf{x})|}$$
	for some $\hat{m}\geq m-d$. By Theorem \ref{T2}, $\x\in \mathcal{W}_{\bm{\mathrm{A}}}$ if and only if  $\rank(\mathcal{E}_{\textbf{A}}(\x))=|\mathrm{N}(\x)|-1$. Hence, there exists $\mathcal{S}_0\subset [\hat{m}]$, $|\mathcal{S}_0|=|\mathrm{N}(\x)|-1$, such that $\rank([\mathcal{E}_{\textbf{A}}(\x)]^{\mathcal{S}_0})=|\mathrm{N}(\x)|-1$. By (\ref{2.7}), (\ref{2.8}), each block $[\Psi_j (\textbf{x})]_{\mathrm{N}(\x)}$ has one or two rows, and  now we consider  
	\begin{equation}
	\label{5.3}
	    \mathcal{J}=\{d+1\leq j\leq m: \mathrm{at}\ \mathrm{least}\ 1\ \mathrm{row}\ \mathrm{of} \ [\Psi_j(\x)]_{\mathrm{N}(\x)}\ \mathrm{appears}\ \mathrm{in} \ [\mathcal{E}_{\textbf{A}}(\x)]^{\mathcal{S}_0}\}.
	\end{equation}
	Then evidently, $|\mathcal{J}|\leq |\mathcal{S}_0|=|\mathrm{N}(\x)|-1$.
	
	Furthermore, we can consider $\textbf{A}$'s submatrix $\textbf{A}^{[d]\cup \mathcal{J}}$. Simple observation confirms that $\mathcal{E}_{\textbf{A}^{[d]\cup \mathcal{J}}}(\x)$ is a submatrix of $\mathcal{E}_{\textbf{A}}(\x)$, and by (\ref{5.3}) $[\mathcal{E}_{\textbf{A}}(\x)]^{\mathcal{S}_0}$ is a submatrix of $\mathcal{E}_{\textbf{A}^{[d]\cup \mathcal{J}}}(\x)$.	Thus, it holds that
	$$|\mathrm{N}(\x)|-1=\rank(\mathcal{E}_{\textbf{A}}(\x))\geq \rank(\mathcal{E}_{\textbf{A}^{[d]\cup \mathcal{J}}}(\x))\geq \rank([\mathcal{E}_{\textbf{A}}(\x)]^{\mathcal{S}_0})=|\mathrm{N}(\x)|-1,$$
	which gives $\x\in \mathcal{W}_{\textbf{A}^{[d]\cup \mathcal{J}}}$. By further noting $$|[d]\cup \mathcal{J}|=d+|\mathcal{J}|\leq d+|\mathrm{N}(\x)|-1\leq 2d-1,$$
	we can find $\mathcal{S}\subset [m]$, such that $[d]\cup \mathcal{J}\subset \mathcal{S}$ and $|\mathcal{S}|=2d-1$, then $\x\in \mathcal{W}_{\textbf{A}^\mathcal{S}}$. This displays the first statement of the Theorem.

	It remains to show the possibility of $\textbf{x}\notin \mathcal{W}_{\textbf{A}^\mathcal{S}}$ for all $ |\mathcal{S}|= 2d-2$.
	First we invoke Theorem \ref{T6}, it gives that
for some $ \textbf{A}\in \mathbb{C}^{(2d-1)\times d}$, $\mathcal{W}_{\bm{\mathrm{A}}}$ contains a generic point of $\mathbb{C}^d$. This obviously implies $\mathcal{W}_{\bm{\mathrm{A}}}$ has dense interior (under Euclidean topology). Theorem \ref{T6} also delivers that when $|\mathcal{S}|=2d-2$, $\mathcal{W}_{\textbf{A}^{\mathcal{S}}}$ is nowhere dense (under Euclidean topology), and hence $\cup_{|\mathcal{S}|=2d-2}\mathcal{W}_{\textbf{A}^{\mathcal{S}}}$ is nowhere dense (under Euclidean topology). Thus, it displays   $ \mathcal{W}_{\bm{\mathrm{A}}} \setminus \bigcup_{|\mathcal{S}|=2d-2}\mathcal{W}_{\textbf{A}^{\mathcal{S}}}\neq \varnothing ,
	$  which concludes the proof. 
	\end{proof}

Theorem \ref{T8} sheds some light on the structure of $\mathcal{W}_{\bm{\mathrm{A}}}$. Specifically, it implies the relation $ \mathcal{W}_{\bm{\mathrm{A}}}=\bigcup_{|\mathcal{S}|=2d-1}\mathcal{W}_{\bm{\mathrm{A}}^\mathcal{S}}$ that may be useful for future study. 
	 
\begin{rem}
	In essence, the mathematical part of phase-only reconstruction problem is solving a phase-only system $\sign(\bm{\mathrm{{Ax}}}) = \bm{\mathrm{{b}}}$. Interestingly, many of our results  are reminiscent of the properties of linear system $\bm{\mathrm{Ax}} = \bm{\mathrm{b}}$. Specifically, the measurement number $2d-1$ for $\sign(\bm{\mathrm{{Ax}}}) = \bm{\mathrm{{b}}}$  seems to act similarly to the measurement number $d$ for  $\bm{\mathrm{Ax}} = \bm{\mathrm{b}}$. For example, when $m< 2d-1$ $\sign(\bm{\mathrm{{Ax}}})=\sign(\bm{\mathrm{{Ax_0}}})$ has more than one solution   for almost all $\bm{\mathrm{x_0}}$\footnote{If $\sign(\textbf{Ay})=\sign(\textbf{Ax}_0)$, then we view $\{t\cdot \textbf{y}:t>0\}$ as one solution of  $\sign(\textbf{Ax})=\sign(\textbf{Ax}_0)$.}; When $m\geq 2d-1$, however,  $\sign(\bm{\mathrm{{Ax}}})=\sign(\bm{\mathrm{{Ax_0}}})$ has a unique solution for a generic $\bm{\mathrm{A}}$ and a generic $\bm{\mathrm{x_0}}$, see Theorem \ref{T6}. For linear system, when $m< d$, $\bm{\mathrm{Ax}} = \bm{\mathrm{Ax_0}}$ always has infinite solutions, whereas for a generic $\bm{\mathrm{A}}$ with $m\geq d$ (those of full column rank), $\bm{\mathrm{Ax}} = \bm{\mathrm{Ax_0}}$ has unique solution $\bm{\mathrm{x_0}}$ for all $\bm{\mathrm{x_0}}$. Moreover, Theorem \ref{T8} can be viewed as the counterpart of the fact that $\rank(\bm{\mathrm{A}}) = d$ implies $\rank(\bm{\mathrm{A}}^\mathcal{S}) = d$ for some $|\mathcal{S}|=d$. These similar properties are perhaps due to the linearized nature of the phase-only system (e.g., see (\ref{eq1})).
	\label{remark4} 
\end{rem}

	\section{Comparison with previous works}
	\label{compare}

	The main aim of this section is to compare our results with some related works.
	
	\subsection{Previous uniqueness conditions}
	Recall that two present Theorems \ref{T1}-\ref{T2} give two necessary and sufficient uniqueness conditions, which to the best of our  knowledge are the first uniqueness results applicable to general measurement matrix $\A$ and complex-valued $\textbf{x}$. However, there  have been some uniqueness criteria for the special case of recovering $\x\in\mathbb{R}^d$ from the Fourier phase \cite{hayes1980signal,hayes1982reconstruction,oppenheim1981importance,ma1991novel,porat1999signal}. 
	
	Here, we give
a brief review of previous uniqueness results.
We note  that these previous results only apply to real-valued signal $\textbf{x} $  and  
	the Fourier measurement matrix $\textbf{A}$, whose $j$-th row with frequency $\omega_j$ is given by \begin{equation}
	      \bm{\mathrm{\gamma_j^\top}} = [e^{-\textbf{i}\omega_j},\cdots,e^{-\textbf{i}d\omega_j}],
	      \label{6.1}
	  \end{equation}  
	  The earliest uniqueness criterion is the so-called minimum-phase or maximum-phase condition (see \cite{oppenheim1981importance,quatieri1981iterative} for instance), which is rather restrictive and hence not included here. Later, a set of more   relaxed conditions that can accommodate most signals was proposed in \cite{hayes1980signal}, and then extended to multi-dimensional signals in \cite{hayes1982reconstruction}. 
When $\textbf{x}\in \mathbb{R}^d$ is considered,   their uniqueness condition is given in Condition \ref{condition1} stated below. In \cite{ma1991novel}, Ma proposed a new condition (see Condition \ref{condition2}   below) that   hinges on the non-singularity of a signal matrix. 

\begin{con}
    \label{condition1}
   The signal $\bm{\mathrm{x}}= [x_k]\in\mathbb{R}^d$ with $x_1\neq 0$ has  a z-transform that   does not have any zero in reciprocal pair or on the unit circle.  
\end{con}
	\begin{con}
	    \label{condition2}
	    The signal $\bm{\mathrm{x}}= [x_k]\in\mathbb{R}^d$ with $x_1\neq 0$ satisfies $\rank\big(\mathcal{B}(\bm{\mathrm{x}})\big) = d-1$, or equivalently $\mathcal{B}(\bm{\mathrm{x}})$ is invertible, where    $\mathcal{B}(\bm{\mathrm{x}})$  is defined to be \begin{equation}
	        \label{6.3}
	        \mathcal{B}(\bm{\mathrm{x}}) = \begin{bmatrix}
	             x_1 & 0 & \cdots & 0 & 0\\
	             x_2 & x_1 & \cdots & 0 & 0 \\
	             \vdots & \vdots & \ddots & \vdots & \vdots \\
	             x_{d-2} & x_{d-3}& \cdots & x_1 & 0 \\
	             x_{d-1} & x_{d-2} & \cdots & x_2 & x_1
	        \end{bmatrix} - \begin{bmatrix}
	             x_{3~} & x_{4~} & \cdots & x_d & 0 \\
	             x_{4~} & x_{5~} & \cdots & 0 & 0 \\
	             \vdots & \vdots & \cdot^{\cdot^{ {\cdot}}} & \vdots & \vdots\\
	             x_{d~} & 0_{~} & \cdots & 0 & 0 \\
	             0_{~} & 0_{~} & \cdots & 0 & 0
	        \end{bmatrix}.
	    \end{equation} 
Note that	$\mathcal{B}(\bm{\mathrm{x}})$ is of 
Toeplitz-minus-Hankel form. 		
	\end{con}
		We consider real-valued signal $\x\in\mathbb{R}^d$. Given $\A\in\mathbb{C}^{m\times d}$,   the set of  recoverable signals can be given by
	  \begin{equation}
	      \label{6.2}
	      \mathcal{W}_{\textbf{A},\mathbb{R}}= \big\{\x\in\mathbb{R}^d:\sign(\textbf{Ay}) = \sign(\textbf{Ax}),~\textbf{y}\in\mathbb{R}^d~\mathrm{implies}~\textbf{y} = t\cdot \textbf{x}~\text{for some }t>0 \big\}.
	  \end{equation}
	  If  more than  $d-1$ Fourier phases are sampled with mutually different frequencies, and $\textbf{Ax}$   contains no zeros, then  Theorem 5 in \cite{hayes1980signal} guarantees that, $\textbf{x}$ satisfying Condition \ref{condition1} belongs to $ \mathcal{W}_{\textbf{A},\mathbb{R}}(\textbf{x})$. For $\textbf{x}$ with $x_1\neq 0$, the main result in \cite{ma1991novel} states that Condition \ref{condition2} is  necessary and sufficient for $\textbf{x} \in\mathcal{W}_{\textbf{A},\mathbb{R}}(\textbf{x})$.

	Obviously, Conditions \ref{condition1}-\ref{condition2} only apply to the Fourier measurement matrix $\textbf{A}$, and it is unclear whether they can be generalized to complex-valued signal $\x\in\mathbb{C}^d$. Although our Theorems \ref{T1}-\ref{T2} are not directly applicable to their real-valued signal setting (due to the additional priori $\x\in\mathbb{R}^d$),  one can readily establish the  conditions  for $\x\in \mathcal{W}_{\textbf{A},\mathbb{R}}$ by   techniques similar to those in Theorems \ref{T1}-\ref{T2}. 
	These conditions will be presented in   Theorems 
	\ref{realcasedisd}-\ref{realdise}, with 
	their proofs   deferred to  Appendix \ref{appen2}.

	We assume that $\rank\big([\Re(\textbf{A})^\top,\Im(\textbf{A})^\top]^\top\big) = d$, otherwise we would have $\bm{\mathrm{Ax_0}}=\bm{0}$ for some non-zero $\bm{\mathrm{x_0}}\in \mathbb{R}^d$, which leads to $\mathcal{W}_{\textbf{A},\mathbb{R}} = \varnothing$. Given $\textbf{A}$ and real-valued signal $\x$,  our first uniqueness condition involves  the discriminant matrix $\mathcal{D}_{\textbf{A},\mathbb{R}}(\textbf{x})$ defined as
	\begin{equation}
	    \label{6.4}
	    \mathcal{D}_{\textbf{A},\mathbb{R}}(\textbf{x}) = \begin{bmatrix}
	         \Re\big(\textbf{A}\big) & \Re\big(\mathrm{dg}(\textbf{Ax})\big)\\
	       \Im\big(\textbf{A}\big)  & \Im\big(\mathrm{dg}(\textbf{Ax})\big)
	    \end{bmatrix}\in \mathbb{R}^{2m \times (m+d)},
	\end{equation}
	\begin{theorem}
	    \label{realcasedisd}
	    For $\bm{\mathrm{A}}\in\mathbb{C}^{m\times d}$ we assume that $ [\Re\bm{\mathrm{(A)^\top}},\Im\bm{\mathrm{(A)^\top}}]^{\bm{\top}}$ has full column rank and that $\bm{\mathrm{x}}\in\mathbb{R}^d$ is non-zero. Then $\bm{\mathrm{x}}\in\mathcal{W}_{\bm{\mathrm{A}},\mathbb{R}}$ if and only if $\rank\big(\mathcal{D}_{\bm{\mathrm{A}},\mathbb{R}}(\bm{\mathrm{x}})\big)=d+|\mathrm{N}(\bm{\mathrm{Ax}})|-1$.
	\end{theorem}
	
	Recall the entry-wise notation $\textbf{A} = [r_{jk}\cdot e^{\textbf{i}\theta_{jk}}]$ and the $j$-th row $\bm{\mathrm{\gamma_j^\top}}$. Similar to the idea of $\mathcal{E}_{\textbf{A}}(\textbf{x})$, a different discriminant matrix can be defined. For $j\in [m]$ such that $\bm{\mathrm{\gamma_j^\top x}}\neq 0$, we let $e^{\textbf{i}\delta_j} = \sign(\bm{\mathrm{\gamma_j^\top x}})$ and define \begin{equation}
    \label{6.5}
    	\Psi_{j,\mathbb{R}}(\x):=\begin{bmatrix}
		r_{j1}\cdot\sin(\theta_{j1} -\delta_j) & r_{j2}\cdot\sin(\theta_{j2} -\delta_j)&\cdots&r_{jd}\cdot\sin(\theta_{jd} -\delta_j)
	\end{bmatrix}.
\end{equation}
	If $\bm{\mathrm{\gamma_j^\top x}}= 0$, we define \begin{equation}
    \label{6.6}
    	\Psi_{j,\mathbb{R}}(\x):=\begin{bmatrix}
		r_{j1}\cdot\sin(\theta_{j1} )&r_{j2}\cdot\sin(\theta_{j2} )&\cdots &r_{jd}\cdot\sin(\theta_{jd})\\
		r_{j1}\cdot\cos(\theta_{j1})&r_{j2}\cdot\cos(\theta_{j2})&\cdots&r_{jd}\cdot\cos(\theta_{jd})
	\end{bmatrix} = \begin{bmatrix}
	     \Re\big(\bm{\mathrm{\gamma_j^\top}}\big)\\\Im\big(\bm{\mathrm{\gamma_j^\top}}\big)
	\end{bmatrix}.
\end{equation}
	Then we stack these matrices to obtain 
	\begin{equation}
	    \label{6.7}
	    \mathcal{E}_{\textbf{A},\mathbb{R}}(\textbf{x}) = \begin{bmatrix}
	         	\Psi_{ 1,\mathbb{R}}(\x)\\
		\vdots\\
		\Psi_{m,\mathbb{R}}(\x)
	    \end{bmatrix}.
	\end{equation}
	The following result characterizes $\x\in \mathcal{W}_{\textbf{A},\mathbb{R}}$   via $\rank \big(  \mathcal{E}_{\textbf{A},\mathbb{R}}(\textbf{x})\big)$.
	\begin{theorem}
	    \label{realdise}
	      For $\bm{\mathrm{A}}\in\mathbb{C}^{m\times d}$ we assume that $ [\Re\bm{\mathrm{ (A)^\top}},\Im\bm{\mathrm{ (A)^\top}}]^{\bm{\top}}$ has full column rank, $\bm{\mathrm{x}}\in\mathbb{R}^d$ is non-zero, and $\bm{\mathrm{\gamma_j^\top x}}\neq 0$ for some $j\in [m]$. Then $\bm{\mathrm{x}}\in\mathcal{W}_{\bm{\mathrm{A}},\mathbb{R}}$ if and only if $\rank \big(  \mathcal{E}_{\bm{\mathrm{A}},\mathbb{R}}(\bm{\mathrm{x}})\big) = d-1$.
	\end{theorem}
	
	With the discriminant matrices $\mathcal{D}_{\textbf{A},\mathbb{R}}(\textbf{x})$ and $\mathcal{E}_{\textbf{A},\mathbb{R}}(\textbf{x})$, one can further explore the minimal measurement number or other interesting properties, but here we simply focus on the comparison between our Theorems \ref{realcasedisd}-\ref{realdise} and the previous Conditions \ref{condition1}-\ref{condition2}. 
 Specifically, our uniqueness conditions can be specialized to Fourier measurement matrix and hence readily encompass the special case studied in \cite{hayes1980signal,ma1991novel}. 
	Compared with Condition \ref{condition1},   
		our results exhibit two significant advantages. Firstly, our results are not only sufficient but also necessary, whereas Condition \ref{condition1} is only sufficient and can  exclude some signals of interest. 
	Secondly, 
	solving a polynomial equation can take essentially more efforts than calculating the matrix rank, and so our uniqueness conditions are more practically appealing. 
	Compared with Condition \ref{condition2} in \cite{ma1991novel}, our uniqueness criteria merit the generality of $\textbf{A}$ and can exactly recover Condition \ref{condition2}. This can be done by specializing our Theorem \ref{realdise} to Fourier measurement matrix $\textbf{A}$, see the next Proposition. 
	
	\begin{pro}
	\label{recover}
	Assume  $\bm{\mathrm{A}}\in\mathbb{R}^{m\times d}~(m\geq d-1)$ are the Fourier measurement matrix with rows given by (\ref{6.1}), and $\omega_1,\cdots,\omega_{d-1}\in (0,\pi)$ are mutually different. For non-zero $\bm{\mathrm{x}}\in\mathbb{R}^d$ such that $\bm{\mathrm{Ax}}$   contains no zero entries, our uniqueness condition $\rank\big(\mathcal{E}_{\bm{\mathrm{A}},\mathbb{R}}(\bm{\mathrm{x}})\big) = d-1$ is equivalent to $\rank\big(\mathcal{B}_{f}(\bm{\mathrm{x}})\big) = d-1$, where 
	\begin{equation}
	    \label{6.8}
	    \mathcal{B}_{f}(\textbf{x}) = \begin{bmatrix}
	         0 & x_1 & x_2 &\cdots & x_{d-2} & x_{d-1} \\
	         0 & 0 & x_1 &\cdots & x_{d-3} & x_{d-2}\\
	         \vdots & \vdots & \vdots & \ddots & \vdots & \vdots \\
	         0 & 0 & 0 & \cdots & x_1 & x_2 \\
	         0 & 0 & 0 & \cdots & 0 & x_1
	    \end{bmatrix}-\begin{bmatrix}
	         x_2 & x_3 & \cdots & x_{d-1}& x_d & 0 \\
	         x_3 & x_4 & \cdots & x_d & 0 & 0 \\
	         \vdots & \vdots & \cdot^{\cdot^{ {\cdot}}} & \vdots & \vdots & \vdots \\
	         x_{d-2} & x_d&\cdots  & 0 & 0 & 0 \\
	         x_d & 0 &\cdots &0 & 0 & 0 
	    \end{bmatrix} .
	\end{equation}
	     Moreover, if $x_1 \neq 0$, this is   equivalent to $\rank\big(\mathcal{B}(\bm{\mathrm{x}})\big) = d-1$ (see  (\ref{6.3}) for the definition of $\mathcal{B}(\bm{\mathrm{x}})$), hence the uniqueness criterion in Condition \ref{condition2} is recovered.  
	\end{pro}

	\begin{proof}
	We assume the Fourier measurement matrix $\textbf{A} = [e^{-\textbf{i}k\omega_j}]_{j\in [m],k\in [d]}$. Because $\textbf{Ax}$   contains no  zero entries, from (\ref{6.5}) and (\ref{6.7}), $\mathcal{E}_{\textbf{A},\mathbb{R}}(\textbf{x})\in\mathbb{R}^{m\times d}$ is given by $$ \mathcal{E}_{\textbf{A},\mathbb{R}}(\textbf{x}) = [\sin(-k\omega_j - \delta_j)] = [-\Im (e^{\textbf{i}(k\omega_j+\delta_j)}) ],~\mathrm{where}~e^{\textbf{i}\delta_j} = \sign\big(\sum_{l=1}^d e^{-\textbf{i}l\omega_j}x_l\big).$$
	Without changing the rank, we multiply the $j$-th row of $\mathcal{E}_{\textbf{A},\mathbb{R}}(\textbf{x}) $ by $-|\sum_{l=1}^d e^{-\textbf{i}l\omega_j}x_l|$, then the resulting matrix (with the same rank of $\mathcal{E}_{\textbf{A},\mathbb{R}}(\textbf{x})$) reads as  $[\sum_{l=1}^d \sin\big((k-l)\omega_j\big)x_l]_{j\in[m],k\in[d]}$, which is denoted by $\widetilde{\mathcal{E}}$ in this proof. Now we let $u = k-l\in [1-d,d-1]$, $x_k=0$ for $k\geq d+1$ or $k\leq 0$, then the $(j,k)$-th entry of $\widetilde{\mathcal{E}}$ is given by  
	$$\sum_{l=1}^d\sin\big((k-l)\omega_j\big)x_l = \sum_{u= 1-d}^{d-1} \sin(u\omega_j)x_{k-u} = \sum_{u=1}^{d-1}\sin (u\omega_j)\big(x_{k-u}-x_{k+u}\big).$$
	This delivers $\widetilde{\mathcal{E}} = [\sin(u\omega_j)]_{j\in [m], u\in [d-1]} \cdot  \mathcal{B}_f(\textbf{x})$ with $\mathcal{B}_f(\textbf{x})$ given in (\ref{6.8}). Note that when $m\geq d-1$, $\omega_1,\cdots , \omega_{d-1}\in (0,\pi)$ are mutually different, we have $\rank \big([\sin(u\omega_j)]\big) = d-1$, hence $\rank(\widetilde{\mathcal{E}}) = d-1$ if and only if $\rank \big(\mathcal{B}_f(\textbf{x})\big) = d-1$.

	To  show that this recovers Condition \ref{condition2} for $\x$ with $x_1\neq 0$, we first observe that $\rank\big((\mathcal{B}_f(\textbf{x}))_{[d]\setminus \{1\}}\big)= \rank \big(\mathcal{B}(\textbf{x})\big)$. Obviously,    $\rank \big(\mathcal{B}(\textbf{x})\big) = d-1$ in condition \ref{condition2} trivially leads to our $\rank \big(\mathcal{B}_f(\textbf{x})\big) = d-1$. On the other hand, it is not hard to verify   $\mathcal{B}_f(\textbf{x})\textbf{x}=0 $. So when $\rank \big(\mathcal{B}_f(\textbf{x})\big) = d-1$, if $x_1 \neq 0$, it must hold that $\rank\big((\mathcal{B}_f(\textbf{x}))_{[d]\setminus \{1\}}\big) = d-1$, which gives $\rank \big(\mathcal{B}(\textbf{x})\big) = d-1$. The proof is concluded. 
	\end{proof}
	
	
	\subsection{Phase-only compressed sensing}\label{pocs}
    
    A recent line of research is concerned with
    phase-only compressed sensing  where the goal is to recover a sparse $\x\in\mathbb{R}^d$ from $\sign(\textbf{Ax})$ for some $\textbf{A}\in\mathbb{C}^{m\times d}$ \cite{feuillen2020ell,boufounos2013sparse,jacques2021importance,chen2022uniform}. However, note that these works and our paper are not directly comparable. Technically, the key strategy of these works is to establish the restricted isotropy property (or its variants) via various  concentration inequalities, which is in sharp contrast to our non-probabilistic arguments. Secondly, these results are only valid for  $\textbf{A}$ whose entries are i.i.d. drawn from complex Gaussian distribution, whereas ours are for   general $\A$ and do not require randomness. More prominently, the theoretical results in \cite{boufounos2013sparse,feuillen2020ell} 
	do not provide exact reconstruction, and the exact recovery guarantee in \cite{jacques2021importance} (see their Theorem 3.1) is non-uniform and only for a fixed signal in $\mathbb{R}^d$\footnote{After the revision of this paper,
     we improved the non-uniform result for 
		real-valued $\textbf{x}$ in \cite{jacques2021importance} to a uniform guarantee for complex-valued $\textbf{x}$, see our subsequent work \cite{chen2022uniform}.}. In comparison, our main results are uniform and guarantee that $\mathcal{W}_{\textbf{A}}$   contains a generic signal or even all signals in $\mathbb{C}^d$. Of course,  the  strength of these works is  that, the signal structure like sparsity can be effectively incorporated into the recovery to reduce measurement number. Also, measurement noise is considered in \cite{jacques2021importance,chen2022uniform}. Admittedly, these aspects are beyond the range of our current theoretical results, and it would be interesting to consider whether our theoretical framework can be extended to phase-only compressed sensing.

    \subsection{Phase versus magnitude}
    
    We provide one more interesting comparison to close this section. Recall that our Theorem \ref{T6} states that a generic $\textbf{A}$ of $\mathbb{C}^{(2d-1)\times d}$ is almost everywhere magnitude retrievable. Nevertheless, it was shown in Theorem 3.5 of \cite{huang2021almost} that, a generic $\textbf{A}$ of $\mathbb{C}^{(2d-1)\times d}$ is not almost evewhere phase retrievable\footnote{More precisely, the set of $\x$ that can not be uniquely specified (up to a global phase factor) by $|\textbf{Ax}|$ has positive Lebesgue measure.}. For general linear measurement, this seems to   indicate the phase is (slightly) more informative than the magnitude.

	\section{Affine phase-only reconstruction}
	\label{sec6}
	
    In this section, we 
		study the reconstruction of $\textbf{x}$ from phases of the affine measurements $\sign(\textbf{Ax}+\textbf{b})$ with $\textbf{A}\in\mathbb{C}^{m\times d}$, $\textbf{b}\in\mathbb{C}^{m\times 1}$. For convenience, we call $[\textbf{A},\textbf{b}]\in \mathbb{C}^{m\times (d+1)}$ the measurement matrix and term this problem   affine phase-only reconstruction.

    There are several motivations for considering this extension. For example,   linear measurement becomes   affine measurement when some entries  $\textbf{x}$ are known  a priori. This extension is also motivated by related recovery problems. Specifically, some  recent works   began to study affine phase retrieval, e.g., regarding the minimal measurement number \cite{gao2018phase,huang2021phase}, Newton's method \cite{gao2022newton}. Another problem related to phase-only measurement is 1-bit compressed sensing, where the aim is to recover the sparse real-valued signal $\textbf{x} $ from $\sign(\textbf{Ax})$ with $\textbf{A}\in\mathbb{R}^{m\times d}$ (e.g., see \cite{boufounos20081,plan2012robust})\footnote{Because $\textbf{Ax}$ in 1-bit compressed sensing is real-valued, entries of $\sign(\textbf{Ax})$ are binary ($1$ or $-1$), unlike the  phase-only measurement herein.}. To overcome some limitations in 1-bit compressed sensing, it is fruitful to introduce   $\textbf{b}\in\mathbb{R}^d$ and study the recovery of $\textbf{x}$ from $\sign(\textbf{Ax}+\textbf{b})$ ($\textbf{b}$ is called dither or dithering noise in related papers), for instance, extension of Gaussian sensing vectors to sub-Gaussian or even heavy-tailed ones \cite{dirksen2021non,chen2022high}, faster convergence rate \cite{baraniuk2017exponential}.

    In the affine case, an essential difference is that   the trivial ambiguity can be removed. Thus, the set of   signals that can be reconstructed 
 should be accordingly defined as
    \begin{equation}
        \label{7.1}
        \mathcal{W}_{\textbf{A,b}}:=\{\textbf{x}\in \mathbb{C}^d:\sign(\textbf{Ay}+\textbf{b})=\sign(\textbf{Ax}+\textbf{b})\ \mathrm{implies} \ \textbf{y}=\textbf{x}\}.
    \end{equation}
	We will establish the uniqueness criteria for affine phase-only reconstruction, which we then use to study  the problem of minimal measurement number. Except for Theorem \ref{T13}, the implications and proofs of other results are parallel to the corresponding ones for phase-only reconstruction. Thus, we relegate these proofs to Appendix \ref{appen3}. 
	
		\subsection{Discriminant matrices}
	We begin with some simple facts. If $\rank(\A)<d$, then there exists nonzero $\bm{\mathrm{y_0}}\in \ker(\A)$, thus implying $\mathcal{W}_{\textbf{A,b}}=\varnothing$ due to  $\sign(\textbf{Ax}+\textbf{b})=\sign(\bm{\mathrm{A(x+y_0)}}+\textbf{b})$. Moreover, if $\rank(\textbf{A})=d$ while $\textbf{b}=\bm{\mathrm{Ax_0}}\in \textbf{A}\mathbb{C}^d:=\{\textbf{Ay}:\textbf{y}\in\mathbb{C}^d\}$, 
	for any $\textbf{x}\neq -\bm{\mathrm{x_0}}$ we have $$\sign(\bm{\mathrm{A(2x+x_0)+b}})=\sign(\bm{\mathrm{A(2x+2x_0}}))=\sign(\Ax+\bm{\mathrm{Ax_0}})=\sign(\Ax+\textbf{b}).$$ Note that $2\textbf{x}+\bm{\mathrm{x_0}}\neq \textbf{x}$, hence $\textbf{x}\notin \mathcal{W}_{\bm{\mathrm{A,b}}}$. So in this case, $\mathcal{W}_{\textbf{A,b}}=\{-\bm{\mathrm{x_0}}\}$. Thus, we always assume $\rank(\textbf{A})=d,~\textbf{b}\notin \textbf{A}\mathbb{C}^d$ in this section, and these two assumptions may not be explicitly mentioned in the following.

	Evidently, exchanging rows of $[\textbf{A},\textbf{b}]$ cannot change $\mathcal{W}_{\textbf{A,b}}$. For any invertible $\textbf{P}\in \mathbb{C}^{d\times d}$, $\hat{\textbf{x}}\in \mathbb{C}^d$, it is  not difficult to verify that $\x\in \mathcal{W}_{\textbf{A,b}}$ if and only if $\bm{\mathrm{P^{-1}x}}+\hat{\textbf{x}}\in \mathcal{W}_{\textbf{AP},\textbf{b}-\bm{\mathrm{AP\hat{x}}}}$. This gives $\mathcal{W}_{\textbf{A,b}}=\textbf{P}\mathcal{W}_{\bm{\mathrm{AP,b-AP\hat{x}}}}-\hat{\textbf{x}}$, i.e., $\mathcal{W}_{\textbf{A,b}}$ and $\mathcal{W}_{\bm{\mathrm{AP,b-AP\hat{x}}}}$ only differ by an invertible affine transformation. 

	The facts above enable a canonical form of $[\textbf{A},\textbf{b}]$ that may be used without losing generality.  Specifically,   we can ensure $\rank(\textbf{A}^{[d]})=d$ by exchanging rows. Note that
	\begin{equation}
	    \label{7.2}
	    \Big[\textbf{A} \big(\textbf{A}^{[d]}\big)^{-1},\textbf{b}-\textbf{A} \big(\textbf{A}^{[d]}\big)^{-1} \textbf{b}^{[d]}\Big] = \begin{bmatrix}
	     \bm{\mathrm{I_d}} & \bm{0} \\ \bm{\mathrm{A_1}} & \bm{\mathrm{b_1}}
	\end{bmatrix}  
	\end{equation}
	for some $\bm{\mathrm{A_1}}\in \mathbb{C}^{(m-d)\times d}$, $\bm{\mathrm{b_1}}\in \mathbb{C}^{(m-d)\times 1}$. In subsequent development, the right-hand side of (\ref{7.2}) is referred to as the canonical measurement matrix.

	Given $[\textbf{A},\textbf{b}]\in \mathbb{C}^{m\times (d+1)}$, $\textbf{x}\in \mathbb{C}^d$, we define
	\begin{equation}
		\begin{aligned}
				\mathcal{D}_{\textbf{A,b}}(\x) =\begin{bmatrix}
					\Re(\A) & \Im(\A) & \Re(\mathrm{dg}(\textbf{Ax+b}))\\
					-\Im(\A) & \Re(\A) & -\Im(\mathrm{dg}(\textbf{Ax+b}))
				\end{bmatrix}.
				\label{7.3}
		\end{aligned}
	\end{equation}

	\begin{theorem}
		$\bm{\mathrm{x}}\in \mathcal{W}_{\bm{\mathrm{A,b}}}$ if and only if $\rank(\mathcal{D}_{\bm{\mathrm{A,b}}}(\bm{\mathrm{x}}))=2d+|\mathrm{N}(\bm{\mathrm{Ax}}+\bm{\mathrm{b}})|$.
		\label{T9}
	\end{theorem}
	

	Consider a canonical $[\textbf{A},\textbf{b}]$ as the right-hand side of (\ref{7.2}), we denote its $(j,k)$-th entry  by $r_{jk}e^{\textbf{i}\theta_{jk}}$ with $r_{jk}\geq 0$. The signal can also be written in a polar form $\textbf{x}=[|x_k|e^{\textbf{i}\alpha_k}]_{ k\in[ d]}$. The first $d$ measurements of the canonical measurement matrix give $\sign(\textbf{x})$. Let $\bm{\mathrm{\gamma_j^\top}},b_j$ be the $j$-th row of $\textbf{A}$, $j$-th entry of $\textbf{b}$ respectively. For $j\in [m]\setminus [d]$ such that $\bm{\mathrm{\gamma_j^{\top}x}}+b_j \neq 0$, we  assume $\sign(\bm{\mathrm{\gamma_j^{\top}x}}+b_j)=e^{\textbf{i}\delta_j}$ and define 
	\begin{equation}
	    \label{7.4}
	    \Psi_j^{'}(\x)=\begin{bmatrix}
		r_{j1}\cdot\sin(\theta_{j1}+\alpha_1-\delta_j) & r_{j2}\cdot\sin(\theta_{j2}+\alpha_2-\delta_j)&\cdots&r_{jd}\cdot\sin(\theta_{jd}+\alpha_d-\delta_j)
	\end{bmatrix}.
	\end{equation}
	While when $\bm{\mathrm{\gamma_j^\top}}x+b_j=0$, we let 
	\begin{equation}
	    \label{7.5}
	    \Psi^{'}_j(\x)=\begin{bmatrix}
		r_{j1}\cdot\sin(\theta_{j1}+\alpha_1)&r_{j2}\cdot\sin(\theta_{j2}+\alpha_2)&\cdots&r_{jd}\cdot\sin(\theta_{jd}+\alpha_d)\\
		r_{j1}\cdot\cos(\theta_{j1}+\alpha_1)&r_{j2}\cdot\cos(\theta_{j2}+\alpha_2)&\cdots&r_{jd}\cdot\cos(\theta_{jd}+\alpha_d)
	\end{bmatrix}.
	\end{equation}
	By stacking these blocks, we define another discriminant matrix to be 
	\begin{equation}
	    \label{7.6}
	    \mathcal{E}_{\textbf{A,b}}^0(\x)=\begin{bmatrix}
		\Psi^{'}_{d+1}(\x)\\
		\vdots\\
		\Psi^{'}_m(\x)
	\end{bmatrix};\mathcal{E}_{\textbf{A,b}}(\x)=\big(\mathcal{E}_{\textbf{A}}^0(\x)\big)_{\mathrm{N}(\x)}.
	\end{equation}
	
	\begin{theorem}
		Consider  a canonical measurement matrix $[\bm{\mathrm{A}},\bm{\mathrm{b}}]$ (see the right-hand side of (\ref{7.2})), then $\bm{\mathrm{x}}\in \mathcal{W}_{\bm{\mathrm{A,b}}} $ if and only if $\rank(\mathcal{E}_{\bm{\mathrm{A,b}}}(\bm{\mathrm{x}}))=|\mathrm{N}(\bm{\mathrm{x}})|$.
		\label{T10}
	\end{theorem}

	\subsection{Reconstruction of  all signals}
	
The measurement matrix $[\textbf{A},\textbf{b}]$ is said to be {\it affine magnitude retrievable} if $\mathcal{W}_{\textbf{A,b}}=\mathbb{C}^d$. The minimal measurement number for this property 
	\begin{equation}
	\label{87add}
	    \bm{\mathrm{m_{all}^{'}}}(d)=\min\{m\in \mathbb{N}_+:\text{Some } [\textbf{A},\textbf{b}]\in \mathbb{C}^{m\times (d+1)} \text{ is affine magnitude retrievable}\},
	\end{equation}
	is the focus of this subsection.

	As done in in Theorem \ref{T5}, we can establish the upper bound $4d+1$ by identifying the measurement matrices that are not affine magnitude retrievable. However, instead of following this approach, we can directly construct an affine magnitude retrievable $[\textbf{A},\textbf{b}]$. 	It is interesting to note that $3d$ measurements can also recover all signals in $\mathbb{C}^d$ in affine phase retrievable, see \cite{gao2018phase}. 
	
	\begin{theorem}
		For $d\in \mathbb{N}_+$, $\bm{\mathrm{m_{all}^{'}}}(d)\leq 3d$.
		\label{T13}
	\end{theorem}
	
\begin{proof}
	Let $\textbf{x} = [x_k]$, we consider $[\textbf{A},\textbf{b}]$ that gives $3d$ measurements $\big\{\sign(x_k),\sign(x_k+1),\sign(x_k+ \textbf{i}):k\in [d]\big\}$. It suffices to show   any $a\in\mathbb{C}$ can be determined by $\sign(a)$, $\sign(a+1)$ and $\sign(a+\textbf{i})$. For $a=0,-1,-\textbf{i}$, the result holds trivially, so we can only discuss the following two cases.
	
	\noindent{\it Case 1.} If $a\notin \mathbb{R}$, then by $a+1 = |a|\sign(a)+1$, we obtain $\big(|a|\sign(a)+1\big)/\big(\sign(a+1)\big)>0$. By taking the imaginary part it gives \begin{equation}
	\label{7.7}
	    \Im \Big(\frac{\sign(a)}{\sign(a+1)}\Big)\cdot |a| + \Im \Big(\frac{1}{\sign(a+1)}\Big)=0.
	\end{equation}
    If $\sign(a)/\sign(a+1) \in \mathbb{R}$, then $\frac{a}{a+1}\in \mathbb{R}$, which can lead to $a\in \mathbb{R}$, which is contradictory to our initial assumption. Hence, the coefficient of $|a|$ in (\ref{7.7}) is non-zero, and $|a|$ can be obtained via (\ref{7.7}). Combining with $\sign(a)$, we obtain $a$.

    \noindent{\it Case 2.} If $a\notin \textbf{i}\mathbb{R} = \{\textbf{i}u:u\in\mathbb{R}\}$, then can similarly obtain $|a|$ by using $\Im \big(\frac{|a|\sign(a)+\textbf{i}}{\sign(a+\textbf{i})}\big)=0$, the proof is hence concluded. 	
	\end{proof}
	
	By using $\mathcal{E}_{\textbf{A,b}}(\x)$ we can derive $2d+1$ as a lower bound. 
	
	\begin{theorem}
		For $d\in \mathbb{N}_+$, $\bm{\mathrm{m_{all}^{'}}}(d)\geq 2d+1$. 
		\label{T12}
	\end{theorem}

	\subsection{Reconstruction of almost all signals}

In this subsection, we switch to the minimal measurement number  for reconstruction of almost all signals. We say $[\textbf{A},\textbf{b}]$ is {\it almost everywhere affine magnitude retrievable} if $\mathbb{C}^d \setminus \mathcal{W}_{\textbf{A,b}}$
is of zero Lebesgue measure. Accordingly, the minimal measurement number required for this property is given by 
\begin{equation}
    \label{7.8}
    \bm{\mathrm{m^{'}_{ae}}}(d) = \big\{m : \text{some }[\textbf{A},\textbf{b}]\in\mathbb{C}^{m\times (d+1)}\text{ is almost everywhere affine magnitude retrievable}\big\}
\end{equation}

 More notations are needed to present the result. We use $\ker(\textbf{A,b})$ to denote $\{\textbf{x}:\textbf{Ax}+\textbf{b}=\bm{0}\}$. Given $\mathcal{S}\subset [m]$, then  $\mathcal{S}^c$ represents $[m]\setminus \mathcal{S}$, and we define 
	\begin{equation}
	    \label{7.9}
	    \mathcal{H}_{\textbf{A,b}}(\mathcal{S})=\{\x\in \mathbb{C}^d:\bm{\mathrm{\gamma_j^{\top}x}}+b_j=0,~\forall j\in \mathcal{S};\bm{\mathrm{\gamma_j^{\top}x}}+b_j\neq 0,~\forall j\in \mathcal{S}^c\}.
	\end{equation}
 For brevity, we write 
 \begin{equation}
     \label{7.10}
     \mathcal{H}_{\textbf{A,b}} = \mathcal{H}_{\textbf{A,b}}(\varnothing).
 \end{equation}
 
	\begin{theorem}
		Consider $[\bm{\mathrm{A,b}}]\in \mathbb{C}^{m\times (d+1)}$. When $m \leq 2d-1$, $\mathcal{W}_{\bm{\mathrm{A,b}}}$ is nowhere dense (under Euclidean topology) and of zero Lebesgue measure. When $m\geq 2d$, a generic $[\bm{\mathrm{A,b}}]$ satisfies
	\begin{equation}
		\begin{aligned}
		    \text{for all }\mathcal{S}\subset &[m] \text{ such that } \ker(\bm{\mathrm{A}}^\mathcal{S},\bm{\mathrm{b}}^\mathcal{S})\neq \varnothing,~\mathcal{W}_{\bm{\mathrm{A,b}}}\cap \ker(\bm{\mathrm{A}}^\mathcal{S},\bm{\mathrm{b}}^\mathcal{S}) \\
		    &\text{contains a generic point of } \ker(\bm{\mathrm{A}}^\mathcal{S},\bm{\mathrm{b}}^\mathcal{S})
		\end{aligned}
		\label{7.11}
	\end{equation} 
		Specifically, letting $S=\varnothing$  gives $\mathcal{W}_{\bm{\mathrm{A,b}}}$ contains a generic point of $\mathbb{C}^d$. Therefore, the minimal measurement number for almost everywhere affine magnitude retrievable property is $ \bm{\mathrm{m_{ae}^{'}}}(d)= 2d$.
		\label{T14}
	\end{theorem}

 To close this section, we present an interesting property analogous to Theorem \ref{T8}.

	\begin{theorem}
		Assume $m\geq 2d$, $\bm{\mathrm{x}}\in \mathcal{W}_{\bm{\mathrm{A,b}}}$, then there exists $\mathcal{S}\subset [m]$ with size $|\mathcal{S}|=2d$ such that $\bm{\mathrm{x}}\in \mathcal{W}_{\bm{\mathrm{A}}^\mathcal{S},\bm{\mathrm{b}}^\mathcal{S}}$. On the other hand, it is possible that $\bm{\mathrm{x}}\notin \mathcal{W}_{\bm{\mathrm{A}}^\mathcal{S},\bm{\mathrm{b}}^\mathcal{S}}$ holds for all $ \mathcal{S}\subset [m]$ with size $|\mathcal{S}|=2d-1$, .
		\label{T16}
	\end{theorem}

	\section{Concluding remarks}\label{sec7}
	
	In this paper, we 
	 built   a theoretical framework for phase-only reconstruction with a general measurement matrix $\textbf{A}\in\mathbb{C}^{m\times d}$ and the underlying complex-valued signal $\textbf{x}\in \mathbb{C}^d$.   Necessary and sufficient uniqueness conditions based on the rank of discriminant matrices were proposed. When specialized to real-valued signal and Fourier measurement matrix, our conditions recover the     those derived in previous studies (Section \ref{compare}). We derived an upper bound $4d-2$ and a lower bound $2d$ for the minimal measurement number required to reconstruct all signals. For reconstruction of almost all signals, $2d-1$ generic measurements are sufficient and minimal. Beyond measurement number, we   applied discriminant matrices to obtain two interesting results  (Section \ref{sec5}). We also note that, the theoretical framework can be readily extended to affine phase-only reconstruction (Section \ref{sec6}). To conclude the paper, we point out two possible directions for future research. Firstly, while we have shown $2d\leq \bm{\mathrm{m_{all}}}(d)\leq 4d-2$, a natural question is to derive tighter bound. Secondly, it is also interesting to extend the current framework to a noisy setting or phase-only compressed sensing.  
	 


\begin{appendix}

	\section{Proofs: real-valued signal reconstruction}
	\label{appen2}
	We first give the proofs of Theorem \ref{realcasedisd}, \ref{realdise}. The main strategy is adapted from the proofs of Theorem \ref{T1}, \ref{T2}. For   specific $\A\in\mathbb{C}^{m\times d}$ and $\textbf{x}\in\mathbb{R}^d$, we consider a linear system 
	\begin{equation}
	    \label{B.1}
	    \textbf{Ay} = \big[\mathrm{dg}(\sign(\textbf{Ax}))\big]_{\mathrm{N}(\textbf{Ax})}\bm{\lambda}
	\end{equation}
	with real variables $\textbf{y}\in\mathbb{R}^d$ and $\bm{\lambda}\in \mathbb{R}^{|\mathrm{N}(\textbf{Ax})|}$. Given $\bm{\lambda}$, (\ref{B.1}) can be viewed as a linear system of $\textbf{y}$, and it is not hard to see 
	\begin{equation}
	    \label{B.2}
	    \mathrm{V}_{\textbf{x},\mathbb{R}} := \big\{ \bm{\lambda}\in \mathbb{R}^{|\mathrm{N}(\textbf{Ax})|}:(\ref{B.2}) \text{ is consistent (i.e., has a solution)}\big\} 
	\end{equation}
	is a linear subspace. 
	To support the proof of Theorem \ref{realcasedisd} we first present a Lemma.
	
	\begin{lem}
	\label{lemma10}
	     For $\bm{\mathrm{A}}\in\mathbb{C}^{m\times d}$ we assume $ [\Re\bm{\mathrm{(A)^\top}},\Im\bm{\mathrm{(A)^\top}}]^{\bm{\top}}$ has full column rank, $\bm{\mathrm{x}}\in\mathbb{R}^d$ is non-zero. Then $\bm{\mathrm{x}}\in\mathcal{W}_{\bm{\mathrm{A}},\mathbb{R}}$ if and only if $\dim( \mathrm{V}_{\bm{\mathrm{x}},\mathbb{R}}) = 1$.
	\end{lem}
	\begin{proof}
	We write $\mathrm{U}_{\textbf{x},\mathbb{R}} = \mathrm{dg}(\sign(\textbf{Ax}))$, then (\ref{B.1}) equals $\textbf{Ay} = (\mathrm{U}_{\textbf{x},\mathbb{R}})_{\mathrm{N}(\textbf{Ax})}\bm{\lambda}$. Note that 
	we have $(\mathrm{U}_{\textbf{x},\mathbb{R}})_{\mathrm{N}(\textbf{Ax})}|\Ax|^{\mathrm{N}(\Ax)}=\mathrm{U}_{\textbf{x},\mathbb{R}}|\textbf{Ax}|=\textbf{Ax}$, hence $|\Ax|^{\mathrm{N}(\Ax)}\in \mathrm{V}_{\textbf{x},\mathbb{R}}$.  Since $|\textbf{Ax}|^{\mathrm{N}(\textbf{Ax})}\neq 0$, we have $\dim(\mathrm{V}_{\textbf{x},\mathbb{R}})\geq1$.

		We start from the "if" part. Assume  $\sign(\bm{\mathrm{A\tilde{x}}})=\sign(\Ax)$ for some $\bm{\mathrm{\tilde{x}}}\in\mathbb{R}^d$, then by exactly the same argument in the "if" part of the proof of Lemma \ref{lemma1}, we can obtain $\bm{\mathrm{A\tilde{x}}} = t\cdot \Ax $ for some $t>0$. Since $\bm{\mathrm{x,\tilde{x}}}\in\mathbb{R}^d$, this is equivalent to $\Re(\bm{\mathrm{A }})\bm{\mathrm{\tilde{x}}} =  \Re (\textbf{A}) ( t\cdot \textbf{x})  ,\Im(\bm{\mathrm{A }})\bm{\mathrm{\tilde{x}}} =  \Im (\textbf{A}) ( t\cdot \textbf{x})$. Combining with the full column rank of  $[\Re\bm{\mathrm{(A)^\top}},\Im\bm{\mathrm{(A)^\top}}]^{\bm{\top}}$, $\bm{\mathrm{\tilde{x}}} = t\cdot \textbf{x}$. Hence $\x \in \mathcal{W}_{\bm{\mathrm{A}}}$ is concluded.
	    The proof for the "only if" part is exactly parallel to that of the proof for Lemma \ref{lemma1}, thus we omit the details.
	    \end{proof}

	   We now give the proof of Theorem \ref{realcasedisd}. 
	    
	      
	      \vspace{2mm}
	    \noindent
	    {\it Proof of Theorem \ref{realcasedisd}:} The idea is to calculate $\dim(\mathrm{V}_{\textbf{x},\mathbb{R}})$ and then invoke Lemma \ref{lemma10}. Since $\textbf{y}, \bm{\lambda}$ in (\ref{B.1}) are real, the equation (\ref{B.1}) can be equivalently written as 
	    \begin{equation}
	        \label{B.3}
	        \begin{bmatrix}
	             \Re(\textbf{A}) & -\Re\big(  (\mathrm{U}_{\textbf{x},\mathbb{R}} )_{\mathrm{N}(\textbf{Ax})}\big)  \\
	             \Im (\textbf{A}) &  -\Im\big(  (\mathrm{U}_{\textbf{x},\mathbb{R}} )_{\mathrm{N}(\textbf{Ax})}\big) 
	        \end{bmatrix} \begin{bmatrix}
	             \textbf{y} \\ \bm{\lambda}
	        \end{bmatrix} = \bm{0}.
	    \end{equation}
	    Due to the full column rank of $[\Re\bm{\mathrm{(A)^\top}},\Im\bm{\mathrm{(A)^\top}}]^{\bm{\top}}$, for each $\bm{\lambda}\in  \mathrm{V}_{\textbf{x},\mathbb{R}}$, there exists a unique $\textbf{y}$ such that (\ref{B.3}) holds. Thus, the solution space of (\ref{B.3}) has the same dimension as $\mathrm{V}_{\textbf{x},\mathbb{R}}$. Thus, if we denote the coefficient matrix of (\ref{B.3}) by $\bm{\mathrm{\widetilde{D}}}$, a   result from linear system yields $ \dim (\mathrm{V}_{\textbf{x},\mathbb{R}}) = d+|\mathrm{N}(\textbf{Ax})| - \rank (\bm{\mathrm{\widetilde{D}}})$. By noting $\rank (\bm{\mathrm{\widetilde{D}}}) = \rank (\mathcal{D}_{\textbf{A},\mathbb{R}}(\textbf{x}))$ and using Lemma \ref{lemma10}, we can finish the proof. \hfill $\square$
	    
	    \vspace{2mm}
	    
	 Next, we give the proof of Theorem \ref{realdise}.
	    
	    \vspace{2mm}
	   \noindent
	   {\it Proof of Theorem \ref{realdise}:} We   use the notations   in (\ref{6.5}) and (\ref{6.6}).    Consider   $\textbf{y}=[y_k]\in\mathbb{R}^d$. When $\bm{\mathrm{\gamma_j^\top x}} \neq 0$, we let $e^{\textbf{i}\delta_j}=\sign(\bm{\mathrm{\gamma_j^\top x}})$, then   $\sign(\bm{\mathrm{\gamma_j^\top y}}) =\sign(\bm{\mathrm{\gamma_j^\top x}}) =  e^{\textbf{i}\delta_j}$ if and only if $\bm{\mathrm{\gamma_j^\top y}}\cdot e^{-\textbf{i}\delta_j}>0$, which can be further equivalently written as $\Psi_{j,\mathbb{R}}(\textbf{x})\textbf{y} =0$ and $\sum_{k=1}^d r_{jk}\cdot \cos(\theta_{jk}-\delta_j)\cdot y_k >0$. When $\bm{\mathrm{\gamma_j^\top x}} = 0$, then $\sign(\bm{\mathrm{\gamma_j^\top y}}) =\sign(\bm{\mathrm{\gamma_j^\top x}})$ if and only if $\bm{\mathrm{\gamma_j^\top y}} =0$. This is equivalent to $\Psi_{j,\mathbb{R}}(\textbf{x})\textbf{y} =\bm{0}$. Therefore, by (\ref{6.7}) it delivers that 
	   \begin{equation}
	       \label{B.4}
	       \sign(\textbf{Ax}) = \sign(\textbf{Ay}) \iff  \begin{cases}
	      \mathcal{E}_{\textbf{A},\mathbb{R}}(\textbf{x})\cdot \textbf{y} = \bm{0}   \\ 
	      \sum_{k=1}^d r_{jk}\cdot \cos(\theta_{jk}-\delta_j)\cdot y_k >0,~\forall~ j~\text{s.t.}~\bm{\mathrm{\gamma_j^\top x}}\neq 0
	       \end{cases}.
	   \end{equation}
	   Recall the assumption $\bm{\mathrm{\gamma_j^\top x}}\neq 0$ for some $j$. By letting $\textbf{y} = \textbf{x}$, we obtain 
	   \begin{equation}
	       \label{B.5}
	       \begin{cases}
	      \mathcal{E}_{\textbf{A},\mathbb{R}}(\textbf{x})\cdot \textbf{x} = \bm{0}   \\ 
	      \sum_{k=1}^d r_{jk}\cdot \cos(\theta_{jk}-\delta_j)\cdot x_k >0,~\forall~ j~\text{s.t.}~\bm{\mathrm{\gamma_j^\top x}}\neq 0
	       \end{cases}.
	   \end{equation}
	   Note that the first equation in (\ref{B.5}) gives $\rank\big( \mathcal{E}_{\textbf{A},\mathbb{R}}(\textbf{x})\big) \leq  d-1$.

	   We consider the "if" part and assume $\rank\big( \mathcal{E}_{\textbf{A},\mathbb{R}}(\textbf{x})\big) =  d-1$. Then when $\sign(\textbf{Ay}) = \sign(\textbf{Ax})$, we   combine the first equation on the right-hand side of (\ref{B.4}) and the first equation of (\ref{B.5}),  it delivers that $\textbf{y} = t \cdot \textbf{x}$ for some $t\in\mathbb{R}$. Then we further compare the second equation on the right-hand side of (\ref{B.4}) and the second equation of (\ref{B.5}), then it follows that $t>0$. Hence, we arrive at $\textbf{x}\in\mathcal{W}_{\textbf{A},\mathbb{R}}$.

	  Next,  we consider the "only if" part. Due to $\rank\big( \mathcal{E}_{\textbf{A},\mathbb{R}}(\textbf{x})\big) \leq  d-1$, we only need to show contradiction under the assumption $\rank\big( \mathcal{E}_{\textbf{A},\mathbb{R}}(\textbf{x})\big) <  d-1$. If this happens, then $\dim \big(\ker(\mathcal{E}_{\textbf{A},\mathbb{R}}(\textbf{x}))\big) = d-\rank(\mathcal{E}_{\textbf{A},\mathbb{R}}(\textbf{x}))\geq 2$. Note that $\textbf{x}\in \ker(\mathcal{E}_{\textbf{A},\mathbb{R}}(\textbf{x}))$, one can find $\textbf{y} = [y_k]\in \ker(\mathcal{E}_{\textbf{A},\mathbb{R}}(\textbf{x}))$ such that $\textbf{x}$ and $\textbf{y}$ are sufficiently close but linearly independent. By letting $\textbf{x}$, $\textbf{y}$ be sufficiently close, due to the second equation of (\ref{B.5}), the second equation on the right-hand side of (\ref{B.4}) can be guaranteed. Thus, this  $\textbf{y}$ satisfies the right-hand side of (\ref{B.4}), and hence we have $\sign(\textbf{Ay}) = \sign(\textbf{Ax})$. Now we invoke the condition $\textbf{x}\in\mathcal{W}_{\textbf{A},\mathbb{R}}$, it can give $\textbf{y} = t_1 \cdot \textbf{x}$ for some $t_1 >0$. This is contradictory to our choice of $\textbf{y}$. The proof is concluded. \hfill $\square$

	\section{Proofs:  affine  phase-only reconstruction }
	\label{appen3}
	We consider the equation 
\begin{equation}
	\textbf{Ay}+\textbf{b}=\mathrm{dg}(\sign(\textbf{Ax}+\textbf{b}))_{\mathrm{N}(\textbf{Ax}+\textbf{b})}\bm{\lambda},
	\label{C.1}
\end{equation}	
where $\y\in \mathbb{C}^d$, $\bm{\lambda}\in \mathbb{R}^{|\mathrm{N}(\textbf{Ax}+\textbf{b})|}.$ We assume $\rank(\A)=d$, $\textbf{b}\notin \textbf{A}\mathbb{C}^d$.

\begin{lem}
$	\x\in \mathcal{W}_{\textbf{A,b}}$ if and only if $\{\bm{\lambda}\in \mathbb{R}^{|\mathrm{N}(\textbf{Ax+b})|}:(\ref{C.1})\  \mathrm{has}\ \mathrm{solution}\}$ is a $0$-dimensional linear submanifold, i.e., a set containing a single nonzero point.
\label{lemma6}
\end{lem}

\begin{proof}
To lighten the notation we use \begin{equation}
    \label{C.2}
   \begin{aligned}
        &\mathrm{V}_{\x}=\{\bm{\lambda}\in \mathbb{R}^{|\mathrm{N}(\textbf{Ax}+\textbf{b})|}:(\ref{C.1})\ \text{is consistent (or has a solution)}\}\\
        &\mathrm{U}_{\x}=\mathrm{dg}(\sign(\textbf{Ax}+\textbf{b}))
   \end{aligned}
\end{equation} 
in this proof. Hence $$(\ref{C.1}) \iff \textbf{Ay+b}=(\mathrm{U}_{\x})_{\mathrm{N}(\textbf{Ax}+\textbf{b})}\bm{\lambda}.$$ Since $(\mathrm{U}_{\x})_{\mathrm{N}(\textbf{Ax}+\textbf{b})}|\textbf{Ax+b}|^{\mathrm{N}(\textbf{Ax}+\textbf{b})}=\mathrm{U}_{\x}|\textbf{Ax+b}|=\textbf{Ax+b}$, we have $|\textbf{Ax+b}|^{\mathrm{N}(\textbf{Ax}+\textbf{b})}\in \mathrm{V}_{\x}$. Let us prove that $\mathrm{V}_{\x}$ is a linear submanifold of $\mathbb{R}^{|\mathrm{N}(\textbf{Ax}+\textbf{b})|}$. Given $\bm{\lambda_1},\bm{\lambda_2} \in \mathrm{V}_{\x}$, $t \in \mathbb{R}$, by definition there exist $\bm{\mathrm{y_1}}$ and $\bm{\mathrm{y_2}}$ such that $\Ay_1+\textbf{b}=(\mathrm{U}_{\x})_{\mathrm{N}(\textbf{Ax}+\textbf{b})}\bm{\lambda_1}$, $\Ay_2+\textbf{b}=(\mathrm{U}_{\x})_{\mathrm{N}(\textbf{Ax}+\textbf{b})}\bm{\lambda_2}$. Thus, we have   $$\textbf{A}(t\cdot \bm{\mathrm{y_1}}+(1-t)\cdot \bm{\mathrm{y_2}})+\textbf{b}=(\mathrm{U}_{\x})_{\mathrm{N}(\textbf{Ax}+\textbf{b})}(t\cdot \bm{\lambda_1}+(1-t)\cdot\bm{\lambda_2}),$$ which implies $t\cdot \bm{\lambda_1}+(1-t)\cdot\bm{\lambda_2}\in \mathrm{V}_{\x}$. Therefore, $\mathrm{V}_{\x}$ is a linear submanifold.

For the "if" part, assume $\sign(\bm{\mathrm{A\tilde{x}+b}})=\sign(\textbf{Ax}+\textbf{b})$, then $|\bm{\mathrm{A\tilde{x}+b}}|^{\mathrm{N}(\textbf{Ax}+\textbf{b})}\in \mathrm{V}_{\x}$. Under the condition that $\mathrm{V}_{\textbf{x}}$ contains a single point, we obtain $|\bm{\mathrm{A\tilde{x}+b}}|=|\textbf{Ax+b}|$, and hence  $\bm{\mathrm{A\tilde{x}+b}}=\textbf{Ax+b}$. Combining with $\rank(\A)=d$, it delivers $\textbf{x} = \bm{\mathrm{\tilde{x}}}$, thus concluding $\x\in \mathcal{W}_{\textbf{A,b}}$.

For the "only if" part, we assume $\x\in \mathcal{W}_{\textbf{A,b}}$ and only need to rule out the possibility of $\dim(\mathrm{V}_{\x})\geq 1$. Note that $|\textbf{Ax+b}|^{\mathrm{N}(\textbf{Ax}+\textbf{b})}\in \mathbb{R}_+^{|\mathrm{N}(\textbf{Ax}+\textbf{b})|},$ if $\dim(\mathrm{V}_{\x})\geq 1$, there exists $\bm{\lambda_0}\in\mathrm{V}_{\x}\cap\mathbb{R}_+^{|\mathrm{N}(\textbf{Ax}+\textbf{b})|}$ such that $\bm{\lambda_0} \neq |\textbf{Ax+b}|^{\mathrm{N}(\textbf{Ax}+\textbf{b})}$. Since $\bm{\lambda_0} \in \mathrm{V}_{\x}$, there exists $\bm{\mathrm{y_0}}$ such that $\bm{\mathrm{Ay_0+b}}=(\mathrm{U}_{\x})_{\mathrm{N}(\textbf{Ax}+\textbf{b})}\bm{\lambda_0}$. This implies $\sign(\bm{\mathrm{Ay_0+b}})=\sign(\textbf{Ax+b})$. However,   $\bm{\lambda_0} \neq |\textbf{Ax+b}|^{\mathrm{N}(\textbf{Ax}+\textbf{b})}$ gives $\bm{\mathrm{y_0}} \neq \textbf{x}$, which is contradictory to our initial assumption $\x\in \mathcal{W}_{\textbf{A,b}}.$ 
\end{proof}

\noindent 
{\it Proof of Theorem \ref{T9}:} We continue to use notations introduced in (\ref{C.2}), and our strategy is to calculate $\dim (\mathrm{V}_{\textbf{x}})$ and then invoke Lemma \ref{lemma6}. To this end, noting  $\bm{\lambda}\in\mathbb{R}^{|\mathrm{N}(\textbf{Ax+b})|}$,  (\ref{C.1}) is equivalent to        the real linear system $\varphi(\A)\varphi_1(\textbf{y})-\varphi_1((\mathrm{U}_{\x})_{\mathrm{N}(\textbf{Ax}+\textbf{b})})\bm{\lambda} =-\varphi_1(\textbf{b})$, that is,   \begin{equation}
    \begin{bmatrix}
         \Re(\textbf{A}) & \Im(\textbf{A}) & -\Re \big((\mathrm{U}_{\textbf{x}})_{\mathrm{N}(\textbf{Ax+b})}\big) \\
         -\Im(\textbf{A}) & \Re(\textbf{A}) & \Im\big((\mathrm{U}_{\textbf{x}})_{\mathrm{N}(\textbf{Ax+b})}\big)
    \end{bmatrix}\begin{bmatrix}
          \varphi_1(\textbf{y}) \\ \bm{\lambda}
    \end{bmatrix} =-\varphi_1(\textbf{b}).
    \label{C.3}              
\end{equation} 
Since $\rank \big(\varphi(\textbf{A})\big) = 2\cdot \rank(\textbf{A})  =2d$, so for each $\bm{\lambda}\in \mathrm{V}_{\textbf{x}}$, there exists unique $\varphi_1(\textbf{y})$ such that (\ref{C.3}) holds. Thus, the real linear system (\ref{C.3}) possesses a solution space with the same dimension as $\mathrm{V}_{\textbf{x}}$. Denote the coefficient matrix of (\ref{C.3}) via $\widetilde{\mathcal{D}}$, then it  gives
\begin{equation}
    \label{C.4}
    \dim(\mathrm{V}_{\x}) = 2d+|\mathrm{N}(\textbf{Ax}+\textbf{b})|-\rank(\widetilde{\mathcal{D}}) = 2d+|\mathrm{N}(\textbf{Ax}+\textbf{b})| - \rank\big(\mathcal{D}_{\textbf{A,b}}(\textbf{x})   \big),
\end{equation}
where $\rank(\widetilde{\mathcal{D}})=\rank\big(\mathcal{D}_{\textbf{A,b}}(\textbf{x})   \big)$ is a simple observation. Now we can invoke Lemma \ref{lemma6} and conclude the proof. \hfill $\square$

\vspace{2mm}
\noindent
{\it Proof of Theorem \ref{T10}:}
 For convenience, in this proof we use the shorthand 
 \begin{equation}
     \label{C.5}
    \begin{aligned}
         &\mathcal{J}_1=\{j\in[m]\setminus [d]:\bm{\mathrm{\gamma_j^{\top}x}}+b_j \neq 0\},\\
         &\mathcal{J}_2=\{j\in[m]\setminus [d]:\bm{\mathrm{\gamma_j^{\top}x}}+b_j=0\}.
    \end{aligned}
 \end{equation}
  Note that we consider canonical $[\textbf{A},\textbf{b}]$ in a form of the right-hand side of (\ref{7.2}), hence the first $d$ measurements give the phases of the signal, i.e., $\sign(x)$.
  For $j\in \mathcal{J}_1  $, recall  (\ref{7.4}) and we further define  $b_j(\x)=[r_{j,d+1}\sin(\theta_{j,d+1}-\delta_j)](\in\mathbb{R}^{1\times 1})$. for $j\in \mathcal{J}_2$, recall (\ref{7.5}) and we give an additional notation $b_j(\x)=\begin{bmatrix}
	r_{j,d+1}\sin(\theta_{j,d+1})&r_{j,d+1}\cos(\theta_{j,d+1})
\end{bmatrix}^\top\in\mathbb{R}^{2\times 1}$. Now, we concatenate all $b_j(\textbf{x})$ and let $b(\x)=[b_{d+1}(\x)^{\top},b_{d+2}(\x)^{\top},\cdots,b_{m}(\x)^{\top}]^{\top}$. By some algebra, one can easily verify  $\sign(\Ay+\textbf{b})=\sign(\textbf{Ax}+\textbf{b}) \iff$
\begin{equation}
\label{C.6}
\begin{aligned}
 \begin{cases}
	 \sign(\y)=\sign(\x)\\
	 \mathcal{E}_{\textbf{A,b}}(\x)|\textbf{y}|^{\mathrm{N}(\x)}+b(\x)=\bm{0}\\
   \sum_{k\in \mathrm{N}(\x)} r_{jk}\cos(\theta_{jk}+\alpha_k-\delta_j)|y_k|+r_{j,d+1}\cos(\theta_{j,d+1}-\delta_j)>0,\ \forall j\in  \mathcal{J}_1
	\end{cases}
\end{aligned}
\end{equation}
Specifically we can let $\textbf{y}=\textbf{x}$, then it gives 
\begin{equation}
    \label{C.7}
    \begin{cases}
        \mathcal{E}_{\textbf{A,b}}(\x)|\x|^{\mathrm{N}(\x)}+b(\x)=\bm{0}\\
        \sum_{k\in \mathrm{N}(\x)} r_{jk}\cos(\theta_{jk}+\alpha_k-\delta_j)|x_k|+r_{j,d+1}\cos(\theta_{j,d+1}-\delta_j)>0,\ \forall j\in \mathcal{J}_1
    \end{cases}.
\end{equation}

The "if" part is straightforward. Indeed, when $ \sign(\Ay+\textbf{b})=\sign(\textbf{Ax}+\textbf{b})$, the second equation of (\ref{C.6}) holds. Combining with the first equation in (\ref{C.7}), and the assumption $\rank\big(\mathcal{E}_{\textbf{A,b}}(\textbf{x})\big)= |\mathrm{N}(\textbf{x})|$, we obtain $|\textbf{y}|^{\mathrm{N}(\textbf{x})} = |\textbf{x}|^{\mathrm{N}(\textbf{x})}$. Due to $\sign(\textbf{y})=\sign(\textbf{x})$, we arrive at $\y=\x$, thus confirming $\textbf{x}\in\mathcal{W}_{\textbf{A,b}}$.

For the "only if" part, by assuming $\textbf{x}\in\mathcal{W}_{\textbf{A,b}}$ and  $\rank(\mathcal{E}_{\textbf{A,b}}(\x))<|\mathrm{N}(\x)|$, we only need to show the contradiction. Under these two conditions, we can find $\hat{\textbf{y}}\in \mathbb{R}_+^{|\mathrm{N}(\x)|}$, $\hat{\textbf{y}}\neq |\x|^{\mathrm{N}(\x)}$, such that $\mathcal{E}_{\textbf{A,b}}(\x)\hat{\textbf{y}}+b(\x)=\bm{0}$. Based on $\hat{\textbf{y}}$, we can  construct $\textbf{y}^0 = [y^0_k]$ such that $\sign(\textbf{y}^0 )=\sign(\x)$, and $|\textbf{y}^0 |^{\mathrm{N}(\x)}=\hat{\textbf{y}}$. Moreover, we can   choose $\hat{\textbf{y}}$ sufficiently close to $|\x|^{\mathrm{N}(\x)}$ to guarantee $$\displaystyle \sum_{k\in \mathrm{N}(\x)}r_{jk}\cos(\theta_{jk}+\alpha_k-\delta_j)|y^0_k|+r_{j,d+1}\cos(\theta_{j,d+1}-\delta_j)>0,\ \forall j\in \mathcal{J}_1,$$
which displays the third equation of (\ref{C.6}).
Thus, $\textbf{y}^0$ satisfies (\ref{C.6}), and hence $\sign(\Ay^0+\textbf{b})=\sign(\textbf{Ax}+\textbf{b})$. Now we invoke $\x\in \mathcal{W}_{\textbf{A,b}}$ and obtain $\x = \textbf{y}^0$. This is contradictory to $$|\textbf{y}^0|^{\mathrm{N}(\textbf{x})}=\hat{\textbf{y}}\neq |\x|^{\mathrm{N}(\textbf{x})}.$$
The proof is hence concluded. \hfill $\square$

\vspace{2mm}
\noindent
{\it   Proof of Theorem \ref{T12}:} Recall the entry-wise notations 
$[\textbf{A},\textbf{b}]=[r_{jk}e^{\textbf{i}\theta_{jk}}]$, $\textbf{x} = [x_k]^\top$, and $e^{\textbf{i}\delta_j} = \sign(\bm{\mathrm{\gamma_j^\top x}}+b_j)$ when $\bm{\mathrm{\gamma_j^\top x}}+b_j \neq 0$. We now consider a canonical measurement matrix (see the right-hand side of (\ref{7.2})) and $\textbf{x}$ with $\alpha_k=\theta_{d+1,d+1}-\theta_{d+1,k}$, $k\in[d]$. Similar to the proof of Theorem \ref{T4}, we can   find positive numbers $\lambda_1,\lambda_2,\cdots,\lambda_d$ such that \begin{equation}
\label{C.8}
     x=[\lambda_1\cdot e^{\textbf{i}(\theta_{d+1,d+1}-\theta_{d+1,1})},\cdots,\lambda_d\cdot e^{\textbf{i}(\theta_{d+1,d+1}-\theta_{d+1,d})}]^{\top}\in \mathcal{H}_{\textbf{A,b}},
\end{equation} 
where $\mathcal{H}_{\textbf{A,b}}$ is defined in (\ref{7.10}). Note that due to $\textbf{x}\in \mathcal{H}_{\textbf{A,b}}$ and $\mathrm{N}(\x) = [d]$, we have $\mathcal{E}_{\textbf{A,b}}(\x)\in\mathbb{R}^{(m-d)\times d}$.
Now, a simple calculation can give $\delta_{d+1}=\theta_{d+1,d+1}$ (up to an integer multiple of $2\pi$), and hence the first row of $\mathcal{E}_{\textbf{A,b}}(\x)$ equals zero. To conclude the proof, we invoke Theorem \ref{T10}, it yields $m-d-1\geq d$, hence $m\geq 2d+1$ follows.    \hfill $\square$ 

\vspace{2mm}
For clarity, we give several Lemmas to support the proof of Theorem \ref{T14}.

\begin{lem}
	$\x\in \mathcal{W}_{\textbf{A,b}}\cap \mathcal{H}_{\textbf{A,b}}$ if and only if $\rank(\mathcal{D}_{\textbf{A,b}}(\x))\geq 2d+m.$
	\label{lemma7}
\end{lem} 

\begin{proof}
The "only if" part comes from Theorem \ref{T9} directly. For the "if" part, it is evident that $$ \rank(\mathcal{D}_{\textbf{A,b}}(\x))\leq \rank(\varphi(\A))+\rank(\varphi_1((\mathrm{dg}(\textbf{Ax+b})))\leq 2d+|\mathrm{N}(\textbf{Ax}+\textbf{b})|.$$ Combining with $\rank(\mathcal{D}_{\textbf{A,b}}(\x))\geq 2d+m$, we have $|\mathrm{N}(\textbf{Ax}+\textbf{b})|\geq m$. This implies $|\mathrm{N}(\textbf{Ax}+\textbf{b})|=m$ and hence $\x\in \mathcal{H}_{\textbf{A,b}}$. We use Theorem \ref{T9} again, the result follows. 
\end{proof}

\begin{lem}
	Assume $\mathcal{S}\subset[m]$, $1\leq |\mathcal{S}|<d$, and $\rank(\textbf{A}^\mathcal{S}_{[\mathcal{S}]})=|\mathcal{S}|$, $\x\in \ker(\textbf{A}^\mathcal{S},\textbf{b}^\mathcal{S})$. Define 
		\begin{equation}
			\begin{aligned}
				\textbf{A}^{'}(\mathcal{S})&=\textbf{A}^{\mathcal{S}^c}_{[d]\setminus [\mathcal{S}]}-\textbf{A}^{\mathcal{S}^c}_{[\mathcal{S}]}(\textbf{A}^\mathcal{S}_{[\mathcal{S}]})^{-1}\textbf{A}^\mathcal{S}_{[d]\setminus[\mathcal{S}]}\in \mathbb{C}^{(m-|\mathcal{S}|)\times (d-|\mathcal{S}|)};\\
				\textbf{b}^{'}(\mathcal{S})&=	\textbf{b}^{\mathcal{S}^c}-\textbf{A}^{\mathcal{S}^c}_{[\mathcal{S}]}(\textbf{A}^\mathcal{S}_{[\mathcal{S}]})^{-1}	\textbf{b}^\mathcal{S}\in \mathbb{C}^{m-|\mathcal{S}|}.
				\label{C.9}
			\end{aligned}
	\end{equation}
Then we have $\x\in \mathcal{W}_{\textbf{A,b}}\cap \mathcal{H}_{\textbf{A,b}}(	\mathcal{S})$ if and only if $\textbf{x}^{[d]\setminus[\mathcal{S}]}\in \mathcal{W}_{\textbf{A}^{'}(\mathcal{S}),\textbf{b}^{'}(\mathcal{S})}\cap \mathcal{H}_{\textbf{A}^{'}(\mathcal{S}),\textbf{b}^{'}(	\mathcal{S})}$.
	\label{lemma8}
\end{lem}

\begin{proof}
Note that $\x\in  \ker(\textbf{A}^\mathcal{S},\textbf{b}^\mathcal{S})$ equals $\textbf{A}^\mathcal{S}\textbf{x}+\textbf{b}^\mathcal{S} = \bm{0}$. Due to $\rank(\textbf{A}^\mathcal{S}_{[\mathcal{S}]})=|\mathcal{S}|$, this can be equivalently given by
\begin{equation}
    \label{C.10}
    \textbf{A}^\mathcal{S}_{[\mathcal{S}]}\textbf{x}^{[\mathcal{S}]} + \textbf{A}^\mathcal{S}_{[d] \setminus  [\mathcal{S}]}\textbf{x}^{[d] \setminus [\mathcal{S}]}+\textbf{b}^\mathcal{S} = \bm{0} \iff \textbf{x}^{[\mathcal{S}]} = -\big(\textbf{A}^\mathcal{S}_{[\mathcal{S}]}\big)^{-1}\textbf{A}^\mathcal{S}_{[d]\setminus [\mathcal{S}]}\textbf{x}^{[d]\setminus [\mathcal{S}]} -  \big(\textbf{A}^\mathcal{S}_{[\mathcal{S}]}\big)^{-1}\textbf{b}^\mathcal{S}.
\end{equation}
Based on (\ref{C.10}), and recall the notations in (\ref{C.9}), some algebra can verify 
\begin{equation}
    \label{C.11}
    \textbf{A}^{\mathcal{S}^c}\textbf{x}+\textbf{b}^{\mathcal{S}^c} = \textbf{A}^{'}(\mathcal{S})\textbf{x}^{[d]\setminus [\mathcal{S}]} + \textbf{b}^{'}(\mathcal{S}).
\end{equation}
  Then the rest of this proof is analogous to that of Lemma \ref{lemma5}.

  For the "only if" part, we assume $\textbf{x}\in \mathcal{W}_{\textbf{A,b}}\cap \mathcal{H}_{\textbf{A,b}}(\mathcal{S})$. By definition of $\mathcal{H}_{\textbf{A,b}}(\mathcal{S})$ (see (\ref{7.9})), the left-hand side (and hence the right-hand side) of (\ref{C.11})  contains no zero entries. This gives $\textbf{x}^{[d]\setminus[\mathcal{S}]}\in \mathcal{H}_{\textbf{A}^{'}(\mathcal{S}),\textbf{b}^{'}(	\mathcal{S})}$. To show $\textbf{x}^{[d]\setminus[\mathcal{S}]}\in \mathcal{W}_{\textbf{A}^{'}(\mathcal{S}),\textbf{b}^{'}(	\mathcal{S})}$, we assume $$\sign\big(\textbf{A}^{'}(\mathcal{S})\textbf{x}^{[d]\setminus [\mathcal{S}]} + \textbf{b}^{'}(\mathcal{S})\big) = \sign\big(\textbf{A}^{'}(\mathcal{S})\bm{\mathrm{y_0}} + \textbf{b}^{'}(\mathcal{S})\big),\text{ for some }\bm{\mathrm{y_0}}\in\mathbb{C}^{d-|\mathcal{S}|}.$$
  Motivated by (\ref{C.10}) we consider 
  $$\textbf{y}=\begin{bmatrix}
      -\big(\textbf{A}^\mathcal{S}_{[\mathcal{S}]}\big)^{-1}\textbf{A}^\mathcal{S}_{[d]\setminus [\mathcal{S}]}\bm{\mathrm{y_0}} -  \big(\textbf{A}^\mathcal{S}_{[\mathcal{S}]}\big)^{-1}\textbf{b}^\mathcal{S}  \\ \bm{\mathrm{y_0}}
  \end{bmatrix}$$
  that satisfies $\textbf{A}^{\mathcal{S}}\textbf{y} + \textbf{b}^\mathcal{S} = \bm{0} = \textbf{A}^{\mathcal{S}}\textbf{x} + \textbf{b}^\mathcal{S}$, it is not hard to see \begin{equation}
      \nonumber
      \begin{aligned}
         & \sign\big(\textbf{A}^{\mathcal{S}^c}\textbf{y}+\textbf{b}^{\mathcal{S}^c}\big)=\sign\big(\textbf{A}^{'}(\mathcal{S})\bm{\mathrm{y_0}}+\textbf{b}^{'}(\mathcal{S})\big) \\=& \sign\big(\textbf{A}^{'}(\mathcal{S})\textbf{x}^{[d]\setminus [\mathcal{S}]} + \textbf{b}^{'}(\mathcal{S})\big) = \sign\big(\textbf{A}^{\mathcal{S}^c}\textbf{x}+\textbf{b}^{\mathcal{S}^c}\big).
      \end{aligned}
  \end{equation}
  Therefore, we obtain $\sign\big(\textbf{Ay+b}\big)=\sign\big(\textbf{Ax+b}\big)$, which together with $\x\in\mathcal{W}_{\textbf{A,b}}$ can yield $\y = t\cdot \x$ for some $t>0$. This evidently leads to $\bm{\mathrm{y_0}}  = t\cdot \textbf{x}^{[d]\setminus [\mathcal{S}]}$, and hence $\textbf{x}^{[d]\setminus [\mathcal{S}]}\in\mathcal{W}_{\textbf{A}^{'}(\mathcal{S}),\textbf{b}^{'}(\mathcal{S})}$.

  We go into the "if" part and assume $\textbf{x}^{[d]\setminus[\mathcal{S}]}\in\mathcal{W}_{\textbf{A}^{'}(\mathcal{S}),\textbf{b}^{'}(\mathcal{S})}\cap \mathcal{H}_{\textbf{A}^{'}(\mathcal{S}),\textbf{b}^{'}(\mathcal{S})}$. This implies the right-hand side (and hence also the left-hand side) of (\ref{C.11})     contains no zero entries. Since $\textbf{A}^\mathcal{S}\textbf{x}+\textbf{b}^\mathcal{S} = \bm{0}$, we obtain $\textbf{x}\in\mathcal{H}_{\textbf{A,b}}(\mathcal{S})$. It remains to show $\textbf{x}\in\mathcal{W}_{\textbf{A,b}}$. For this purpose, we assume $\sign(\textbf{Ax+b}) = \sign(\textbf{Ay+b})$ for some $\textbf{y}\in\mathbb{C}^d$, which gives $\textbf{A}^{\mathcal{S}}\textbf{y}+\textbf{b}^{\mathcal{S}}=\bm{0}$, or equivalently, $\textbf{y}^{[\mathcal{S}]} = -\big(\textbf{A}^\mathcal{S}_{[\mathcal{S}]}\big)^{-1}\textbf{A}^\mathcal{S}_{[d]\setminus [\mathcal{S}]}\textbf{y}^{[d]\setminus [\mathcal{S}]} -  \big(\textbf{A}^\mathcal{S}_{[\mathcal{S}]}\big)^{-1}\textbf{b}^\mathcal{S}$. Based on this relation, recall (\ref{C.11}), some algebra gives \begin{equation}
      \nonumber\begin{aligned}
         &\sign \big(\textbf{A}^{'}(\mathcal{S})\textbf{y}^{[d]\setminus [\mathcal{S}]}+\textbf{b}^{'}(\mathcal{S})\big) = \sign\big(\textbf{A}^{\mathcal{S}^c}\textbf{y}+\textbf{b}^{\mathcal{S}^c}\big) \\
         =&\sign\big(\textbf{A}^{\mathcal{S}^c}\textbf{x}+\textbf{b}^{\mathcal{S}^c}\big) = \sign\big(\textbf{A}^{'}(\mathcal{S})\textbf{x}^{[d]\setminus [\mathcal{S}]} + \textbf{b}^{'}(\mathcal{S})\big).
      \end{aligned}
  \end{equation}
  Now, we can obtain $\textbf{y}^{[d]\setminus [\mathcal{S}]} =\textbf{x}^{[d]\setminus [\mathcal{S}]} $ by $\textbf{x}^{[d]\setminus [\mathcal{S}]} \in \mathcal{W}_{\textbf{A}^{'}(\mathcal{S}),\textbf{b}^{'}(\mathcal{S})}$, which directly leads to $\textbf{x} = \textbf{y}$. Thus, $\textbf{x}\in\mathcal{W}_{\textbf{A,b}}$ and the proof is concluded. 
 \end{proof}

Similar to (\ref{4.3}), (\ref{4.4}), we introduce the notations that are more amenable for analyzing a fixed signal $\x$. Specifically we let 
\begin{equation}
    \label{C.12}
    \begin{aligned}
       & \mathcal{W}^{'}_{\textbf{x}}(m) = \big\{[\textbf{A},\textbf{b}]\in \mathbb{C}^{m\times (d+1)}: \textbf{x}\in\mathcal{W}_{\textbf{A,b}}\big\}, \\
       & \mathcal{H}^{'}_{\textbf{x}}(m) = \big\{[\textbf{A},\textbf{b}]\in\mathbb{C}^{m\times (d+1)}: \textbf{x}\in\mathcal{H}_{\textbf{A,b}}\big\}.
    \end{aligned}
\end{equation}
By definition, $\mathcal{W}^{'}_{\textbf{x}}(m)\cap \mathcal{H}^{'}_{\textbf{x}}(m)$ can be interpreted as the measurement matrix that can reconstruct $\textbf{x}$ from purely phase-only measurements. 

\begin{lem}
    \label{lemma11}
    Consider a fixed signal $\x\in\mathbb{C}^d$, then $\mathcal{W}^{'}_{\textbf{x}}(m)\cap \mathcal{H}^{'}_{\textbf{x}}(m)$ contains a generic $[\textbf{A},\textbf{b}]$ in $\mathbb{C}^{m\times (d+1)}$ when $m\geq 2d$.
\end{lem}

\begin{proof}
By (\ref{C.12}) and some simple arguments, we have 
\begin{equation}
    \begin{aligned}
       \label{C.13}
     & [\textbf{A},\textbf{b}]\in \mathcal{W}_{\textbf{x}}^{'}(m)\cap \mathcal{H}_{\textbf{x}}^{'}(m)\iff \textbf{x} \in \mathcal{W}_{\textbf{A,b}} \cap \mathcal{H}_{\textbf{A,b}}\iff \bm{0}\in\mathcal{W}_{\textbf{A,Ax+b}}\cap \mathcal{H}_{\textbf{A,Ax+b}}\\
      &\iff  [\textbf{A},\textbf{b}]\begin{bmatrix}
          \bm{I_d} & \textbf{x} \\ \bm{0} & 1
     \end{bmatrix} \in \mathcal{W}_{\textbf{0}}^{'}(m)\cap \mathcal{H}_{\textbf{0}}^{'}(m) \iff  [\textbf{A},\textbf{b}]  \in \Big(\mathcal{W}_{\textbf{0}}^{'}(m)\cap \mathcal{H}_{\textbf{0}}^{'}(m)\Big) \begin{bmatrix}
          \bm{I_d} & \textbf{-x} \\ \bm{0} & 1
     \end{bmatrix} .
    \end{aligned}
\end{equation}
Hence, we only need to consider $\textbf{x}=\bm{0}$. Then by Lemma \ref{lemma7}, for $[\textbf{A},\textbf{b}]\in\mathbb{C}^{m\times (d+1)}$ 
$$ [\textbf{A},\textbf{b}]\in\mathcal{W}_{\textbf{0}}^{'}(m)\cap \mathcal{H}_{\textbf{0}}^{'}(m) \iff \bm{0}\in\mathcal{W}_{\textbf{A,b}}\cap\mathcal{H}_{\textbf{A,b}} \iff \rank\big(\mathcal{D}_{\textbf{A,b}}(\textbf{0})\big)\geq 2d+m.$$
Moreover, by (\ref{7.3}) one can see entries of $\mathcal{D}_{\textbf{A,b}}(\textbf{0})$ are polynomials of the real variables $\Re(\textbf{A}),\Im(\textbf{A})$, so  Lemma \ref{lemma4} delivers that $\mathcal{W}_{\textbf{0}}^{'}(m)\cap \mathcal{H}_{\textbf{0}}^{'}(m)$ is Zariski open set of $\mathbb{C}^{m\times (d+1)}$. It remains to   find one $[\bm{\mathrm{A_0}},\bm{\mathrm{b_0}}]\in\mathbb{C}^{m\times (d+1)}$ such that $\bm{0}\in \mathcal{W}_{\bm{\mathrm{A_0,b_0}}}\cap\mathcal{H}_{\bm{\mathrm{A_0,b_0}}}$ when $m\geq 2d$, and let us consider 
\begin{equation}
    \nonumber
    \big[\bm{\mathrm{A_0}},\bm{\mathrm{b_0}}\big] = \begin{bmatrix}
         \bm{\mathrm{I}}_d & \bm{1}_{d\times 1} \\
         \textbf{i}\bm{\mathrm{I}}_d & \bm{1}_{d\times 1} \\
         \bm{0}_{(m-2d)\times d} & \bm{1}_{(m-2d)\times 1}
    \end{bmatrix}.
\end{equation}
Obviously, $\sign \big(\bm{\mathrm{A_0 \cdot0 + b_0 }}\big) = \bm{1}_{2m \times 1}$. Assuming $\sign\big(\bm{\mathrm{A_0 y + b_0}}\big)= \bm{1}_{2m\times 1}$ for some $\textbf{y} = [y_k]\in \mathbb{C}^d$, for each $k\in [d]$ we have $\sign(y_k+1)=\sign(\textbf{i}y_k+ 1)=1$. This directly implies $y_k+1\in\mathbb{R}$ and $\textbf{i}  y_k + 1 \in\mathbb{R}$, which can further lead to $y_k = 0$. Hence $\textbf{y}=\bm{0}$ and $\bm{0}\in \mathcal{W}_{\bm{\mathrm{A_0,b_0}}}$. On the other hand, evidently we have $\bm{0} \in \mathcal{H}_{\bm{\mathrm{A_0,b_0}}}$, so the proof can be concluded. 
\end{proof}


\noindent
{\it Proof of Theorem \ref{T14}:} 
Based on several previous lemmas, our strategy is parallel to the proof of Theorem \ref{T8}.
Note that $m\leq 2d-1$ leads to
$$ \rank\big(\mathcal{D}_{\textbf{A,b}}(\textbf{x})\big) \leq 2m <2d+m,
$$hence Lemma \ref{lemma7} gives $\mathcal{W}_{\textbf{A,b}}\cap \mathcal{H}_{\textbf{A,b}}=\varnothing$, or equivalently $$\mathcal{W}_{\textbf{A,b}}\subset (\mathcal{H}_{\textbf{A,b}})^c = \big\{\textbf{x}\in\mathbb{C}^d:\text{for some }j\in [m], \bm{\mathrm{\gamma_j^\top x}}+b_j \neq 0 \big\}.$$ Thus, 
$\mathcal{W}_{\textbf{A,b}}$ is nowhere dense (under Euclidean topology) and of zero Lebesgue measure. When $m\geq 2d$, we consider the following set of $[\textbf{A},\textbf{b}]\in\mathbb{C}^{m\times (d+1)}$ satisfying     property (a), (b):

\begin{equation}
    \label{C.14}
    \Xi =\left\{[\textbf{A},\textbf{b}] : ~\begin{aligned}
        &\mathrm{(a)~}\forall \mathcal{S}\in [m],0<|\mathcal{S}|\leq d,\rank (\textbf{A}^\mathcal{S}_{[\mathcal{S}]}) = |\mathcal{S}|; \\
        &\mathrm{(b)~}[\textbf{A},\textbf{b}]\in \mathcal{W}_{\bm{0}}^{'}(m)\cap \mathcal{H}_{\bm{0}}^{'}(m) ;\\
        &\mathrm{(c)~}\forall \mathcal{S}\subset [m],0<|\mathcal{S}|<d, \bm{0}_{(d-|\mathcal{S}|)\times 1}\in \mathcal{W}_{\textbf{A}^{'}(\mathcal{S}), \textbf{b}^{'}(\mathcal{S})}\cap  \mathcal{H}_{\textbf{A}^{'}(\mathcal{S}), \textbf{b}^{'}(\mathcal{S})}
    \end{aligned}\right\}
\end{equation}
 where   $\textbf{A}^{'}(\mathcal{S}),\textbf{b}^{'}(\mathcal{S})$ in property (c) are defined in (\ref{C.9}). We will prove the result via two steps.

 \vspace{1mm}
 \noindent
 {\it Step 1.} We aim to show that $\Xi$ contains a generic measurement matrix in $\mathbb{C}^{m\times (d+1)}$. Evidently, a generic $[\textbf{A},\textbf{b}]$ satisfies property (a) in (\ref{C.14}). By Lemma \ref{lemma11}, there also exists a generic $[\textbf{A},\textbf{b}]$ satisfying property (b). Thus, it remains to show a generic measurement matrix satisfies property (c), and evidently we can only consider a fixed $\mathcal{S}\subset [m]$, $0<|\mathcal{S}|<d$ (Since the result can be extended to any possible $\mathcal{S}$ via a finite intersection). We first use Lemma \ref{lemma7}, it yields 
 \begin{equation}
     \label{C.15}
     \begin{aligned}
       & \bm{0}_{(d-|\mathcal{S}|)\times 1}\in \mathcal{W}_{\textbf{A}^{'}(\mathcal{S}), \textbf{b}^{'}(\mathcal{S})}\cap  \mathcal{H}_{\textbf{A}^{'}(\mathcal{S}), \textbf{b}^{'}(\mathcal{S})}\\  \iff& \rank \Big(\mathcal{D}_{\textbf{A}^{'}(\mathcal{S}), \textbf{b}^{'}(\mathcal{S})}\big(\bm{0}_{(d-|\mathcal{S}|)\times 1}\big)\Big) \geq 2(d-|\mathcal{S}|) + (m-|\mathcal{S}|).
     \end{aligned}
 \end{equation}
 Recall (\ref{C.9}) and the definition of discriminant matrix in (\ref{7.3}), one can see the entries of $\mathcal{D}_{\textbf{A}^{'}(\mathcal{S}), \textbf{b}^{'}(\mathcal{S})}\big(\bm{0}_{(d-|\mathcal{S}|)\times 1}\big)$ are in the form of $\frac{f_{ij}(\textbf{A},\textbf{b})}{g_{ij}(\textbf{A},\textbf{b})}$ where $f_{ij}$, $g_{ij}$ are polynomials of the real variables $\Re(\textbf{A})$, $\Im(\textbf{A})$, $\Re(\textbf{b})$, $\Im(\textbf{b})$ with complex coefficients. Thus, Lemma \ref{lemma4} delivers that for a fixed $\mathcal{S}$, the set of $[\textbf{A},\textbf{b}]$ satisfying the second line (and hence also the first line) of (\ref{C.15}) is Zariski open. To show it contains a generic point of $\mathbb{C}^{m\times (d+1)}$, we still need to show it is non-empty. To raise an example, we consider $[\textbf{A},\textbf{b}]$ with $\textbf{A}$ satisfying $\textbf{A}^{\mathcal{S}^c}_{[\mathcal{S}]} = \bm{0}$, note that $m\geq 2d$ implies $m-|\mathcal{S}|\geq 2(d-|\mathcal{S}|)$, we can further let 
 \begin{equation}
     \nonumber
     \mathbb{C}^{(m-|\mathcal{S}|)\times (d-|\mathcal{S}|+1)}\ni \big[\textbf{A}^{\mathcal{S}^c}_{[d]\setminus [\mathcal{S}]}, \textbf{b}^{\mathcal{S}^c}\big]   = \begin{bmatrix}
          \textbf{I}_{d-|\mathcal{S}|} & \bm{1}_{(d-|\mathcal{S}|)\times 1}\\
          \textbf{i}\textbf{I}_{d-|\mathcal{S}|} & \bm{1}_{(d-|\mathcal{S}|)\times 1}\\
         \bm{0} & \bm{1} 
     \end{bmatrix}.
 \end{equation}By (\ref{C.9}) we have $\big[\textbf{A}^{'}(\mathcal{S}),\textbf{b}^{'}(\mathcal{S})\big]= \big[\textbf{A}^{\mathcal{S}^c}_{[d]\setminus [\mathcal{S}]}, \textbf{b}^{\mathcal{S}^c}\big]$, then combining with the proof of Lemma \ref{lemma11},   these $[\textbf{A},\textbf{b}]$ (that we consider) satisfy the first line of (\ref{C.15}). Thus, a generic   $[\textbf{A},\textbf{b}]$ satisfies the first line of (\ref{C.15}), and {\it Step 1} can be concluded. 
 
 \vspace{1mm}
 \noindent
 {\it Step 2.}  This step focuses on showing the elements of $\Xi$ satisfy (\ref{7.11}), which can then directly yield $\bm{\mathrm{m^{'}_{ae}}}(d) = 2d$. For this purpose, we only need to consider a fixed $[\textbf{A},\textbf{b}]\in \Xi$. To show (\ref{7.11}), we discuss the following cases according to $\mathcal{S}$ with $\ker (\textbf{A}^\mathcal{S},\textbf{b}^\mathcal{S})$.
 
 \vspace{1mm}
  \noindent{\it Case  1.} If $|\mathcal{S}|\geq d$, then by (a) in (\ref{C.14}) it is immediate that $\rank(\textbf{A}^\mathcal{S})=d$, which gives $\ker (\textbf{A}^\mathcal{S},\textbf{b}^\mathcal{S})$ only contains a single point. Moreover, it is easy to confirm this point belongs to $\mathcal{W}_{\textbf{A,b}}$, and hence (\ref{7.11}) is true trivially.

 \vspace{1mm}
 \noindent{\it Case  2.} If $\mathcal{S} = \varnothing$, then (\ref{7.11}) states that $\mathcal{W}_{\textbf{A,b}}$ contains a generic point of $\mathbb{C}^d$. Our idea is similar to the proof of Lemma \ref{lemma11}, and the only difference is that the measurement matrix is fixed now, while the signal $\textbf{x}$ will be viewed as variable. First, Lemma \ref{lemma7} gives \begin{equation}
     \label{C.16}
     \mathcal{W}_{\textbf{A,b}}\cap \mathcal{H}_{\textbf{A,b}} = \big\{ \textbf{x}:\rank \big( \mathcal{D}_{\textbf{A,b}}(\textbf{x})\big)  \geq 2d+m\big\}.
 \end{equation}
 Observing (\ref{7.3}), one can see entries of $\mathcal{D}_{\textbf{A,b}}(\textbf{x})$ are polynomials of the real variables $\Re(\x),\Im(\x)$ (with degree at most 1). Thus, Lemma \ref{lemma4} delivers that $\mathcal{W}_{\textbf{A,b}}\cap \mathcal{H}_{\textbf{A,b}}$ is Zariski open. Moreover, it is non-empty due to property (b) in (\ref{C.14}), hence the desired result is displayed.

 \vspace{1mm}
 \noindent{\it Case 3.} If $0<|\mathcal{S}|<d$, by exactly the same argument in {\it Case 2}, one can prove $\mathcal{W}_{\textbf{A}^{'}(\mathcal{S}), \textbf{b}^{'}(\mathcal{S})}\cap  \mathcal{H}_{\textbf{A}^{'}(\mathcal{S}), \textbf{b}^{'}(\mathcal{S})}$ is Zariski open. Combining with (c) in (\ref{C.14}), $\mathcal{W}_{\textbf{A}^{'}(\mathcal{S}), \textbf{b}^{'}(\mathcal{S})}\cap  \mathcal{H}_{\textbf{A}^{'}(\mathcal{S}), \textbf{b}^{'}(\mathcal{S})}$ contains a generic point of $\mathbb{C}^{d-|\mathcal{S}|}$. We now invoke Lemma \ref{lemma8} to show (\ref{7.11}). Recall $\rank (\textbf{A}^\mathcal{S}_{[\mathcal{S}]})$, under the condition $\textbf{x}\in \ker(\textbf{A}^\mathcal{S},\textbf{b}^{\mathcal{S}})$, Lemma \ref{lemma8} gives 
 $$\x\in \mathcal{W}_{\textbf{A,b}}\cap \mathcal{H}_{\textbf{A,b}}(	\mathcal{S})\iff\textbf{x}^{[d]\setminus[\mathcal{S}]}\in \mathcal{W}_{\textbf{A}^{'}(\mathcal{S}),\textbf{b}^{'}(\mathcal{S})}\cap \mathcal{H}_{\textbf{A}^{'}(\mathcal{S}),\textbf{b}^{'}(	\mathcal{S})}.$$
 Moreover, the above $\textbf{x}$ and $\textbf{x}^{[d]\setminus [\mathcal{S}]}$ are connected by a linear isomorphism between $\ker(\textbf{A}^\mathcal{S},\textbf{b}^\mathcal{S})$ and $\mathbb{C}^{d-|\mathcal{S}|}$ (See the proof of Lemma \ref{lemma8}, especially (\ref{C.10})). Therefore, it yields that $\mathcal{W}_{\textbf{A,b}}\cap \mathcal{H}_{\textbf{A,b}}(\mathcal{S})$ contains a generic point of $\ker(\textbf{A}^\mathcal{S},\textbf{b}^\mathcal{S})$, which can imply (\ref{7.11}). The proof is complete. \hfill $\square$

 \vspace{2mm}
 \noindent
{\it Proof of Theorem \ref{T16}:} 
With no loss of generality, we consider the canonical measurement matrix $$[\textbf{A},\textbf{b}]=\begin{bmatrix}
	\textbf{I}_d & \bm{0}\\ \bm{\mathrm{A_1}} &\bm{\mathrm{b_1}}
\end{bmatrix}$$ 
for some $[\bm{\mathrm{A_1}},\bm{\mathrm{b_1}}]\in \mathbb{C}^{(m-d)\times (d+1)}$. We plan to use the discriminant matrix $\mathcal{E}_{\textbf{A,b}}(\x)$ to yield the result. Recall (\ref{7.4}), (\ref{7.5}) and (\ref{7.6}), specifically 
$$\mathcal{E}_{\textbf{A,b}}(\x)=\begin{bmatrix}
[\Psi_{d+1}^{'}(\x)]_{\mathrm{N}(\x)}\\ 
\vdots\\
[\Psi_{m}^{'}(\x)]_{\mathrm{N}(\x)}
\end{bmatrix}.$$
Due to $\x\in \mathcal{W}_{\textbf{A,b}}$, Theorem \ref{T10} gives $\rank(\mathcal{E}_{\textbf{A,b}}(\x))=|\mathrm{N}(\x)|$. Thus, there exists $\mathcal{S}_0\subset [m-d]$, such that $|\mathcal{S}_0|=|\mathrm{N}(\x)|$, $\rank([\mathcal{E}_{\textbf{A,b}}(\x)]^{\mathcal{S}_0})=|\mathrm{N}(\x)|$. Now we consider  $$\mathcal{J}=\{j\in[m]\setminus [d]: \mathrm{at}\ \mathrm{least}\ 1\ \mathrm{row~of}\ \big[\Psi^{'}_{j}(\x)\big]_{\mathrm{N}(\x)}~\mathrm{appears}\ \mathrm{in}\ (\mathcal{E}_{\textbf{A,b}}(\x))^{\mathcal{S}_0}\}.$$ Then evidently, $|\mathcal{J}| \leq |\mathcal{S}_0| = |\mathrm{N}(\x)|$. Moreover, we can consider the submatrix of $[\textbf{A},\textbf{b}]$ given by $[\textbf{A}^{[d]\cup\mathcal{J}},\textbf{b}^{[d]\cup\mathcal{J}}]$, then we have the rank relation as 
\begin{equation}
    \begin{aligned}
      |\mathrm{N}(\x)| = \rank\big(\mathcal{E}_{\textbf{A,b}}(\x)\big)\geq\rank\big(\mathcal{E}_{\textbf{A}^{[d]\cup\mathcal{J}},\textbf{b}^{[d]\cup\mathcal{J}}}(\x)\big) \geq \rank \big((\mathcal{E}_{\textbf{A,b}}(\x))^{\mathcal{S}_0}\big)=|\mathrm{N}(\x)|,
      \nonumber
    \end{aligned}
\end{equation}
which leads to $\rank(\mathcal{E}_{\textbf{A}^{[d]\cup \mathcal{J}},\textbf{b}^{[d]\cup \mathcal{J}}})= |\mathrm{N}(\x)|$. We use Theorem \ref{T10} again, then we obtain $\x\in \mathcal{W}_{\textbf{A}^{[d]\cup \mathcal{J}},\textbf{b}^{[d]\cup \mathcal{J}}}$. Since $|[d]\cup \mathcal{S}_0|\leq d+|\mathcal{S}_0|\leq 2d$, we can find $\mathcal{S}\subset [m]$, $|\mathcal{S}|=2d$ such that $\x\in \mathcal{W}_{\textbf{A}^\mathcal{S},\textbf{b}^\mathcal{S}}$. For the latter part, from Theorem \ref{T14}, there exists $[\textbf{A},\textbf{b}]\in \mathbb{C}^{2d\times d}$, such that $\mathcal{W}_{\textbf{A,b}}$ contains a generic point, but $\mathcal{W}_{\textbf{A}^\mathcal{S},\textbf{b}^\mathcal{S}}$ where $|\mathcal{S}|=2d-1$ is nowhere dense (under Euclidean topology) and of zero Lebesgue measure. Thus, there exists some $\hat{\textbf{x}}$ such that $\hat{\textbf{x}}\in\mathcal{W}_{\textbf{A,b}}$ but $\hat{\textbf{x}}\notin \mathcal{W}_{\textbf{A}^\mathcal{S},\textbf{b}^\mathcal{S}}$ for all $\mathcal{S}$ with $|\mathcal{S}|=2d-1$. \hfill $\square$
\end{appendix}	
  
\end{document}